%% file: arxiv-main.tex
\documentclass[msom]{informs3}

\RequirePackage{tgtermes}
\RequirePackage{newtxtext}
\RequirePackage{newtxmath}
\RequirePackage{endnotes}

\OneAndAHalfSpacedXI 

\usepackage{tikz}

\usepackage{natbib}
 \bibpunct[, ]{(}{)}{,}{a}{}{,}%
 \def\bibfont{\small}%
\input{preamble}

\EquationsNumberedThrough    

\TheoremsNumberedThrough     
\ECRepeatTheorems  %

\begin{document}

	\RUNAUTHOR{Banerjee et al.} %

	\RUNTITLE{Price Competition Under Platform-Mediated Search}

\TITLE{{Price Competition Under Platform-Mediated Search:\\ A Consider-Then-Choose Framework}}

\ARTICLEAUTHORS{%
\AUTHOR{Siddhartha Banerjee}
\AFF{Cornell University, Ithaca, NY\\
     \EMAIL{sbanerjee@cornell.edu}}

\AUTHOR{Chamsi Hssaine}
\AFF{University of Southern California, Marshall School of Business,
     Los Angeles, CA\\
     \EMAIL{hssaine@usc.edu}}

\AUTHOR{Vijay Kamble}
\AFF{University of Illinois at Chicago, Chicago, IL\\
     \EMAIL{kamble@uic.edu}}
}	 


\input{final-parts/abstract}




\KEYWORDS{Price competition,  consider-then-choose models, lexicographic preferences, algorithmic pricing, platform design}

	
	%
	
\maketitle
	
\input{final-parts/intro-chamsi-v2}

\input{final-parts/preliminaries-msom24.tex}


\input{final-parts/equilibrium-analysis-vijay.tex}

\input{final-parts/algorithmic-competition-vijay.tex}

\input{final-parts/case-study}

\input{final-parts/discussion.tex}

\begingroup \parindent 0pt \parskip 0.0ex \def\enotesize{\normalsize} \theendnotes \endgroup

\ACKNOWLEDGMENT{The authors gratefully acknowledge the department editor, associate editor, and anonymous referees whose comments and suggestions were tremendously helpful in improving the paper.}

\bibliographystyle{informs2014} 
\SingleSpacedXI
\bibliography{chamsi-refs} 

\newpage

\newpage

\AtBeginEnvironment{APPENDICES}{%
  \renewcommand{\theHsection}{appendix.\Alph{section}}%
  \renewcommand{\theHsubsection}{\theHsection.\arabic{subsection}}%
  \renewcommand{\theHsubsubsection}{\theHsubsection.\arabic{subsubsection}}%
  \renewcommand{\theHfigure}{\theHsection.\arabic{figure}}%
  \renewcommand{\theHtable}{\theHsection.\arabic{table}}%
  \renewcommand{\theHequation}{\theHsection.\arabic{equation}}%
  \renewcommand{\theHtheorem}{\theHsection.\arabic{theorem}}%

  \crefalias{section}{appendix}%
  \crefalias{subsection}{appendix}%
  \crefalias{subsubsection}{appendix}%
}

\crefalias{section}{appendix}
\begin{APPENDICES}
\input{final-parts/msom24-appendix}

\end{APPENDICES}
		
		
\end{document}

%% file: preamble.tex
\usepackage{amsmath,amsfonts,cancel}
\usepackage{thmtools}
\usepackage{adjustbox}
\usepackage[shortlabels]{enumitem}
\usepackage[table]{xcolor}
\definecolor{edits}{rgb}{0,0,0}
\usepackage{mathtools}
\usepackage{dsfont}
\usepackage{tcolorbox}
\usepackage{booktabs}
\usepackage{xspace}
\usepackage{subcaption}
\usepackage{tensor}
\usepackage{tikz}
\usetikzlibrary{patterns}
\usepackage{natbib}

\usepackage{placeins}
\usepackage[colorlinks=true,breaklinks=true,bookmarks=true,urlcolor=magenta,citecolor=magenta,linkcolor=magenta,bookmarksopen=false,draft=false]{hyperref}
\usepackage[capitalize]{cleveref}

\usepackage{crossreftools}
\hypersetup{
  hypertexnames=false,
  bookmarksnumbered=true
}

\usepackage{algorithm}
\usepackage[noend]{algpseudocode}

\usepackage{color-edits}

\usepackage{changepage}

\crefname{condition}{Condition}{Conditions}

\newcommand{\oma}{\bar{\mathcal{A}}}
\newcommand{\pmax}{\bar{v}}

\newcommand{\cS}{\mathcal{S}}
\newcommand{\cU}{\mathcal{U}}
\newcommand{\cN}{\mathcal{N}}
\newcommand{\cV}{\mathcal{V}}

\newcommand{\attributevalues}[1]{\mathcal{V}^{#1}}
\newcommand{\attributecriteria}[1]{\mathcal{Y}_{#1}}

\newcommand{\bp}{\mathbf{p}}
\newcommand{\bq}{\mathbf{q}}

\newcommand{\pvec}{\mathbf{p}}

\newcommand{\nearestminprice}{p^{\min}(\cS)}

\newcommand{\nopriceconsiderationset}{\overline{\mathcal{N}}_c}
\newcommand{\npaplus}[2]{\overline{\lambda}_{{#1}\succ{#2}}}
\newcommand{\equalpref}[2]{\overline{\lambda}_{{#1}={#2}}}
\newcommand{\Prob}[1]{\mathbb{P}\left(#1\right)}

\usepackage{multirow}

\crefname{assumption}{assumption}{assumptions}
\crefname{algocf}{algorithm}{algorithm}

\mathchardef\mhyphen="2D 
\ifdefined \argmin
\else \DeclareMathOperator*{\argmin}{arg\,min}
\fi
\ifdefined \argmax
\else \DeclareMathOperator*{\argmax}{arg\,max}
\fi

\let\originalleft\left
\let\originalright\right
\renewcommand{\left}{\mathopen{}\mathclose\bgroup\originalleft}
\renewcommand{\right}{\aftergroup\egroup\originalright}

\usepackage{natbib}
 \bibpunct[, ]{(}{)}{,}{a}{}{,}%
 \def\bibfont{\small}%


%% file: final-parts/abstract.tex
	\ABSTRACT{
    \textbf{Problem definition:} We study the problem of predicting price equilibria on e-commerce platforms where sellers compete across multiple attributes (e.g., price, average rating, delivery speed). In these settings, a platform's design choices --- such as its display ranking, filtering tools, and promotional badges --- critically shape customer search and purchase behavior, which in turn determine sellers' equilibrium pricing strategies. Our goal is to develop a tractable framework that allows a platform to anticipate the market impact of its design interventions. \textbf{Methodology/results:} We consider a behavioral model --- Consider-then-Choose with Lexicographic Choice (CLC) --- specifically tailored to platform-mediated search. We establish that any local Nash equilibrium admits a sequential-move characterization; this yields a tractable procedure for computation under an interpretable sufficient condition, which we term gradient dominance. We further prove that under gradient dominance, simple, decentralized gradient-based algorithms converge to an equilibrium, providing platforms with a method for simulating market outcomes. {Finally, we use our framework to study how platform design affects market outcomes.} \textbf{Managerial implications:} Our framework applies to any platform in which sellers compete on multiple attributes and customer choice is guided by the platform's interface. In these environments, sellers' pricing strategies must be understood not in isolation, but as a response to the platform's design. Our work provides platform operators with a rigorous toolbox to efficiently evaluate how changes to interface design, information disclosure, and ranking policies can affect competitive outcomes.
 }

%% file: final-parts/intro-chamsi-v2.tex

\section{Introduction}\label{sec:intro}

The rise of e-commerce platforms has introduced {fundamental} changes to the environment in which sellers price their goods and services. In particular, the supply diversity available on these platforms has shaped customer behavior in two important, related ways:
\begin{enumerate}
\item Customers are now faced with the burden of choice, with a plethora of seller offerings differing across a range of attributes such as price, {average rating, delivery speed, and brand}.
\item {Limited in time and information, customers resort to ``fast and frugal'' decision-making heuristics, relying on platforms' sorting and filtering functionalities to search for offerings. For instance, customers may filter their search results to items that are within their budget and whose average rating exceeds a minimum threshold; default search rankings (e.g., ``Price: Low to High'', ``Avg. Customer Review'', ``Featured'') critically shape the order in which customers evaluate offerings, with higher-ranked items viewed with higher probability \citep{granka2004eye,joachims2005accurately,gao2022joint}.}
\end{enumerate}

{Given this state of affairs, platforms' design choices --- such as whether to disclose granular rating breakdowns, highlight an item with a ``{Top/Overall Pick}'' badge, or privilege a particular attribute in the default ranking --- play a pivotal role in determining {which} offerings are considered by customers, and how competing sellers adjust their prices in response. In turn, sellers' equilibrium prices influence key platform metrics such as aggregate revenue and market penetration (i.e., fraction of demand served). Hence, anticipating how proposed interventions impact equilibrium prices is critical to a platform's ability to choose amongst competing design choices. 

Motivated by this practical need, this work asks the following research question:

\begin{center}
{\em 
How can a platform efficiently evaluate the impact of its design interventions --- such as information disclosure, and sorting and filtering functionalities --- on customer behavior and the resulting equilibrium pricing outcomes?
}
\end{center}

}

\subsection{Main contributions}

\textbf{A practical model for platform-mediated customer choice.} 
Toward answering this question, we consider a model of purchase behavior that is specifically tailored to capture customers' responses to platforms' filtering and sorting capabilities, and that generalizes lexicographic choice, a decision-making heuristic for multi-attribute choice that has been extensively studied and empirically validated in the behavioral economics and marketing literature \citep{alexis1968consumer,fishburn1974exceptional,parkinson1979information, kohli2007representation,jedidi2008inferring,kohli2019randomized}. In this model --- which we call {\it Consider-then-Choose with Lexicographic Choice} (CLC) --- customer decision-making is composed of two phases. The first phase, which models the use of platforms' filtering capabilities, involves the formation of a consideration set of options. In this phase, customers narrow down the set of feasible options by specifying (i) desirable non-price features (e.g., brand, rating range, size and specifications, display position, etc.) and (ii) an upper bound on the price, reflecting their budget or maximum willingness-to-pay. The second phase involves the selection from this set according to the lexicographic choice heuristic. Specifically, customers have heterogeneous preference orderings over item attributes (e.g., price, average rating, number of reviews), capturing which attributes they consider more important than others, and a (weak) ranking of items along each attribute. They first rank items with respect to their most preferred attribute, and choose the most desirable product according to that attribute. If they are indifferent amongst multiple products within this attribute-specific ranking, they evaluate the topmost products according to their second-most preferred attribute, choosing the top product for that attribute, and so on. 

{Empirical evidence of lexicographic choice behavior has been documented across multiple contexts, including electronics purchases \citep{kohli2007representation,jedidi2008inferring,kohli2019randomized,yee2007greedoid}, and in limited-differentiation settings such as grocery shopping \citep{parkinson1979information}. CLC behavior may also arise when customers impose strict requirements on non-price attributes, making price the primary differentiator among qualifying alternatives. For example, a grocery shopper seeking organic certification in yogurt may simply choose the lowest-priced organic option. Likewise, a television buyer specifying a 65-inch OLED display may then select the least expensive model among premium brands such as Samsung, LG, and Sony. Another common anecdotal example of CLC behavior is ``{\it Best-I-Can-Afford}'' (BICA) behavior, in which a customer has a fixed budget, and items are characterized by price and quality (indicated, e.g., by customer rating). Customers only include items priced below their budget in their consideration set, and subsequently choose the highest-quality item, breaking ties according to price. This pattern is commonly observed among quality-sensitive buyers and delegated purchasers: for instance, business travelers constrained by reimbursement policies who choose the highest-quality accommodation within allowable limits. We refer the reader to Appendix \ref{ex:attributes} for additional examples of behaviors encompassed by CLC choice, anchored within the context of e-commerce.}

{\textbf{A toolbox for tractably predicting equilibrium outcomes.} 
A natural concern with the CLC choice model is that its generality may undermine the tractability of equilibrium computation, and thus a platform’s ability to evaluate the impact of its design interventions. Indeed, computing Nash equilibria is notoriously hard in general, even in discrete two-player games \citep{daskalakis2009complexity}.} {We however establish that, under CLC choice, any {\it local} Nash equilibrium (LNE) --- a generalization of pure-strategy Nash equilibria well-suited to algorithmic pricing environments such as modern e-commerce platforms \citep{ratliff2016characterization} --- admits a {sequential-move} representation. Specifically, any LNE can be viewed as an ordering of sellers and a corresponding sequence of locally optimal undercutting prices (\Cref{thm:general-LNE}). One of the main drivers of this result is a key unifying property --- the {\it pseudo-competitive property} --- which we show holds for {\it any} CLC model. Under this property, the revenue of any seller, given other sellers' prices, is invariant to changes to any strictly lower competing price, as long as this price remains strictly lower (\Cref{prop:clc-choice-implies-psg}).

This structural characterization yields a natural sequential procedure to compute LNEs. Under mild regularity conditions, the procedure we derive exhaustively identifies all LNEs that exist, or concludes that none do. Moreover, we provide an interpretable condition, which we term {\it gradient dominance}, under which a LNE is guaranteed to exist (\Cref{prop:gd-implies-lne-existence-uniqueness}), and can be computed by our procedure in $O(N^2)$ iterations, where $N$ is the number of sellers. Gradient dominance captures markets with a hierarchical competitive structure, where certain sellers hold advantages due to non-price attributes such as product quality. Importantly, this condition holds across a range of CLC models, including BICA.

{One drawback of the above procedure is that it requires the platform to estimate sellers' {\it entire} respective revenue curves. In complex settings --- for example, when the number of sellers is large or when customers exhibit diverse CLC choice behaviors --- doing so may be challenging.} A more practical alternative in such settings is to rely on {\it agent-based simulation}. Specifically, a platform can implement a testbed to simulate customer behavior and approximate equilibria by simulating sellers' price updates under various design interventions. Such simulations are far more tractable than directly estimating sellers’ aggregate revenue functions, as CLC models of customer behavior can often be learned directly from available user-level data (e.g., purchase histories) \citep{kohli2007representation,kohli2019randomized}. However, the validity of this simulation approach hinges on whether the price update dynamics converge to a LNE. Our main contribution is to formalize this link: we prove that simple, distributed gradient-based price update dynamics converge to a LNE under gradient dominance (\cref{thm:ogd-converges-to-LNE-n-firms}). Numerically, we further observe that such dynamics often converge to a LNE even when gradient dominance does not hold. 
The crucial implication of this fact is that agent-based simulations assuming such dynamics provide a reliable and computationally efficient means of predicting price equilibria, and therefore the outcomes of design interventions.
}

{
{\textbf{Implications for platform design.}} {We conclude by leveraging our framework to study how platform design impacts market outcomes.} We first consider a duopoly in which items are defined by price and average rating. We establish that disclosing sellers' average ratings softens price competition, with both sellers (and therefore the platform) achieving higher revenue in equilibrium (\Cref{prop:metrics-special-case-q1}). We also prove that a default ranking that sorts sellers in decreasing order of average rating maximizes both platform revenue and market penetration, compared to the default ranking that sorts sellers in increasing order of price (\Cref{prop:ranking}). We build on these results to study the impact of disclosing a third differentiating attribute, e.g., delivery speed, popularity, or a platform certification badge such as ``Overall Pick.'' In contrast to the above results, the impact of disclosing a third attribute is more nuanced. In disclosing this additional attribute, a second equilibrium, in which the lower-rated seller prices above the higher-rated seller, arises (\Cref{prop:special-case-LNE-q2}). Moreover, while the platform's revenue {decreases} in the equilibrium where the higher-rated seller prices above the lower-rated seller, it may increase in this second equilibrium, if the third attribute is important enough to customers (\Cref{prop:metrics-special-case-q2}).

{Finally, we extend our analysis to a more complex market setting using our simulation-based framework.} Our results confirm that platform design choices that encourage customers to be more ratings-sensitive, such as displaying seller ratings, implementing a ``Ratings-First'' default ranking, and providing finer ratings granularity, increase total platform revenue by mitigating price competition. Furthermore, our numerical analysis uncovers the critical dependence of platform revenue on inter-attribute correlations. While an attribute that is positively correlated with seller ratings can further boost revenue, an independent or negatively correlated attribute can significantly decrease it by intensifying price competition. This underscores a key managerial insight: disclosing more information is not unequivocally beneficial, and must be carefully evaluated based on interactions with existing displayed attributes.
}

{
Overall, our work offers a tractable framework for platforms to predict equilibrium outcomes of design interventions in complex, multi-attribute environments characterized by diverse customer behavior.
}

{
\subsection{Paper organization} We survey relevant literature in Section~\ref{sec:related-work}. In~\cref{sec:preliminaries} we present the CLC choice model, and describe the structure and efficient computation of equilibria in \cref{sec:pseudo-equilibria}. In  \Cref{sec:numerical-results} we analyze the dynamics of gradient-based pricing algorithms, thereby lending theoretical support for our simulation framework. {Finally, we use the developed tools to derive insights into commonly encountered platform design decisions in \Cref{sec:case-study}.}
}
Unless specified, the proofs of all results are included in the appendix. 

\input{final-parts/related-work-msom24.tex}

%% file: final-parts/related-work-msom24.tex

\section{Related Work}\label{sec:related-work}
Our work relates to a wide range of customer choice models studied in the literature and their impact on competitive outcomes. We discuss the relevant literature below.

\textbf{Lexicographic preference models.} As noted above, lexicographic preferences have been extensively studied within the behavioral economics and marketing literature; see \citet{fishburn1974exceptional} for an excellent survey of earlier work. {
Lexicographic choice is a special case of bounded rationality, first articulated by Herbert Simon in his Nobel Prize-winning work \citep{simon1955behavioral,simon_rational_1956}. It is a noncompensatory decision rule in which individuals avoid fine-grained trade-offs among alternatives, partly due to the negative emotions and cognitive load that such deliberation entails \citep{luce1997choice}. Empirical research supports this behavior: \citet{payne1976task} found that while utility-maximizing rules were more common in two-alternative settings, as the number of options increased, subjects increasingly relied on heuristics to eliminate alternatives. Similarly, \citet{hogarth2005simple} showed that simple trade-off–avoiding strategies can approximate the choices generated by linear multi-attribute models, especially for binary-valued attributes. More recently, \citet{jagabathula2019limit} analyzed transactions of consumer packaged goods across 47 U.S. grocery and drugstores, and found that random utility maximization models failed to achieve acceptable predictive performance for products such as coffee, yogurt, and milk; they suggest that going beyond rational choice models is necessary to obtain acceptable predictive performance for these goods.

Of the many noncompensatory rules proposed in the literature, we focus on the lexicographic model because of its close alignment with the sorting functionalities commonly used on e-commerce platforms. Empirical support for lexicographic choice dates back to \citet{slovic1975choice}, who conducted a laboratory experiment where participants chose between two gift packages of cash and coupons they had previously deemed equal in value. When asked which attribute --- cash or coupons --- they considered more important, 88\% of participants selected the option superior on that attribute (typically cash), rather than splitting their choices evenly as a utility maximization model would predict.

In retail contexts, lexicographic choice rules have been shown to perform comparably to standard utility maximization models in predicting consumer purchases. For example, \citet{alexis1968consumer} found that women’s clothing purchases often followed a lexicographic order over attributes such as type, size, color and style, comfort and fit, and price, with the first garment meeting all criteria being chosen. Similarly, \citet{parkinson1979information} observed such behavior in purchases of toothpaste and deodorant. More recent studies have provided consistent evidence in electronics: \citet{yee2007greedoid}, \citet{kohli2007representation}, \citet{jedidi2008inferring}, and \citet{kohli2019randomized} all document lexicographic choice. For instance, \citet{kohli2007representation} reported strong support for the heuristic when consumers evaluated laptops based on price, memory, brand, processor speed, and hard drive capacity, while \citet{kohli2019randomized} found the lexicographic rule outperformed the logit model in predicting tablet purchases across attributes including brand, screen size, storage capacity, RAM, battery life, and price. Beyond retail, \citet{harris2018smokers} provided recent evidence of smokers applying a lexicographic rule when choosing between cigarette packages. From a theoretical standpoint, \citet{kohli2007representation} also established necessary and sufficient conditions under which a linear utility function can represent a lexicographic decision rule.}

Despite its prominence in the computer science and marketing literature, the use of lexicographic models in the operations literature remains rather limited, with a few exceptions. \citet{soberman2007campaign} study how political campaign budgets affect advertising strategies, and assume that voters' decisions are represented by a lexicographic rule with two attributes, valence (i.e., inherent attractiveness of a candidate) and party position, with valence being the more important attribute. \citet{bansal_product_2009} study the product design problem of a monopolist when customers purchase the cheapest product above a quality floor (note that this is a special case of CLC choice).  \citet{stamatopoulos2019welfare} study dynamic pricing as a means of managing inventory costs when strategic consumers can time their purchases. They assume individuals exhibit lexicographic preferences over price and wait times: a customer always prefers to wait for a better price, but given the same price she prefers to purchase earlier rather than later. To the best of our knowledge, we are the first to study price competition under a general model of lexicographic choice.

{Finally, we note that our simulation framework requires the ability to infer lexicographic choice rules from data. A number of works have proposed algorithms for this task \citep{kohli2007representation,jedidi2008inferring,kohli2019randomized}.}

{\bf Consider-then-choose models.} {The CLC model generalizes the basic lexicographic choice heuristic by incorporating the consideration phase in customers' decision-making process.
This phase --- which naturally models platforms' filtering functionalities --- addresses a key weakness of the vanilla lexicographic heuristic: used alone, it may lead to the selection of an option that is only marginally better according to the customer's primary attribute, but unacceptably poor on others. By allowing customers to set acceptable ranges for different attributes in this preliminary phase, the consideration phase ensures that only satisfactory options enter the choice set; applying the lexicographic rule thereafter results in a final selection of acceptable quality along all dimensions.}

{Consider-then-choose} models have been studied in areas spanning marketing, behavioral economics, and operations management; as with the lexicographic choice model, this class of models was inspired by the finding that people often utilize simple decision-making heuristics to narrow their options \citep{howard1969theory} and lessen the cognitive burden of making multi-alternative decisions \citep{tversky1974judgment}. Numerous empirical studies have supported this notion {(see, e.g., \citet{hauser1978testing,swait1987incorporating, roberts1991development}). There also exists an extensive body of literature focused on rational agent models that justify consideration sets. In such models, screening heuristics have been shown to be rational under search costs \citep{hauser1990evaluation} and constraints of limited time and knowledge  \citep{roberts1991development}.}

\textbf{Sequential search models.} In the past two decades, a large body of work has identified the sequential nature of customer search on platforms: given a set of ranked search results, customers largely browse products one by one, top to bottom \citep{ansari2003customization,granka2004eye,joachims2005accurately}. {The well-documented tendency of customers to interact more with items that are ranked higher on a platform's display page is known as {\it position bias}. A variety of models of position bias have been proposed in the literature (see, e.g., \citet{gallego2020approximation,chen2021revenue,gao2022joint,derakhshan2022product}), with the main focus of these works being the platform's ranking problem. Similar to these models, CLC choice is able to capture the phenomenon of position bias (see \Cref{sec:preliminaries} for a discussion). However, our work lies in contrast to this literature on two fronts. First, we are interested in analyzing equilibrium outcomes of price competition between sellers on an e-commerce platform, as opposed to studying the platform's monopolistic pricing / ranking decisions. Second, the majority of models assume that the platform uses a single ranking as a lever, which controls the order in which all customers view products; CLC choice, on the other hand, models the reality that there are many ways in which customers may sort and filter search results (independently of the default ranking chosen by the platform), and explicitly models this heterogeneity.
}

Related to position bias is the study of the impact of search costs in decision-making. Search costs have been used as a way to explain price dispersion in competitive environments \citep{stigler1961economics}; see \citet{anderson2018firm} for an exhaustive survey. As pointed out by \citet{choi2018consumer}, much of the existing work confines itself to simple settings (e.g., symmetric duopoly environments, homogeneous consumers), given the complexity of incorporating endogenous search. {While we do not explicitly model search costs, the CLC model can be viewed as implicitly modeling the {outcome} of search costs, since customers do not necessarily screen all alternatives before making a final decision (e.g., by first forming a smaller consideration set of items, or by incorporating a position bias). This implicit modeling is in line with the philosophy behind the widely studied click models in the computer science and operations literatures. These latter models have similarly been described as the outcomes of search costs, despite not explicitly modeling them (see \cite{gao2022joint} for an excellent overview).}

{\bf Price competition under linear and attraction demand models.} Compared to the economics literature, our work is more closely related to studies of price competition under exogenous models of customer behavior; see \citet{vives1999oligopoly} and \citet{gallego2006price} for excellent surveys. By and large, this literature seeks to characterize conditions for existence and uniqueness of pure-strategy Nash equilibria under the assumption that customers are utility-maximizing. Seminal works include \citet{topkis1979equilibrium} and \citet{milgrom1990rationalizability}, who respectively show existence of equilibria in supermodular games and monotone transformations thereof. \citet{caplin1991aggregation} derive separate sufficient conditions on customers' utility functions for existence and uniqueness of equilibria among differentiated products.

A considerable focus in subsequent literature has been on logit-based choice models. \citet{anderson1992discrete} and \citet{bernstein2004general} established existence and uniqueness of equilibria under the MNL model with a single customer class under symmetric and asymmetric firms, respectively. \citet{gallego2006price} show existence and uniqueness in a Bertrand oligopoly for an attraction demand model with convex costs; they moreover show linear convergence to the unique equilibrium under a {\it t\^{a}tonnement} scheme. \citet{gallego2014multiproduct} and \citet{aksoy2013price} later characterize Nash equilibria under nested logit choice and mixed MNL demand under mild conditions, respectively.

{\bf  Equilibrium convergence in repeated price competition under gradient-based methods.} Our work joins a small line of papers that study convergence of simple gradient-based schemes in games of repeated price competition. Such schemes have attracted significant attention due to their simplicity and their widespread use in industry. \citet{cooper2015learning} analyze price dynamics when two sellers estimate parameters of the underlying linear demand model without explicitly incorporating competition, and identify conditions under which sellers' prices converge to  (i) the Nash equilibrium, (ii) the cooperative solution which maximizes their joint revenue, or (iii) limit points that are neither of the two.  For duopolistic price competition with linear demand and reference effects, \citet{golrezaei2020no} show that sellers running Online Mirror Descent with decreasing stepsizes converge to a stable Nash equilibrium, a solution concept tailored to reference prices. Most recently, \citet{goyal2023learning} leveraged the {\it t\^{a}tonnement} scheme of \citet{gallego2006price} to show that under unknown MNL demand, sellers running a version of stochastic online gradient descent in the space of aggregated logit parameters converge to the unique Nash equilibrium. 

For more general games, it is well-known that gradient-based methods can fail to converge \citep{mertikopoulos2018cycles}. The first positive result in this regard was that of \citet{rosen1965existence}, who showed that under a condition termed diagonal strict concavity (DSC), a Nash equilibrium exists and is unique; moreover, distributed gradient dynamics converge to the Nash outcome. More recently, \citet{bravo2018bandit} showed that under diagonal strict concavity, no-regret learning based on mirror descent almost surely converges to the unique Nash equilibrium. In the same vein, \citet{mertikopoulos2019learning} defined the notion of variationally stable equilibrium that attracts gradient-based dynamics in games that may not satisfy DSC. We note that none of these regularity conditions are assumed to hold in our model; in fact, revenue functions may even be discontinuous under CLC choice. So, while convergence to the pure-strategy Nash equilibrium is a natural desideratum when this equilibrium exists and is unique, in settings such as ours where sellers' revenue curves are ill-behaved and exhibit potentially many strict local maxima, one need not expect convergence of gradient methods to these outcomes (analogous to gradient methods only converging to local optima in monopolist settings). Our setup thus requires us to expand our solution concept to {\it local} Nash equilibria, as opposed to pure-strategy Nash equilibria.

{{\bf Impact of information design on platform outcomes.} The final contribution of our work relates to the literature on information design in operations. We refer the reader to \citet{candogan2020information} for a general survey, focusing here on papers that study how e-commerce platform design decisions affect pricing outcomes in competitive settings. Within this line of work, early empirical studies highlight the role of platforms' feedback mechanisms, showing that positive customer feedback generally results in higher prices \citep{dellarocas2003digitization}. More recently, studies have examined the impact of seller rankings on prices. For instance, using data from a major U.S. e-commerce platform, \citet{huang2026ranking} find that a change in the platform's ranking algorithm causes sellers to reduce prices to improve future rankings. In a similar vein, \citet{dinerstein2018consumer} study eBay's website redesign, which displayed seller listings head-to-head and ranked them by price, and find that prices fell significantly after the change. \citet{gallino2025algorithmic} also identify that the buybox --- which prominently presents one option as the default choice to customers, using price as one important feature --- causes sellers to lower prices in order to influence the platform's display algorithm. Overall, these empirical findings are consistent with the high-level insights we derive in \Cref{sec:case-study}: that displaying quality information raises equilibrium prices by differentiating sellers along a non-price attribute, and that the platform's ranking has a material impact on sellers' pricing decisions.

From a theoretical perspective, the platform's information design problem has largely been tackled via the Bayesian persuasion framework \citep{kamenica2011bayesian}. In these works, a privately informed platform sends quality-sensitive buyers a signal about the quality of each seller; buyers then make a utility-maximizing purchase given both sellers' posted prices, as well as their posterior belief about seller qualities. Most closely related to our findings in \Cref{sec:case-study} is \citet{armstrong2022consumer}, which establishes that, in a duopoly, the profit-optimal signal structure amplifies perceived product differentiation, thereby relaxing price competition. Our finding that increasing ratings granularity raises platform revenue by strengthening seller differentiation echoes this latter insight. When the platform can also decide which sellers participate, \citet{light2024quality} show that the profit-maximizing policy bans a portion of low-quality sellers and conceals distinguishing information among the remaining sellers. \citet{bimpikis2024information} also find that full-information disclosure can be suboptimal under endogenous seller supply. In contrast to these two works, we treat the population of sellers as fixed. We complement their findings by identifying settings in which full-information disclosure is suboptimal even when the platform cannot control seller composition. In particular, our numerical results show that concealing a secondary differentiating attribute can be optimal when that attribute is negatively correlated with an important displayed attribute such as quality.

We conclude by highlighting that our work takes a different methodological approach from the Bayesian persuasion literature. By design, the CLC model lends itself to evaluating concrete platform interventions, such as default rankings and attribute-display policies, whereas the general signal structures considered in the Bayesian persuasion framework need not admit an analogous operational interpretation. A second distinction is that the Bayesian persuasion framework assumes that buyers are fully rational; in contrast, the lexicographic choice model we consider explicitly captures boundedly rational behavior in platform-mediated search. Finally, while many of the theoretical works discussed above study relatively stylized settings for analytical tractability, our simulation-based approach allows us to evaluate the consequences of operationally meaningful information design interventions --- such as search rankings, rating granularity, and attribute ordering --- in richer multi-seller, multi-attribute environments which may not be analytically tractable. Taken together, these features provide a framework for studying how concrete platform design decisions shape competitive outcomes in rich environments with boundedly rational customers. 
}

%% file: final-parts/preliminaries-msom24.tex

\section{{Preliminaries}}\label{sec:preliminaries}

We consider a game of price competition in which $N$ sellers, denoted by $\mathcal{N} := \{1,\ldots,N\}$, each offering a specific item for sale, are faced with a unit mass of customers, each seeking to purchase a single item. Each item has a price set by its seller, as well as a set of $K$ non-price attributes, denoted by $\mathcal{A} := \{1,\ldots, K\}$. We use $0$ to denote the price attribute and let $\oma := \{0, 1,\ldots, K\}$ be the set of all attributes. Attribute $k \in \oma$ is associated with a set $\attributevalues{k}$ of possible values; we allow $\attributevalues{k}$, for $k \in \mathcal{A}$, to be arbitrary, and assume that $\attributevalues{0}=[0,\bar{v}]$ for some finite $\bar{v}>0$ for the price attribute.  For $i\in\cN$ and $k \in \oma$, let $v^k_i\in\attributevalues{k}$ denote the value of attribute $k$ for the item sold by seller $i$ (also referred to as item $i$). 

{\subsection{Consider-then-Choose with Lexicographic Choice}}

Each customer $c$ has a strict ranking $\succ_c$ over attributes in $\oma$, representing attributes' relative importance to the customer. 
Moreover, for each attribute $k\in\oma$, the customer has a weak preference relationship $\succsim^k_c$ over elements of $\cV^k$. In particular, for any $v\neq v'$, $v \succ^k_c v'$ denotes that customer $c$ prefers value $v$ over $v'$ for attribute $k$, and $v \sim^k_c v'$ denotes that $c$ is indifferent between $v$ and $v'$. If $v= v'$, we assume $v \sim^k_c v'$. This preference relation is transitive, i.e., if $v$ is (weakly) preferred to $v'$ and $v'$ is (weakly) preferred to $v''$, then $v$ is (weakly) preferred to $v''$. The preference relationship for any non-price attribute may vary arbitrarily across customers; however, all customers strictly prefer a lower price. Formally, for all $c$, $p< p' \iff p \succ^0_c p'$, and $p = p' \iff p \sim^0_c p'$.

The choice model we consider is partially defined by the lexicographic heuristic \citep{fishburn1974exceptional}. 
\begin{definition}[Lexicographic preferences]\label{def:lex-pref}
Customer $c$ {\it lexicographically prefers} item $i$ to item $j$ if there exists an attribute $k$ such that:
\begin{enumerate}
\item The customer is indifferent between $i$ and $j$ for all attributes considered more important than $k$:
$$v_i^{k'} \sim^{k'}_c v_j^{k'} \quad \forall \ k' \text{ s.t. } k' \succ_c k.$$
\item $i$ is preferred to $j$ on attribute $k$, i.e., $v_i^k \succ_c^k v_j^k.$
\end{enumerate}
If no such $k$ exists, then the customer is indifferent between the two items.
\end{definition}

The lexicographic preference relation is transitive; thus, given any subset of items, it induces a weak preference ordering over items. This results in a top set of items that the customer prefers over all others. 

{
\begin{remark}
The indifference relation $\sim^{k}_c$ for any non-price attribute $k$ allows a customer to be indifferent between different values for that attribute, which is a practical modeling feature. For instance, {if the platform displays customer ratings on a scale of zero to 5, our model allows customers to be indifferent between all items with ratings in the respective ranges of $[0,0.5),[0.5,1),\ldots, [4.5,5]$. Moreover, the indifference levels themselves can vary arbitrarily across customers.} It is, however, important for our results that customers are not indifferent between different prices.
\end{remark}
}

We now provide a complete description of CLC choice.

\begin{definition}[Consider-then-Choose with Lexicographic Choice]\label{def:clc}
Under the {Consider-then-Choose with Lexicographic Choice} (CLC) model, customer $c$ makes her purchase decision in two stages:
\begin{enumerate}
\item {\bf Consideration set formation.} The customer forms a consideration set $\mathcal{N}_c \subseteq \mathcal{N}$ of sellers according to (i) an arbitrary set of criteria on items' non-price attributes and (ii) an upper bound $w_c \in \mathbb{R}_{\geq 0}$ on the price of the item representing her willingness-to-pay, where we assume $w_c\leq \bar{v}$. Formally, let {$\mathcal{X}_c = (-\infty, w_c]$} and $\attributecriteria{c} \subseteq \attributevalues{1}\times \attributevalues{2}\times \ldots\times\attributevalues{K}$. Then, 
$$\mathcal{N}_c = \{i\in\mathcal{N} : (v^0_i,v^1_i,\ldots, v^K_i)\in  \mathcal{X}_c\times \attributecriteria{c}\}.$$
If $\mathcal{N}_c$ is empty, then the customer leaves the platform without purchasing an item.   
\item {\bf Lexicographic choice.} The customer forms a lexicographic preference ordering of items in $\mathcal{N}_c$ according to \cref{def:lex-pref} and chooses the top-ranked item. If there are multiple such items, ties are broken using an arbitrary customer-specific preference ordering over items.
\end{enumerate}
\end{definition}
\begin{remark}
The tie-breaking rule specified above can be embedded in the lexicographic choice framework by introducing a new dummy attribute that is the index of the seller, preferences over which are strict and correspond to a customer's tie-breaking order. This dummy attribute is considered to be the least important attribute by everyone, and is only resorted to in the case of indifferences. Hence, it is without loss of generality to assume that the top-ranked item in any non-empty consideration set is a singleton.
We make this assumption from this point on, and say that the customer chose based on $k$ if $k$ was the first attribute according to which the customer strictly preferred the chosen item. 
\end{remark}

\subsubsection*{Special cases.} The CLC model includes the following important special cases of customer behavior:
\begin{itemize}
\item \textbf{``Best-I-Can-Afford'' (BICA) customers:} These customers choose the highest-quality (or highest-rated) item whose price is at most their willingness-to-pay, breaking ties according to price. Such behavior is anecdotally reported amongst quality-sensitive customers and delegated purchasers with a budget (e.g., for business expenses). {It has moreover been studied within the context of advertising markets \citep{balseiro2021landscape}, and carbon and energy markets \citep{golrezaei2024bidding}, where buyers are modeled as budgeted {value maximizers} in auction settings.}

\item \textbf{``Cheapest-I-Can-Tolerate'' (CICT) customers:} These customers choose the cheapest item whose quality exceeds an acceptable threshold, breaking ties according to quality. Such behavior is commonly reported amongst price-sensitive customers. An extreme version is the case where there is no lower threshold on quality, in which case the customers simply choose the cheapest item. This is the assumption made in the classical Bertrand model of price competition \citep{bertrand1883book}. We refer to this special case as {\it greedy} behavior in the rest of the paper.
\item \textbf{Loyal customers:} These customers only include a single seller in their consideration set, and purchase from this seller as long as the price of the item is below their willingness-to-pay. 
\item \textbf{Satisficers:}  Similar to the model considered in \citet{gallego2023constrained}, these customers determine an acceptable set of items in the consideration phase, and choose an item from this set arbitrarily. This is equivalent to having a strict, customer-specific ranking of sellers along the ``identity'' attribute, and choosing the highest-ranked seller according to this attribute in the lexicographic choice phase.
\end{itemize}
{The lexicographic choice model is also able to encode the well-established phenomenon of {\it position bias} on e-commerce platforms. This can be done by assuming customers have a ``position'' attribute according to which they sort items. For this attribute, if an item appears in the top-$k_1$ results after filtering, then it takes on the highest value; if it appears in between positions $k_1+1$ and $k_2$, then it takes the next highest value, etc., where $k_1<k_2<k_3<\ldots$ are customer-dependent. The simplest version of position bias is one where $k_i = i$. Another version considers a page-dependent bias where customers prefer items on earlier pages to those on later ones. If there are $k$ entries on one page then a simple way to model this is to define $k_i = i\times k$, for each page $i \in \mathbb{N}$.}

{Finally, we highlight that CLC choice is able to model settings where customers make arbitrary trade-offs along non-price attributes. To see this, notice that for any CLC choice model with $K$ attributes, one can define a ($K+1$)-st attribute representing the ``aggregate, non-price quality''  of each item. Letting the value $v_{K+1,i}$ be any function of all non-price attribute values of seller $i$ (e.g., any linear combination of these values), we have defined a valid CLC choice model where the ($K+1$)-st attribute can appear anywhere in customers' preference ordering over attributes. Hence, all results that follow hold for this setting as well; it will however be critical that customers evaluate price separately from all other attributes. 
}

\subsubsection*{Additional notation.} In the remainder of the paper we denote the value of the price attribute of item $i$ as $p_i$ instead of $v^0_i$. Abusing notation, we denote the set of possible prices by $\cV$ instead of $\attributevalues{0}$, and let $\cV^N$ be the space of all possible prices set by the $N$ sellers. Let $\mathcal{G}$ denote the joint probability distribution from which the primitives of customers' choice process (i.e., the ordering over attributes $\succ_c$, the preference relations for each attribute $(\succ_c^k; k\in \{0,\ldots, K\})$, the willingness-to-pay $w_c$, and the subset of admissible non-price attribute values $\attributecriteria{c}$) are drawn. 
 Given a price vector $\pvec = (p_1,\ldots,p_N)$, $\mathcal{G}$ specifies the expected demand $D_i(\bp)$, and thus the expected revenue $R_i(\bp)=p_iD_i(\bp)$ of seller $i \in \cN$.\endnote{We assume the cost for a seller to provision an item is $0$ for expositional simplicity.} We often denote $R_i(\bp)$  and $D_i(\bp)$ as $R_i(p_i,\mathbf{p}_{-i})$ and $D_i(p_i,\mathbf{p}_{-i})$, respectively, when we wish to focus on the dependence of these quantities on seller $i$'s price. Finally, in all examples that follow, we assume that the distribution of customers' willingness-to-pay conditioned on the remaining primitives has a Lipschitz continuous probability density function.

{\subsection{Equilibrium concept}

In general, CLC choice behavior may induce ill-behaved seller revenue functions. For instance, sellers' revenue curves need not be continuous in many settings, and may even be globally non-convex {(as can be seen in \Cref{fig:satisficing}, for the case of satisficing behavior in a duopoly)}. In such settings, local optima of sellers' revenue curves become the point of focus both for structural analysis --- which we study in \Cref{sec:pseudo-equilibria} --- and convergence of gradient-based methods --- which we study in \Cref{sec:numerical-results}. As such, in this work we focus on a generalization of the standard notion of the pure-strategy Nash equilibrium, referred to in the literature as the {\it local} Nash equilibrium (LNE)  \citep{ratliff2016characterization}. We formally define this solution concept below.

\begin{figure}[t]
    \centering
    \includegraphics[scale=0.4]{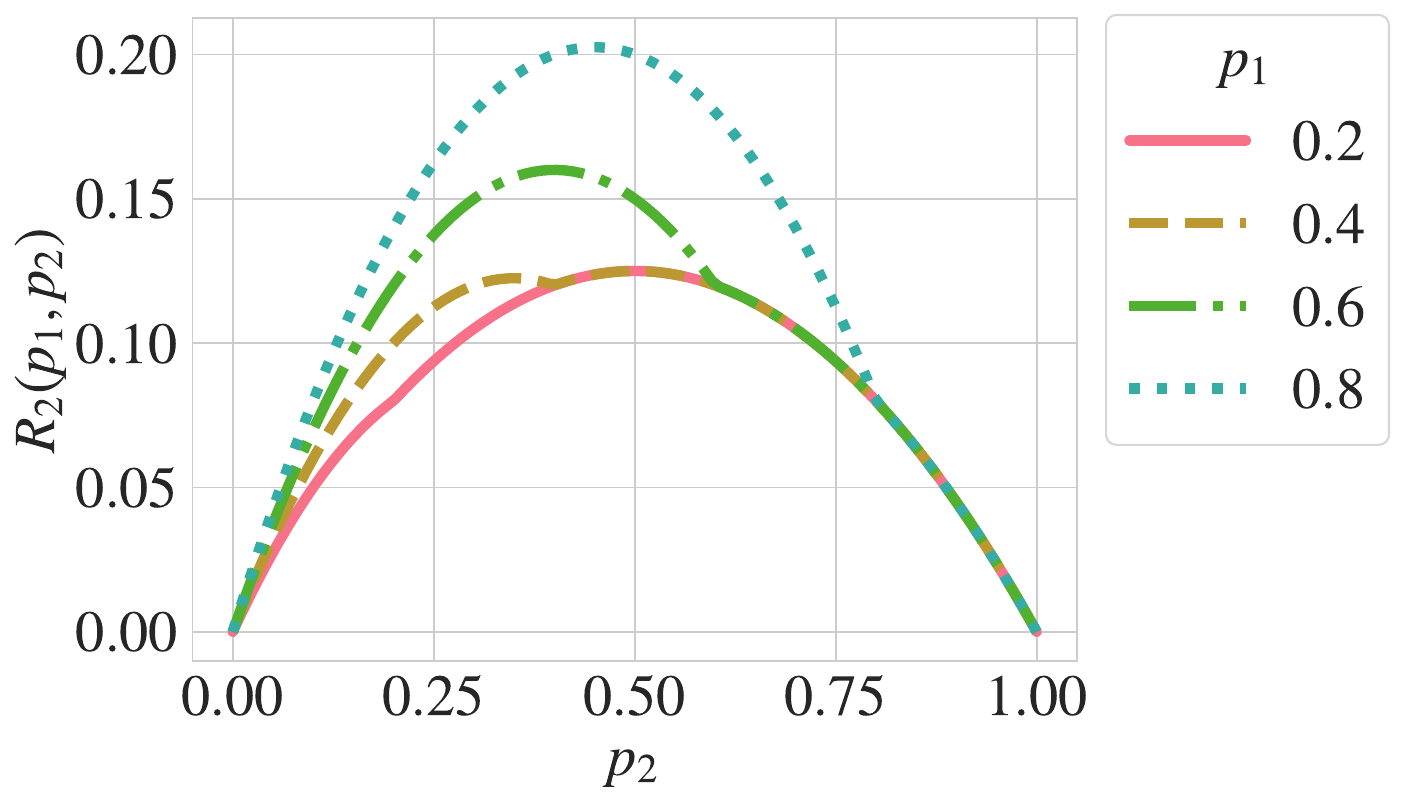}
    \caption{\centering Plot of $R_2(p_1, p_2)$, for different values of $p_1$, when customers satisfice in a duopoly. Here we assume an equal proportion of customers prefer seller 1 to seller 2 (and vice versa), and $w_c\sim\textup{Unif}[0,1]$. This induces a revenue function $R_2(p_1,p_2) = p_2\left(\frac12(1-p_2) + \frac12(p_1-p_2)^+\right)$ for seller 2. 
    }
    \label{fig:satisficing}
\end{figure}

\begin{definition}[Local Nash Equilibrium]
A vector of prices $\bp = (p_1,\cdots, p_N)$, with $p_i\in\mathcal{V}$ for $i\in\mathcal{N}$, is a \emph{local Nash equilibrium} (LNE) if there exist {open} neighborhoods $\cU_i\subseteq \mathbb{R}$ for $i\in\mathcal{N}$, such that (i) $p_i\in\cU_i$, and (ii) $p_i \in {\arg\max_{p \in \overline{\cU}_i} R_i(p,\bp_{-i})}$ for all $i \in \mathcal{N}$, where $\overline{\mathcal{U}}_i$ denotes the closure of $\mathcal{U}_i$.
\end{definition}

Informally, a price vector constitutes a LNE if each price is a ``local'' best-response (i.e., a local maximum of the seller's revenue curve) to all other sellers' prices. Observe that any pure-strategy Nash equilibrium is a LNE. Since sellers' prices are {\it global} best-responses under this former solution concept, we alternatively refer to these as {\it global} Nash equilibria (GNE).
}

%% file: final-parts/equilibrium-analysis-vijay.tex
\section{Structure and Computation of Equilibria}\label{sec:pseudo-equilibria}

{Having introduced our equilibrium concept, we now turn to the question of tractably computing LNE under CLC choice. We first provide a structural characterization of equilibria under this model. This characterization will allow us to provide a computational procedure that (i) determines whether an equilibrium exists, and (ii) outputs all LNE when it does. We then strengthen these results by providing an interpretable sufficient condition under which an equilibrium is {\it guaranteed} to exist and can be efficiently computed using our computational procedure.}

\subsection{Structural characterization of equilibria}

{We first introduce a property,} which we term the {\it pseudo-competitive} property, that drives the majority of our results. 
Given price vector $\bp$ and price $x$, let $\mathcal{S}_{\bp}(x)=\left\{ i\in \cN:\, p_i\geq x\right\}$ be the set of sellers whose price is at least $x$. For $\mathcal{S}\subseteq \mathcal{N}$, $\bp^\mathcal{S}$ denotes the price vector in which the prices of all sellers in $\mathcal{S}$ are preserved, and those in $\mathcal{N}\setminus \mathcal{S}$ are set to $0$.

\begin{definition}[Pseudo-Competitive Property]\label{def:psg-property}
The {\it pseudo-competitive} property holds if, for all $i\in\mathcal{N}$, and $\pvec$, $\pvec' \in \cV^N$ such that:
\begin{enumerate}
\item $\mathcal{S}_{\bp}(p_i) =  \mathcal{S}_{\bp'}(p_i)$, i.e., the same set of sellers price weakly above $p_i$ in $\pvec$ and $\pvec'$, and 
\item $p_j = p'_j$ for all $j\in \mathcal{S}_{\bp}(p_i)$, i.e., these higher-priced sellers set the same prices in $\pvec$ and $\pvec'$,
\end{enumerate}
 we have that $R_{i}(\bp) = R_{i}(\bp')$. 
\end{definition}

{In words, the pseudo-competitive property holds if the revenue of a seller priced at $p$ remains unchanged if any seller pricing below $p$ changes its price to some other $p'  < p$. That is, price changes of lower-priced sellers do not impact the revenue of a higher-priced seller, as long as they continue to be lower-priced. The pseudo-competitive property formalizes the intuition that customers who select a higher-priced alternative over a lower-priced one do so because of its dominance on some non-price attribute (e.g., quality), and are unlikely to respond to lower price changes, given that they are choosing based on a non-price attribute.} 

Our first main result establishes that the revenue functions arising from the CLC choice model satisfy this property. We defer its proof to Appendix \ref{apx:psg-proof}.
\begin{theorem}\label{prop:clc-choice-implies-psg}
The pseudo-competitive property holds under any CLC choice model.
\end{theorem}

We now define the notion of a {\it valid order-price pair}, which will be used for our main structural result. For any subset $\cS\subseteq\cN$, we let {$\nearestminprice = \min\{p_j ; j \in \cS\}$} denote the lowest price amongst sellers in $\cS$, where we adopt the convention that $p^{\min}(\phi) = \bar{v}$. For $i \in \cN$, $\cS\subseteq \cN \setminus\{i\}$, the {\it undercutting interval} or {\it undercutting region} of seller $i$ relative to $\cS$ refers to the interval $[0, \nearestminprice]$. 
We refer to $R_i(p, \bp_{-i}^{\cS})$ in the undercutting region $[0, \nearestminprice]$ as the {\it undercutting revenue curve} of seller $i$ relative to $\cS$ (where the prices of all sellers not in $\cS$ are set to $0$).
\begin{definition}[Valid order-price pair]\label{def:valid-order}
A {\it valid order-price pair} (VOP) is a tuple $(\pi, \bp)$ consisting of an ordering of sellers $\pi:[N]\rightarrow\mathcal{N}$ (i.e., $\pi(i)$ is the seller in rank $i$), and prices $\bp\in\cV^N$ such that:
\begin{enumerate}
\item $0\leq p_{\pi(N)} \leq  p_{\pi(N-1)} \leq\cdots\leq p_{\pi(1)}\leq \bar{v}$.
\item Let $\mathcal{S}_1$ be the empty set, denoted by $\phi$, and for $i\in \{2,\cdots, N\}$, let $\mathcal{S}_i = \{\pi(j) : j<i\}$ denote the sellers ranked higher than $i$ in $\pi$. For all $i$, $p_{\pi(i)}$ is a local maximum of $R_{\pi(i)}(p, \bp^{\mathcal{S}_i}_{-\pi(i)})$ in the undercutting interval $[0,  p^{\min}(\mathcal{S}_i)]$. We call these prices {\it valid} with respect to $\pi$.
\end{enumerate}
\end{definition} 
Computationally, a valid order-price pair is more tractable than a LNE. In particular, for any ordering of sellers, the corresponding valid prices can be determined by sequentially solving (according to the ordering) $N$ constrained single-agent local optimization problems. Each optimization problem corresponds to a rank in the ordering where the correspondingly ranked seller computes its undercutting best-response to the prices of sellers higher in the ordering, assuming that all other sellers ranked lower set their prices to $0$. Our main structural result shows that any LNE under CLC choice {\it must} be a VOP. 

\begin{theorem}\label{thm:general-LNE}
Under CLC choice, a price vector $\bp$ is a LNE only if it belongs to a VOP.
\end{theorem}

In order to prove \cref{thm:general-LNE}, we argue that, in any LNE, the local optimality of a seller's price is preserved if {\it any} subset of sellers pricing weakly lower in that LNE reduced their price to 0. When sellers' prices are distinct, this fact is implied by the pseudo-competitive property, i.e., a seller's revenue is not impacted by the prices of strictly lower-priced sellers. The non-trivial aspect of the result is in showing that it holds {\it even in LNE where sellers set equal prices}, in which case the pseudo-competitive property fails to yield any meaningful implications (as it only pertains to {\it strictly} lower price changes). {The result in this case leverages the structure of CLC choice}; we defer its formal proof to Appendix \ref{apx:proof-of-lne-result}.

\subsection{Efficient computation of equilibria}\label{sec:efficient-computation}

 \Cref{thm:general-LNE} has important implications for equilibrium computation under CLC choice. To elaborate further, we introduce a mild assumption on sellers' undercutting revenue curves. For $i \in \cN$, and $\bp_{-i} \in \cV^{N-1}$, let  $g_i(p_i,\bp_{-i})=\partial R_i(p, \bp_{-i})/\partial p\big|_{p=p_i}$ denote the partial derivative of $R_i(p,\bp_{-i})$ with respect to $i$'s price. Similarly, for all $i\in \cN$, $\cS \subseteq \cN \setminus \{i\}$ and $\bp_{-i} \in \cV^{N-1}$, such that $\nearestminprice>0$, we define  $g_i(p_i,\bp_{-i}^{\cS})=\partial R_i(p, \bp_{-i}^{\cS})/\partial p\big|_{p=p_i}$. We assume that this partial derivative exists for all $p_i \in [0,\nearestminprice)$.

\begin{assumption}[Unimodality of undercutting revenue curves]\label{asp:unimodality}
For all $i\in \cN$, $\cS \subseteq \cN \setminus \{i\}$ and $\bp_{-i} \in \cV^{N-1}$, such that $\nearestminprice>0$, $R_i(p,\bp_{-i}^{\cS})$ is unimodal in the undercutting region, i.e., if $g_i(p,\bp_{-i}^{\cS})\leq 0$ for some $p\in [0,\nearestminprice)$, then $g_i(p',\bp_{-i}^{\cS})< 0$ for all $p'\in (p,\nearestminprice)$.
\end{assumption}

This assumption enforces the mild condition that if a local optimum exists in the undercutting region for a seller, it must be unique. (Note that we assume unimodality of the revenue curve of a seller {\it only in the undercutting region} relative to the prices of the set of sellers $\cS$. This assumption still allows for multiple local optima over $[0, \bar{v}]$.) {We present special cases of CLC choice for which \Cref{asp:unimodality} holds in Appendix \ref{apx:unimodality}. Namely, when $N = 2$, under mild conditions on customers' willingness-to-pay distribution, it holds as long as a strictly positive fraction of customers prefers each seller according to a non-price attribute ranked higher than price. For general $N$, it holds for any mixture of loyal, BICA, and greedy customers.} 

 Under \Cref{asp:unimodality}, \cref{thm:general-LNE} implies that enumerating {\it all} LNE (and thus GNE) is computationally tractable when the number of sellers is small. This is because {the platform} can compute all valid order-price pairs, each of which requires solving a sequence of unimodal optimization problems, and then performing a computationally easy check for whether each of these is a LNE (or a GNE). We formalize this in \cref{cor:eq-easy}, deferring its proof to Appendix \ref{apx:efficient-comp}.

\begin{proposition}\label{cor:eq-easy}
Under Assumption \ref{asp:unimodality}, there exists a computational procedure that:
\begin{enumerate}
\item finds all LNE, or determines that none exists, in {$\textup{O}(N\cdot N!)$} undercutting best-response computations;
\item finds all GNE, or determines that none exists, in $\textup{O}(N^2 \cdot N!)$ undercutting best-response computations.
\end{enumerate}
{Moreover, each of the undercutting best-response computations in these procedures is a one-dimensional unimodal maximization problem.}{\endnote{{Though the unimodality assumption is standard, an analogous procedure would still compute a LNE/GNE (or return that none exists) when undercutting revenue curves exhibit multiple local optima, at the cost of a corresponding increase in the number of best-response computations.}}}
\end{proposition}

 The factorial dependence on the number of sellers in our procedure results from the exhaustive search over all seller orderings and the corresponding valid prices (if they exist). This dependence is inevitable in the worst case, even if the goal is to find a single LNE/GNE. This is because there could be exactly one LNE corresponding to an ordering that is a priori unknown, which may require $O(N!)$ iterations to discover. Moreover, in the case where a LNE does not exist {(we provide an example where this is the case in Appendix \ref{apx:non-existence-equilibria})}, an exhaustive search over all orderings is unavoidable to determine this. 
 
{Although the factorial dependence on the number of sellers poses computational challenges for large $N$, this class of games remains more tractable than general game-theoretic settings. For instance, computing a GNE is well-known to be intractable even in simple {\it two}-player bimatrix games \citep{daskalakis2009complexity}, whereas our setting admits tractable analysis for small $N$. The source of computational hardness in our setting is thus the combinatorial growth in $N$, rather than a fundamental intractability present even in the two-player case. Nevertheless, in the remainder of this section we turn to the task of finding sufficient conditions on the CLC class such that (i) a LNE is guaranteed to exist, and (ii) computing a LNE is polynomial in $N$. We present such a condition in \Cref{asp:grad-dom} below.}  

\begin{assumption}[Gradient dominance]\label{asp:grad-dom}
For two sellers $i, j \in \cN$, seller $i$ gradient-dominates seller $j$ if for all $\bp_{-\{i,j\}} \in\cV^{N-2}$, and $p \in (0,\bar{v})$:
\begin{align}
 g_i(p_i = p, p_j = 0, \bp_{-\{i,j\}}) &> \lim_{q\uparrow p} g_j(p_i = p, p_j = q,\bp_{-\{i,j\}}).
 \label{eq:i-dominates-j} 
\end{align}
Under gradient dominance, either $i$ gradient-dominates $j$ or $j$ gradient-dominates $i$, for all $i, j \in \cN$.
\end{assumption}
We say that {\it weak} gradient dominance holds if this inequality holds for all price vectors in which a seller setting a positive price makes strictly positive revenue. (We call such price vectors {\it non-trivial}.)

{At a high level, gradient dominance describes settings in which there is an asymmetric competitive relationship between each pair of sellers, with one seller maintaining market dominance relative to the other due to a non-price attribute. To illustrate, consider a premium and a standard yogurt brand, and suppose the premium brand gradient-dominates the standard brand. This means that, if the price of the standard brand is just below that of the premium brand, the premium brand has a lower incentive to reduce its price than the standard brand. This is intuitive: the premium brand may likely already command a large customer base willing and able to pay for higher quality, making price cuts relatively costly. Such price cuts are more lucrative for the standard brand, on the other hand. This is because, by providing a lower price, the standard brand stands to gain quality-sensitive customers who cannot afford the premium brand. {More generally, for any two sellers $i$ and $j$, $i$ gradient-dominates $j$ if $j$ benefits from ``spillover'' from $i$; this spillover corresponds to customers who initially prefer $i$ on a non-price attribute, but whose low willingness-to-pay makes $i$ unaffordable, causing them to purchase from $j$ instead. This then pushes the revenue-maximizing price of $j$ downward, relative to that of $i$. We mathematically formalize this intuition in Appendix \ref{ssec:gradient-dominance-discussion}, where we present sufficient conditions on the choice primitives for gradient dominance to hold (in particular, we show that it holds under BICA behavior, for general $N$).
}
}

\Cref{prop:gd-implies-lne-existence-uniqueness} below establishes the connection between gradient dominance and existence and tractability of LNE; we defer its proof to Appendix \ref{apx:gd-implies-lne-existence-uniqueness}. 
\begin{proposition}\label{prop:gd-implies-lne-existence-uniqueness}
Under Assumptions  \ref{asp:unimodality} and \ref{asp:grad-dom}, there exists a LNE $\bp^*$ in which all prices are distinct, with $p_i^* > p_j^*$ for all sellers $i,j$ such that $i$ gradient-dominates $j$. Moreover, there exists a procedure that computes such a LNE in $O(N^2)$ iterations, assuming knowledge of the gradient dominance relation for every pair of sellers.
\end{proposition}

{In addition to being practically useful for the task of predicting platform outcomes, \Cref{prop:gd-implies-lne-existence-uniqueness} suggests that a platform designer can expect stable market outcomes when sellers differentiate themselves along non-price attributes. This provides actionable guidance for platforms: promoting non-price differentiation (e.g., through ratings information, detailed product specifications, or seller reputation indicators) can foster stable market outcomes by reducing direct head-to-head competition. In turn, interface design choices that emphasize these attributes help ensure such outcomes, which can result in higher revenue for the platform, as we will see in \Cref{sec:case-study}. This result also suggests an important strategic insight for sellers themselves: in order to sustain positive revenue in equilibrium, sellers should invest in quality improvements and feature differentiation as competitive strategies.}

%% file: final-parts/algorithmic-competition-vijay.tex

\section{{Simulating Equilibria Via Gradient-Based Dynamics}} \label{sec:numerical-results}

{The equilibrium computation procedure derived in \Cref{sec:pseudo-equilibria} is most practical in settings where the platform has access to sellers' entire revenue functions, for all possible price combinations. However, in complex market scenarios with a large number of sellers and a wide range of behaviors within the CLC class, characterizing the entire revenue curve of each seller may be challenging. A more practical alternative for the platform in such cases is to instead use a generative simulator of customer behavior as a decision-support tool. In particular, given access to fine-grained customer data, the platform can estimate a CLC choice model (using, e.g., the methods developed in \citet{kohli2007representation,kohli2019randomized}), and can then simulate sellers independently running myopic, gradient-based price adjustments in response to simulated customer demand. The platform can interpret the outcome of this simulation as a prediction of equilibrium market conditions and, in turn, use this as a basis for assessing the consequences of design interventions. For concreteness, we outline such a procedure in Appendix \ref{apx:simulation-framework}. 

The usefulness of this approach, however, depends critically on whether such myopic price adjustment dynamics converge to local Nash equilibria. In general, such methods are known to cycle or diverge \citep{mertikopoulos2018cycles}, which would undermine the predictive value of simulations. In this section we lend theoretical support to this practice by proving that, under the gradient dominance condition and additional mild regularity conditions often made in the convex optimization literature, such gradient-based dynamics indeed converge to a local Nash equilibrium.
}

{
\subsection{Setup}
In the simulation procedure described above, the platform is assumed to simulate customer responses to sellers' prices, thereby generating samples of the sellers’ revenue functions at those prices (a setting commonly referred to as {\it zeroth-order} feedback in the learning-in-games literature). For analytical tractability, however, and consistent with prior work \citep{golrezaei2020no,golowich2020tight,hsieh2021adaptive}, our theoretical results assume that the platform has access to the partial gradients of the sellers’ revenue functions ({\it first-order} feedback). {
Given these partial gradients, we assume the platform simulates sellers independently running the well-known gradient ascent (GA) algorithm \citep{boyd2004convex}, the first-order feedback analog of the price updates that the platform would simulate in practice, as described in Appendix \ref{apx:simulation-framework}. The dynamics under the GA algorithm are as follows, at a high level. In each timestep, each seller $i \in \cN$ takes a step in the direction of the partial gradient of their revenue function, holding all other prices fixed; the seller does so only if there is no other seller with the same price in that timestep. If, on the other hand, there exists a seller $j \neq i$ with the same price, the partial gradient may not be well-defined (see \Cref{fig:satisficing} for an illustration); in this case, we assume that seller $i$ samples between the left- and right-partial gradients uniformly at random. We formally present this algorithm in Appendix \ref{apx:alg-ogd}.}
}
{
\begin{remark}
In our numerical experiments (see Appendix \ref{sec:empirics}) we make the more realistic assumption that the platform only has access to zeroth-order feedback and must implement a gradient estimation step, as described in Appendix \ref{apx:simulation-framework}. We observe consistent convergence to a local Nash equilibrium across all experiments, highlighting the robustness of our theoretical results.
\end{remark}
}

\subsubsection*{Additional notation and assumptions.} For $i \in \cN$, recall that  $g_i(p_i,\bp_{-i})=\partial R_i(p, \bp_{-i})/\partial p\big|_{p=p_i}$ denotes the partial derivative of $R_i(p,\bp_{-i})$, assuming it exists. 
For $\mathcal{S} \subseteq \cN \setminus \{i\}$, $p_i^*(\bp_{-i}^\cS) \in \arg\max_{p\in[0, \nearestminprice)} R_i(p,\bp_{-i}^\cS)$ denotes an optimal price of seller $i$ in the open undercutting region relative to the set of sellers $\cS$. When we say that $p_i^*(\bp_{-i}^\cS)$ exists, we mean that the supremum is attained at $p_i^*(\bp_{-i}^\cS)$. If the supremum is not attained, we use the convention that $p_i^*(\bp_{-i}^\cS) = \infty$. Finally, we say that a price vector $\bp$ is {\it aligned} with ordering $\pi$ if $p_{\pi(1)} > p_{\pi(2)} > \ldots > p_{\pi(N)}$. We will often abuse notation and let $\pi(1:l)$ denote the set of sellers in ranks $1$ to $l$ for $l\in[N]$. We additionally let $\pi(1:0)$ be the empty set $\phi$, and define $p_{\pi(0)} = \pmax$. The following assumptions are required for our convergence result.
\begin{assumption}[Regularity of undercutting revenue curves]\label{asp:ogd-rev-n-firms}
The following holds for all $i\in \cN$, $\cS \subseteq \cN \setminus \{i\}$ and $\bp_{-i} \in \cV^{N-1}$.
\begin{enumerate}
 \item {Uniform boundedness of partial gradients}: There exists a finite constant $G>0$ such that  $|g_i(p, \bp_{-i})|\leq G$ for all $p\in \mathbb{R}$.
\item Smoothness: $g_i(p,\bp_{-i}^{\cS})$ is continuous for all $p \in [0, \nearestminprice)$. Moreover, for all $i \in \cN$, if $\bp$ and $\bq$ are aligned with the same ordering, there exists $L > 0$ such that  $|g_i(\bp) - g_i(\bq)|\leq L \|\bp -\bq\|_1.$
\item {Strong unimodality}: There exists a finite constant $K > 0$ such that
\begin{align}\label{eq:loyalist-well-separated}
|g_i(p, \bp_{-i}^\cS)| \geq K |p - \min\{p_i^*(\bp_{-i}^\cS), \nearestminprice\}| \quad \forall \,  p \in [0, \nearestminprice). 
\end{align}
\item No vanishing gradient below a partial VOP: 
Consider any ordering $\pi$, rank $l \in [N]$, and partial valid order-price pair $(\pi, \bp)$ that satisfies: 
\begin{align*}
p_{\pi(l')} := \arg\max_{0\leq p \leq p_{\pi(l'-1)}}R_{\pi(l')}(p, \bp_{-\pi(l')}^{\pi(1:l'-1)}) \quad \textrm{for } l' \in \{1,\ldots,l\}.
 \end{align*}
If the supremum is not attained in the problem $\sup_{0\leq p < p_{\pi(l)}}R_{\pi(l+1)}(p, \bp_{-\pi(l+1)}^{\pi(1:l)})$, then 
\begin{align}\label{eq:not-flat}
\lim_{p\uparrow p_{\pi(l)}}  g_{\pi(l+1)}(p, \bp_{-\pi(l+1)}^{\pi(1:l)}) > 0.
\end{align}
\end{enumerate}
\end{assumption}

Conditions 1 and 2 in \Cref{asp:ogd-rev-n-firms} are standard technical conditions in much of the optimization literature on the convergence of gradient ascent and its variants (see, e.g., \citet{boyd2004convex}). {Condition 3, also common in the existing literature \citep{broadie2011general}, states that sellers' revenue functions are not too flat around their respective undercutting optima. Finally, Condition 4 specifies that if a seller's supremum is not attained when all higher-priced sellers are at their respective local optima, its gradient does not vanish on approaching the next higher-priced seller from below. Informally, if $l$ sellers price optimally according to some ordering $\pi$, and seller $\pi(l+1)$ does not find it optimal to undercut $\pi(l)$, then $\pi(l+1)$ must have a local optimum that is strictly above $p_{\pi(l)}$.} {In Appendix \ref{apx:asp3-discussion} we show that Assumption \ref{asp:ogd-rev-n-firms} holds for common willingness-to-pay distributions, as long as a positive fraction of customers spill over from their preferred seller to another, if their preferred seller is too expensive.
}

\subsection{Convergence of distributed gradient ascent dynamics}

Theorem~\ref{thm:ogd-converges-to-LNE-n-firms} below establishes that gradient dominance not only guarantees that LNE can be efficiently computed in settings where the revenue curves are known, but it also guarantees that a platform can reliably {\it simulate} market outcomes in more complex settings where the revenue curves are unknown and only a simulator of customer behavior is available.  %

\begin{theorem}\label{thm:ogd-converges-to-LNE-n-firms}
Under Assumptions \ref{asp:grad-dom} and \ref{asp:ogd-rev-n-firms}, a LNE with distinct prices exists. If moreover the learning rate schedule under distributed GA (Algorithm~\ref{alg:ogd}) is such that $\lim_{T\to\infty}\sum_{t = 1}^T \eta_{t}^2 < \infty$ and $\lim_{T\to\infty}\sum_{t = 1}^T \eta_{t} = \infty$, sellers' prices converge almost surely to a LNE.
\end{theorem}

{In Appendix \ref{sec:empirics} we numerically demonstrate the robustness of our theoretical findings. In particular, we observe frequent convergence of gradient-based price dynamics in settings where Assumptions \ref{asp:grad-dom} and \ref{asp:ogd-rev-n-firms} need not hold, and under the more realistic assumption that the platform must estimate sellers' partial gradients by simulating customer purchases via the estimated CLC choice model.}

We next provide a proof sketch of \cref{thm:ogd-converges-to-LNE-n-firms}, deferring its formal proof to Appendix \ref{ssec:n-firm-convergence}.
\begin{proof}{\it Proof sketch.}
Existence of a LNE $\bp^*$ in which sellers price distinctly follows from \Cref{prop:gd-implies-lne-existence-uniqueness}, since strong unimodality (Condition 3 in \Cref{asp:ogd-rev-n-firms}) implies \Cref{asp:unimodality}. Let $\pi$ denote the ordering with which $\bp^*$ is aligned.

Suppose first that sellers' prices are aligned with $\pi$ for a large enough timestep. Since sellers' revenue curves are invariant to strictly lower price changes by the pseudo-competitive property, \Cref{asp:ogd-rev-n-firms} implies that sellers will hierarchically converge to their respective undercutting optima.

Suppose now that the prices are aligned with an ordering $\pi' \neq \pi$ for which a seller's undercutting optimum is not attained. By Condition 4 of Assumption~\ref{asp:ogd-rev-n-firms}, this seller's gradient does not vanish as it approaches the boundary of its undercutting region. As a result, it will overtake the next higher price, thus ensuring that sellers exit this ordering.

Gradient dominance finally ensures that, when the stepsize is small enough, once the price dynamics exit an ordering, they cannot subsequently re-enter that ordering. Since the number of orderings is finite, they must eventually enter an ordering in which all sellers attain their undercutting optima. Hence, distributed GA eventually converges to a LNE.\hfill\Halmos
\end{proof}

{We conclude the section by noting that an immediate consequence of \Cref{thm:ogd-converges-to-LNE-n-firms} is that, under the gradient dominance condition, sellers running simple local price updates --- as is common in practice --- are guaranteed to converge to a local Nash equilibrium. This finding then supports our focus on LNE as the relevant solution concept for platforms.}

%% file: final-parts/case-study.tex
{
\section{Implications for Platform Design}\label{sec:case-study}

{In this section we study how platform design shapes equilibrium market outcomes. We focus on three classes of interventions --- the platform's default ranking, the disclosure and granularity of seller ratings, and the disclosure of an additional differentiating attribute --- and characterize how these interventions affect equilibrium prices and, in turn, platform revenue and market penetration. We first derive analytical results in a stylized duopoly, and subsequently use the simulation framework developed in \Cref{sec:numerical-results} to extend our analysis to a richer competitive setting with a large number of sellers and inter-attribute correlations.}

\subsection{The duopoly setting}\label{ssec:duopoly-case-study}

\subsubsection{Impact of displaying average rating.}\label{ssec:avg-rating-impact}

Consider a setting in which items have two attributes: price and average rating.   We assume customers have willingness-to-pay $w_c\sim \textup{Unif}[0,1]$. If the platform only displays sellers' prices, customers purchase from the cheapest seller so long as their willingness-to-pay exceeds this seller's price. On the other hand, if the platform displays sellers' ratings, $\lambda \in (0,1)$ fraction of customers first evaluate sellers according to the rating attribute: they purchase from the higher-rated seller if their willingness-to-pay $w_c \geq p_1$; if $w_c < p_1$, they purchase from the lower-rated seller as long as $w_c \geq p_2$. We refer to these customers as {\it ratings-first} customers. The remaining $1-\lambda$ fraction of customers first evaluate sellers according to the price attribute: they purchase from the cheapest seller, breaking ties in favor of average rating. We refer to these customers as {\it price-first} customers. Without loss of generality, suppose seller 1 is the higher-rated seller. 

The platform seeks to evaluate the impact of displaying sellers' ratings on its revenue, as well as the fraction of total demand served. We refer to the setting in which the platform only displays sellers' prices as the {\it coarse} setting, and the one in which it displays sellers' average ratings as the {\it granular} setting. 

The following proposition shows that the design choice of whether or not to display rating information has a significant impact on equilibrium prices.
\begin{proposition}\label{prop:special-case-LNE-q1}
    In the coarse setting, there exists a unique non-trivial (\emph{coarse}) LNE in which $p_1^* = p_2^* = 0$. In the granular setting, there exists a unique non-trivial (\emph{granular}) LNE, with $p_1^* = \frac12$ and $p_2^* = \frac12-\frac\lambda4$.
\end{proposition}

\Cref{prop:special-case-LNE-q1} recovers the classic insight that quality differentiation softens price competition, here in the context of CLC choice. In particular, by displaying average ratings, the platform {\it increases} the equilibrium prices of all sellers. Moreover, as the price sensitivity of the customer population {increases} (i.e., $\lambda$ decreases), the price of the lower-quality seller increases. This follows from the fact that, as $\lambda$ decreases, more of seller 2's customer population consists of ``greedy'' customers who evaluate price over average rating. Since these customers will choose seller 2 over seller 1, for {\it any} strictly lower price than $p_1^*$, seller 2 has less of an incentive to set a low price. We defer the proof of the result to Appendix \ref{apx:special-case-LNE-q1}.

We now leverage these characterizations to study the impact of this decision on revenue and fraction of demand served. For $i \in \{1,2\}$, let $\mathcal{R}_i^c$ denote the revenue of seller $i$ in the coarse LNE. Moreover, let $\Pi^c$ and $\mathcal{D}^c$ respectively denote the total revenue and demand on the platform under this LNE. Finally, we use the superscript {$g$} to refer to these metrics in the granular LNE. 
\begin{proposition}\label{prop:metrics-special-case-q1}
The following hold, for all $\lambda \in (0,1)$:
    \begin{itemize}
    \item Revenue dominance of the granular LNE: $\mathcal{R}_i^g > \mathcal{R}_i^c$ for all $i \in \{1,2\}$;
    \item Demand dominance of the coarse LNE: $\mathcal{D}^c > \mathcal{D}^g$.
    \end{itemize}
\end{proposition}
\Cref{prop:metrics-special-case-q1} similarly formalizes the intuition that both sellers are better off when ratings information is revealed, since they no longer engage in a race to the bottom. As an immediate corollary, the total revenue on the platform is larger in this LNE. However, the fraction of demand served is smaller, due to the uniform price increase across all sellers. We defer the proof of \Cref{prop:metrics-special-case-q1} to Appendix \ref{apx:metrics-special-case-q1}.

Finally, this setup allows us to gain insight into the impact of the platform's default display ranking on market outcomes. In particular, suppose that sellers are ranked either by decreasing average rating (the {\it Ratings-First} ranking) or by increasing price (the {\it Price-First} ranking). Given the position bias effect, more customers are assumed to anchor on this default: under Ratings-First, $\lambda' = \lambda+\delta$ fraction of customers evaluate rating before price, for some $\delta > 0$; under Price-First, $\lambda' = \lambda-\delta$ evaluate rating before price. The following result builds upon \Cref{prop:metrics-special-case-q1} to establish that the Ratings-First ordering Pareto dominates the Price-First ordering. We defer its proof to Appendix \ref{apx:ranking}.

\begin{proposition}\label{prop:ranking}
The following hold, for all $\lambda \in (0,1), \delta < \min\{\lambda,1-\lambda\}$:
\begin{itemize}
\item The equilibrium price of seller 1 is the same under the Price-First and Ratings-First rankings. However, the equilibrium price of seller 2 is strictly higher under the Price-First ranking.
\item The Ratings-First ordering maximizes both platform revenue and demand served.
\end{itemize}
\end{proposition}

\subsubsection{Impact of displaying a third differentiating attribute.}\label{ssec:third-attribute}

We use these results to study the impact of disclosing a third differentiating attribute. One such example is shipping speed; for instance, some sellers on Amazon offer shorter delivery times by being enrolled in the ``Fulfilled by Amazon'' program. Another example is the ``Overall Pick'' certification badge present on many platforms. Such a certification represents the {\it platform's} evaluation of the overall quality of the item, as compared to {\it customers'} evaluation of the item quality. This distinction appears on other platforms such as Expedia, which display a star rating, determined by the platform, and a guest rating, determined by previous customers \citep{ExpediaHotelStarRating}. For concreteness, we refer to this extra attribute as the certification attribute in the remainder of this section; we moreover assume that only seller 2 is certified. (This latter assumption is meant to analyze the impact of {\it differentiating} attribute information. It is easy to see that, if seller 1 is also certified, this setting collapses to the ratings-only setting analyzed above.)

\textbf{Setup.} If the platform discloses the certification badge, $\alpha_1 \in (0,1)$ fraction of customers are ratings-first. On the other hand, $\alpha_2 \in (0,1)$ fraction evaluate the certification attribute first: these customers purchase from seller 2 as long as $w_c \geq p_2$; if $w_c < p_2$, they purchase from seller 1 as long as $w_c \geq p_1$. We refer to these customers as {\it certification-first} customers. The remaining $\gamma := 1-\alpha_1-\alpha_2$ fraction of customers evaluate the price attribute first, breaking ties uniformly at random. We moreover assume that $\alpha_1 \leq \lambda$, reflecting that the disclosure of this new attribute does not increase the fraction of ratings-first customers. As before, we assume $w_c\sim \textup{Unif}[0,1]$, independent of customer type.

\Cref{prop:special-case-LNE-q2} establishes that this additional disclosure not only shifts equilibrium prices but also creates equilibrium multiplicity.
\begin{proposition}\label{prop:special-case-LNE-q2}
    There exist two non-trivial LNE, characterized as follows:
    \begin{enumerate}[(i)]
    \item The Ratings-First LNE, in which seller 1 prices above seller 2: $p_1^* = \frac12$, $p_2^* = \frac12-\frac{\alpha_1}{4}.$
    \item The Certification-First LNE, in which seller 2 prices above seller 1: $p_2^* = \frac12$, $p_1^* = \frac12-\frac{\alpha_2}{4}$.
    \end{enumerate}
\end{proposition}
We defer the proof of the result to Appendix \ref{apx:special-case-LNE-q2}. \Cref{prop:special-case-LNE-q2} illustrates how disclosing this third attribute fundamentally reshapes market outcomes. In disclosing it, the platform has ``created'' a {\it second} equilibrium in which seller 2 becomes the market leader, with seller 1 undercutting. 

In the Ratings-First equilibrium, seller 1 only sells to ratings-first customers, with seller 2 selling to all types of customers, at a higher price than when only average rating information was shown. Moreover, the equilibrium price of seller 2 increases as the fraction of ratings-first customers decreases. The reason for this behavior is that, as this occurs, the customer population consists of more certification-first and price-first individuals. Given that both of these types of individuals will purchase from seller 2 when it undercuts, seller 2 has less of an incentive to optimize for ratings-first customers for whom $w_c < p_1$, and therefore raises its price. In the Certification-First equilibrium, on the other hand, seller 2 only sells to certification-first customers, with seller 1 selling to all customer types. In this case, we observe that seller 1's equilibrium price increases as the fraction of certification-first customers decreases, for entirely analogous reasons.

Before evaluating platform revenue and demand served in each of these equilibria, we introduce additional notation. We use the superscripts ${rf}$ and ${cf}$ to refer to our metrics in the Ratings-First and Certification-First LNE, respectively. Moreover, for any equilibrium denoted by superscript $e$, we let $\rho^e = \Pi^e/\Pi^g$ be the ratio of platform revenue under equilibrium $e$ to its revenue under the granular equilibrium.
\begin{proposition}\label{prop:metrics-special-case-q2}
For all $\lambda \in (0,1), \alpha_1 \in (0,\lambda]$, $\alpha_2 \in (0,1)$, (i) $\mathcal{D}^{rf} \leq \mathcal{D}^g$, and (ii) $\mathcal{D}^{cf} \geq \mathcal{D}^g$ if and only if $\alpha_2 \geq \lambda$. Moreover:
    \begin{itemize}
    \item Revenue comparison of Ratings-First and Granular LNE: 
    \begin{itemize}
    \item $\mathcal{R}_1^{rf} \leq \mathcal{R}_1^{g}$ and $\mathcal{R}_2^{rf} \geq \mathcal{R}_2^{g}$
    \item $\rho^{rf} \in \left[\max\left\{\left(\frac{\alpha_1}{\lambda}\right)^2,\frac{4}{5}\right\},1\right]$
    \end{itemize}
    \item Revenue comparison of Certification-First and Granular LNE:
    \begin{itemize}
    \item $\mathcal{R}_1^{cf} \geq \mathcal{R}_1^g$ if and only if $\alpha_2 \leq 2-2\sqrt{\lambda}$
    \item $\mathcal{R}_2^{cf} \geq \mathcal{R}_2^g$ if and only if $\alpha_2 \geq 4\left(\frac12-\frac\lambda4\right)^2$
    \item $\rho^{cf} \in  \begin{cases}\left[\max\left\{\left(\frac{\alpha_2}{\lambda}\right)^2,\frac{4}{5}\right\},1\right] \quad &\text{if } \alpha_2 \leq \lambda \\
    \left(1,\min\left\{\left(\frac{\alpha_2}{\lambda}\right)^2,\frac{5}{4}\right\}\right] \quad &\text{otherwise}.
    \end{cases}$
    \end{itemize}
    \end{itemize}
\end{proposition}

\begin{figure}[t]
    \centering
        \includegraphics[scale=0.4]{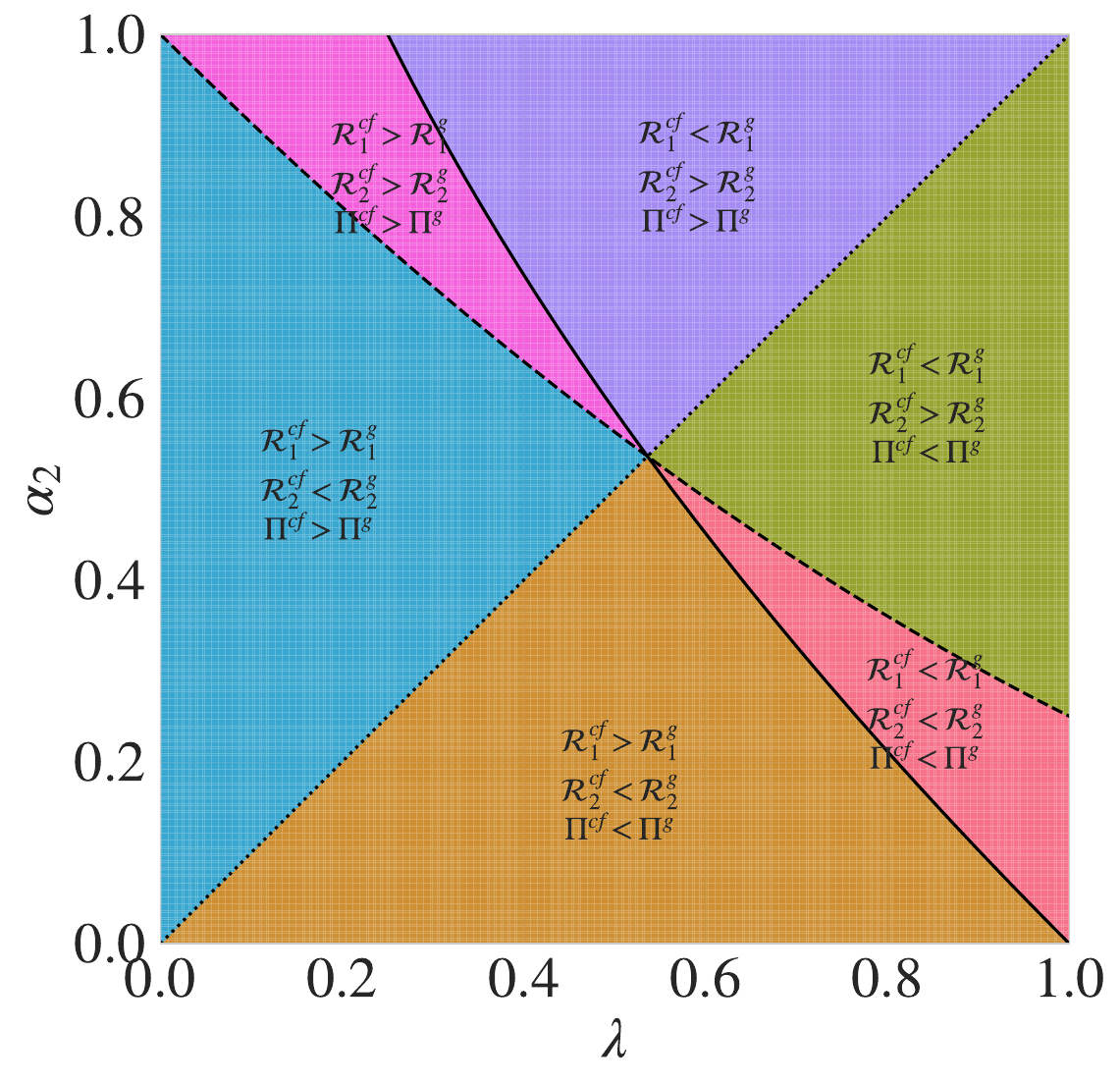}
    \caption{\centering Impact of disclosing certification information on seller and platform revenue in the Certification-First LNE, as a function of $\alpha_2$ and $\lambda$.}
    \label{fig:alpha-comparison}
\end{figure}

We defer the proof of the result to Appendix \ref{apx:metrics-special-case-q2}. While displaying ratings information increases the revenue of both sellers on the platform, \Cref{prop:metrics-special-case-q2} establishes not only that disclosing a third attribute has a differing impact on sellers, but also that this impact is both equilibrium and primitive-dependent. In particular, the total fraction of demand served {\it decreases} in the Ratings-First LNE, since customers for whom $w_c \in \left[\frac12-\frac\lambda4,\frac12-\frac{\alpha_1}{4}\right)$ can now no longer afford the lower-priced seller. From a revenue perspective, the Ratings-First LNE is {\it worse} for seller 1, since some customers who were previously ratings-first have now been converted to certification-first customers, resulting in fewer customers purchasing from seller 1 overall. Seller 2, on the other hand, benefits in this equilibrium, both due to the fact that it serves more customers, as well as the resulting equilibrium price increase. Our results establish that the net result of this is negative for the platform, since $\rho^{rf} \leq 1$; however, this reduction in revenue remains bounded, with at least 80\% of the granular LNE revenue guaranteed in the worst case.

 The impact on revenue in the Certification-First LNE is far more nuanced, however, as illustrated in \Cref{fig:alpha-comparison}. Namely, \Cref{prop:metrics-special-case-q2} establishes that {\it the platform's total revenue improves {only if} the fraction of certification-first customers exceeds the fraction of ratings-first customers {\it before} the disclosure of this new attribute.} The key managerial insight associated with this finding is that the disclosure of a differentiating attribute that only a minority of customers value actually {\it hurts} the platform; more information is not always better. Moreover, even in settings where a majority of customers place high value on the additional attribute, this gain is capped at 25\% higher than the revenue achieved before its introduction. 
 
 With respect to the seller-specific impact of this design decision, our results show that, when the fraction of certification-first customers is small, seller 1 benefits from the disclosure of the additional attribute. This holds despite the fact that seller 1 is the undercutter in this regime, since undercutting at a slightly lower price allows seller 1 to gain access to both certification-first and price-first customers, in addition to ratings-first customers. Seller 2, on the other hand, benefits from disclosing the additional attribute if the pool of certification-first customers (the only customers to whom seller 2 sells in this LNE) is large enough. We moreover note the existence of regions in which disclosing the certification attribute (i) uniformly improves revenues for both sellers, and the platform (upper-left region of \Cref{fig:alpha-comparison}), and (ii) uniformly decreases revenues for both sellers, and therefore the platform (lower-right region of \Cref{fig:alpha-comparison}). 

Finally, our results also highlight the importance of the platform's ability to control {\it which} equilibrium arises, in order to see the revenue gains from disclosing this additional attribute; in practice, this can be done via, e.g., price recommendations for sellers.

{
\subsection{Simulation-based evaluation of platform design choices}

{We next extend our analysis to a more complex setting, where estimating sellers' revenue functions is significantly more challenging. Using the simulation approach developed in \Cref{sec:numerical-results}, we study how ratings disclosure and default ranking affect platform revenue in this richer competitive environment. We moreover examine two additional dimensions of platform design: ratings granularity and the correlation between a newly disclosed attribute and existing seller ratings.}

\subsubsection{Setup.} We consider a setting with $N = 10$ sellers, each with an average rating $Q_i \sim \textup{Unif}[0,1]$, for $i \in [N]$. Each customer has a willingness-to-pay $w_c \sim\textup{Unif}[0,1]$, as well as a ratings floor $f\sim \textup{Unif}[0,1]$ below which they will not consider a seller. We initially assume that price and rating are the only attributes available to customers, and let $\lambda$ denote the fraction of ratings-first customers. 

To evaluate the impact of various design choices, we implement the equilibrium simulator described in \Cref{sec:numerical-results} (formally presented in \Cref{alg:sim}), using the same parameter instantiations as in our experiments validating convergence of gradient-driven price-update dynamics (see Appendix \ref{sec:empirics} for details). We conduct 100 replications for each experiment; in each replication, we average sellers' revenues over the last 10\% of the simulation horizon (5,000 periods), and subsequently average over the 100 replications. We use the terms ``profit'' and ``revenue'' interchangeably in the remainder of this section.

\subsubsection{Designing seller ratings.}

We first test the robustness of the theoretical insights derived in \Cref{ssec:avg-rating-impact} for this more complex setting. 

\textbf{Impact of displaying seller ratings.} \Cref{fig:avg-rating-impact} shows the impact of disclosing sellers' ratings on the platform's total revenue, as a function of the fraction of ratings-first customers $\lambda$. Overall, these results echo those in \Cref{ssec:avg-rating-impact}: as customers become less price-sensitive (i.e., $\lambda$ increases), the total revenue on the platform increases.  This underscores that the platform benefits from interventions that encourage customers to be more ratings-sensitive. (Note that, even when $\lambda = 0$, the total profit is strictly positive. This is because customers have a ratings floor $f$ according to which they filter sellers. As a result, with strictly positive probability the highest-rated seller is the only seller in a customer's consideration set, and therefore makes positive revenue despite not having the lowest price.)

\begin{figure}[t]
    \centering
    \includegraphics[scale=0.4]{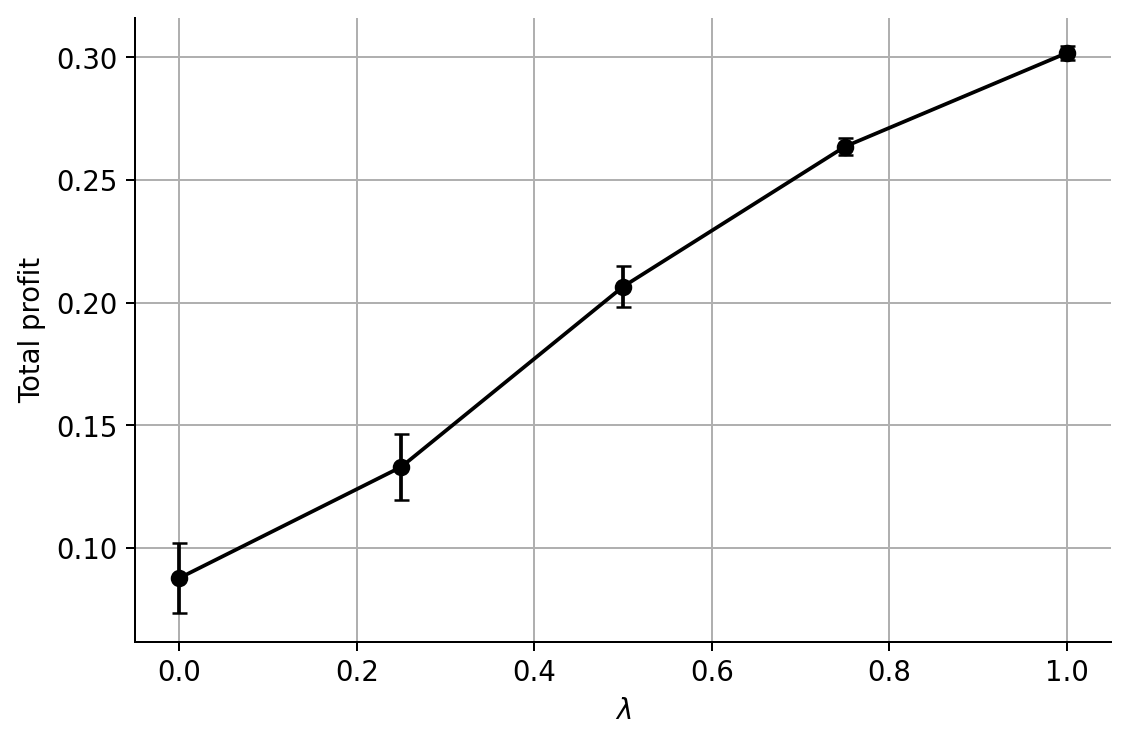}
    \caption{\centering Dependence of platform profit on fraction of ratings-first customers $\lambda$.}
    \label{fig:avg-rating-impact}
\end{figure}

\textbf{Impact of default ranking.} A design intervention geared toward achieving this is the platform's default ranking; we study the impact of displaying sellers in decreasing order of rating next. As in \Cref{ssec:avg-rating-impact}, we assume that, if the platform ranks sellers in increasing order of price, the fraction of price-first customers increases by some $\delta > 0$. Conversely, if the platform ranks sellers in decreasing order of rating, the fraction of ratings-first customers increases by $\delta$. 

We plot the ratio of the platform's total revenue under the Price-First default ranking relative to its revenue under the Ratings-First default ranking in \Cref{fig:default-rating-impact}, for $\lambda \in \{0.25,0.5,0.75\}$, $\delta \in \{0.05,0.1,0.15\}$. Across all combinations of parameters, we observe that the platform's revenue is strictly lower under the Price-First default ranking (as shown by a profit ratio that is strictly less than 1). For fixed values of $\lambda$, this effect is more pronounced as $\delta$ increases, and customers' behavior is more strongly influenced by the default ranking (i.e., the position bias effect is stronger). Importantly, our results also show that the platform has the {\it most} to gain from a Ratings-First default ranking when $\lambda$ is small; indeed, the platform's revenue is approximately 20\% higher under the Ratings-First default ranking than under the Price-First default ranking when $\lambda = 0.25$ and $\delta = 0.05$. In contrast, when $\lambda = 0.75$, the Ratings-First default ranking improves upon the Price-First default ranking by less than 5\%. This is intuitive, as the large majority of customers are ratings-first, without the added position bias effect.

\begin{figure}[t]
    \centering
    \includegraphics[scale=0.17]{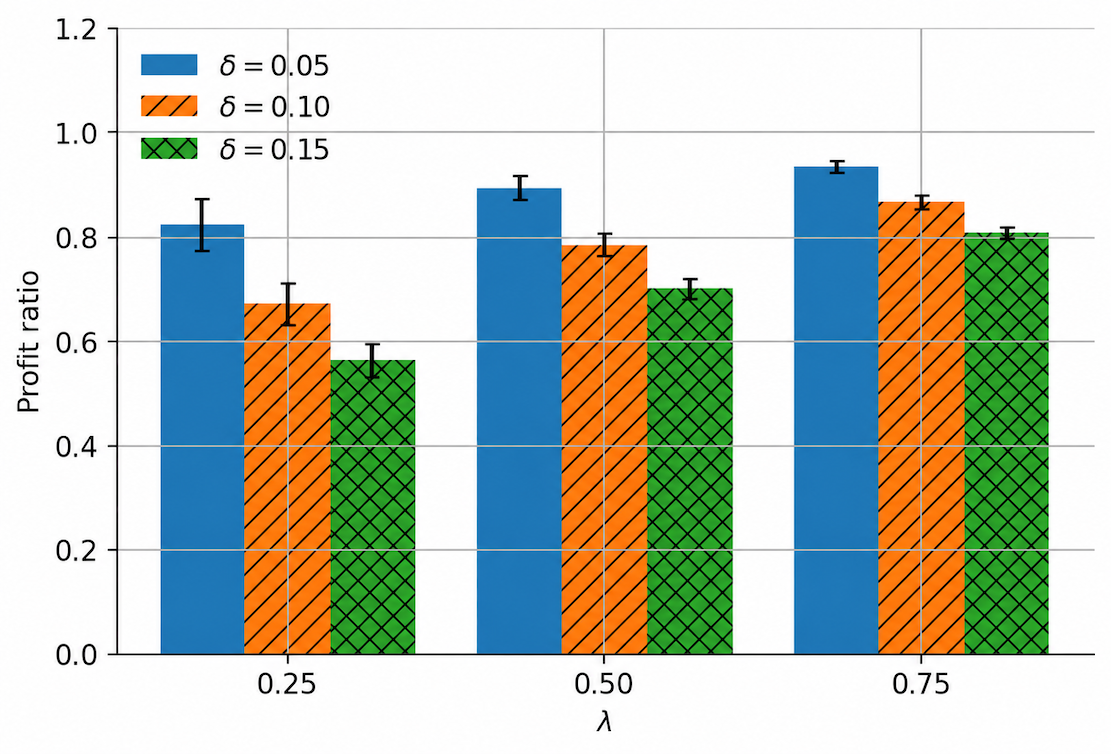}
    \caption{\centering {Impact of default ranking on platform profit. Here, we report the ratio of platform revenue under the Price-First default ranking to its revenue under the Ratings-First default ranking, for various values of $\lambda$ and $\delta$.}}
    \label{fig:default-rating-impact}
\end{figure}

\textbf{Impact of ratings granularity.} Another way to influence customers' ratings sensitivity is through the {\it granularity} of the ratings displayed. To study this design choice, we compare the platform's revenue when it discloses sellers' true average rating (referred to as the {\it continuous} granularity) to its revenue under the following two disclosure policies: (i) a {\it binary} disclosure policy, e.g., a ``High Quality'' badge for above-average sellers, and (ii) a {\it ternary} disclosure policy, e.g., a  ``Best Quality'' badge for the top 25\% of sellers, a ``Good Quality'' badge for sellers in the interquartile range, and nothing otherwise. Formally, under the binary disclosure policy, the platform discloses $Q_i' = \mathds{1}\{Q_i \geq 0.5\}$ for $i \in [N]$, and customers' ratings floors are similarly assumed to be truncated to $f' = \mathds{1}\{f \geq 0.5\}$. Under the ternary disclosure policy, the platform discloses $Q_i' = 0$ if $Q_i \leq 0.25$, {$Q_i' = 0.5$} if $Q_i \in (0.25,0.75]$, and $Q_i' = 1$ otherwise; we truncate customers' ratings floors analogously. 

In \Cref{fig:granularity-impact} we plot the platform's total revenue under the three different granularity levels, for {$\lambda \in \{0.25,0.5,0.75\}$}. For all values of $\lambda$, we observe that the platform's total revenue is increasing in the granularity of displayed ratings. Moreover, the profit impact of displaying continuous ratings is stark: in the worst case, the continuous granularity generates more than twice as much revenue as the binary granularity ($\lambda = 0.75$). The reason for this is that coarser ratings information limits differentiation between sellers. In particular, when ratings-first customers evaluate sellers according to rating, there are more ties between sellers as the displayed ratings become coarser; price then becomes the leading differentiator, yielding more intense price competition, and lower equilibrium prices. The end result of this is lower platform revenue. 

\begin{figure}[t]
    \centering
    \includegraphics[scale=0.13]{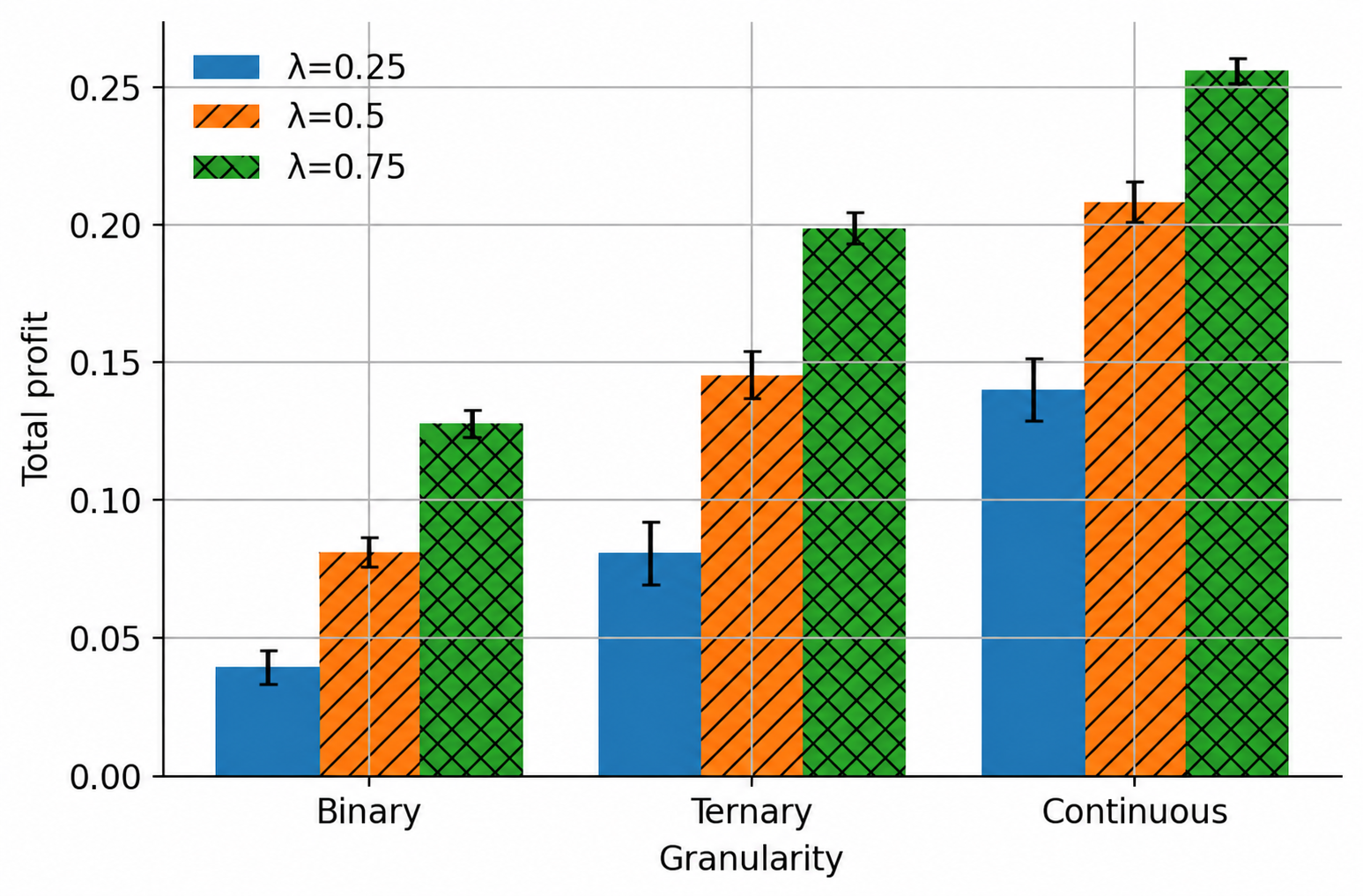}
    \caption{\centering {Impact of ratings granularity on platform profit, for various values of $\lambda$.}}
    \label{fig:granularity-impact}
\end{figure}

\subsubsection{Impact of displaying a third differentiating attribute.} We finally study how displaying a third differentiating attribute affects platform profit. To do so, we let $B_i$ denote the value of this third attribute for each seller $i \in [N]$, and consider three different cases:
\begin{enumerate}
\item $B_i$ is independent of $Q_i$ (e.g., promised delivery time):
\(B_i \sim \textup{Unif}[0,1]\)
\item $B_i$ is positively correlated with $Q_i$ (e.g., seller tenure on platform): 
\(B_i = Q_i + \epsilon, \ \epsilon \sim \mathcal{N}(0,0.1)\)
\item $B_i$ is negatively correlated with $Q_i$ (e.g., proportion of returns):
\(B_i = 1-Q_i + \epsilon, \ \epsilon \sim \mathcal{N}(0,0.1)\)
\end{enumerate}
We respectively let $\alpha_1$ and $\alpha_2$ be the fraction of ratings-first and price-first customers after the disclosure of the new attribute, with $\alpha_1 = \gamma\lambda$ and $\alpha_2 = \beta-\gamma\lambda$, for some $\gamma \in (0,1)$, $\beta \in (0,1)$. Then, the fraction of customers who evaluate the new attribute first is given by $1-\beta$. Finally, we assume that customers choose the second attribute according to which they evaluate sellers uniformly at random.

{
\begin{figure}[t]
    \centering

    \begin{subfigure}[b]{0.4\textwidth}
        \includegraphics[width=\textwidth]{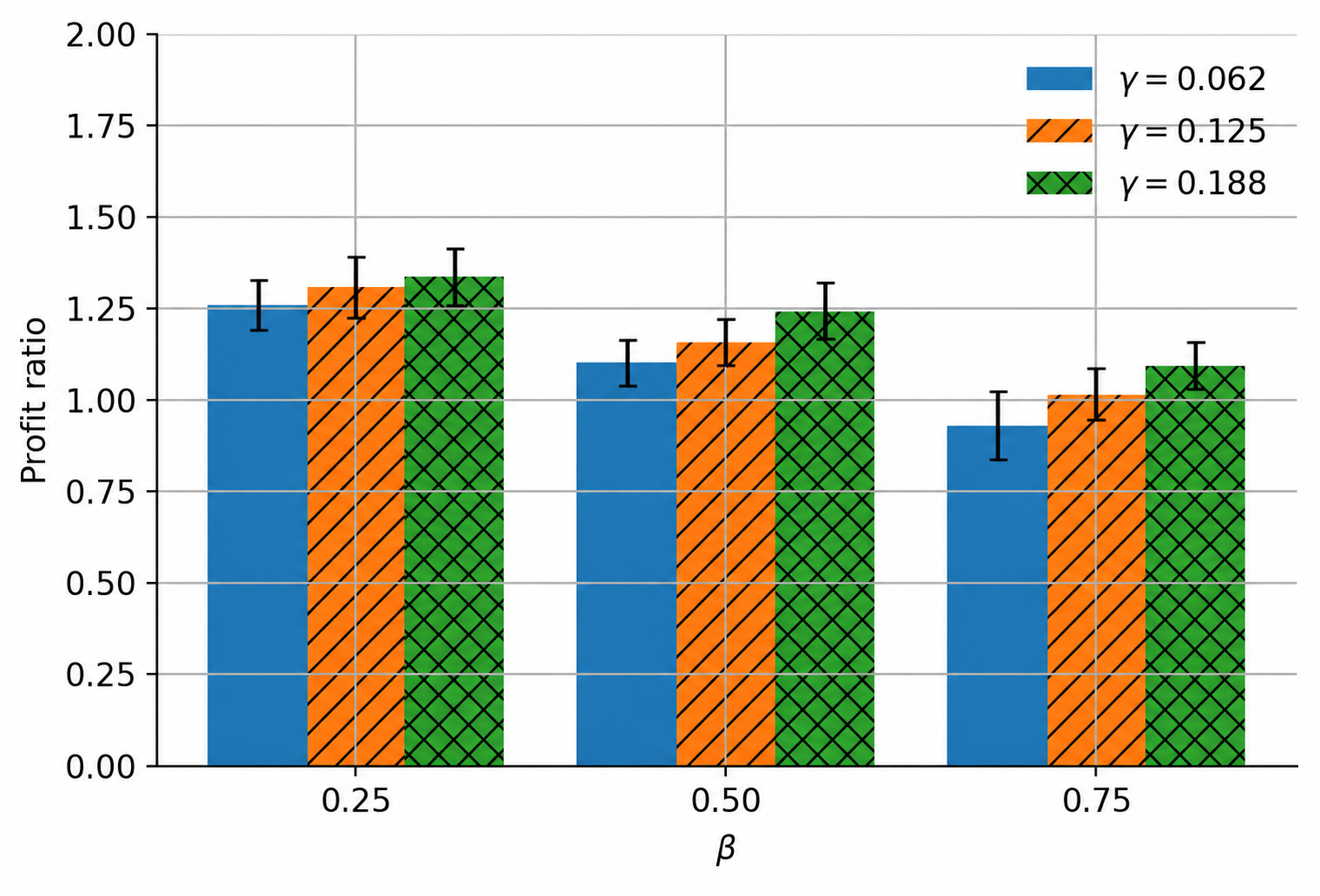}
        \caption{{Positively correlated third attribute}}
        \label{fig:plot1}
    \end{subfigure}
    \hfill
    \begin{subfigure}[b]{0.4\textwidth}
        \includegraphics[width=\textwidth]{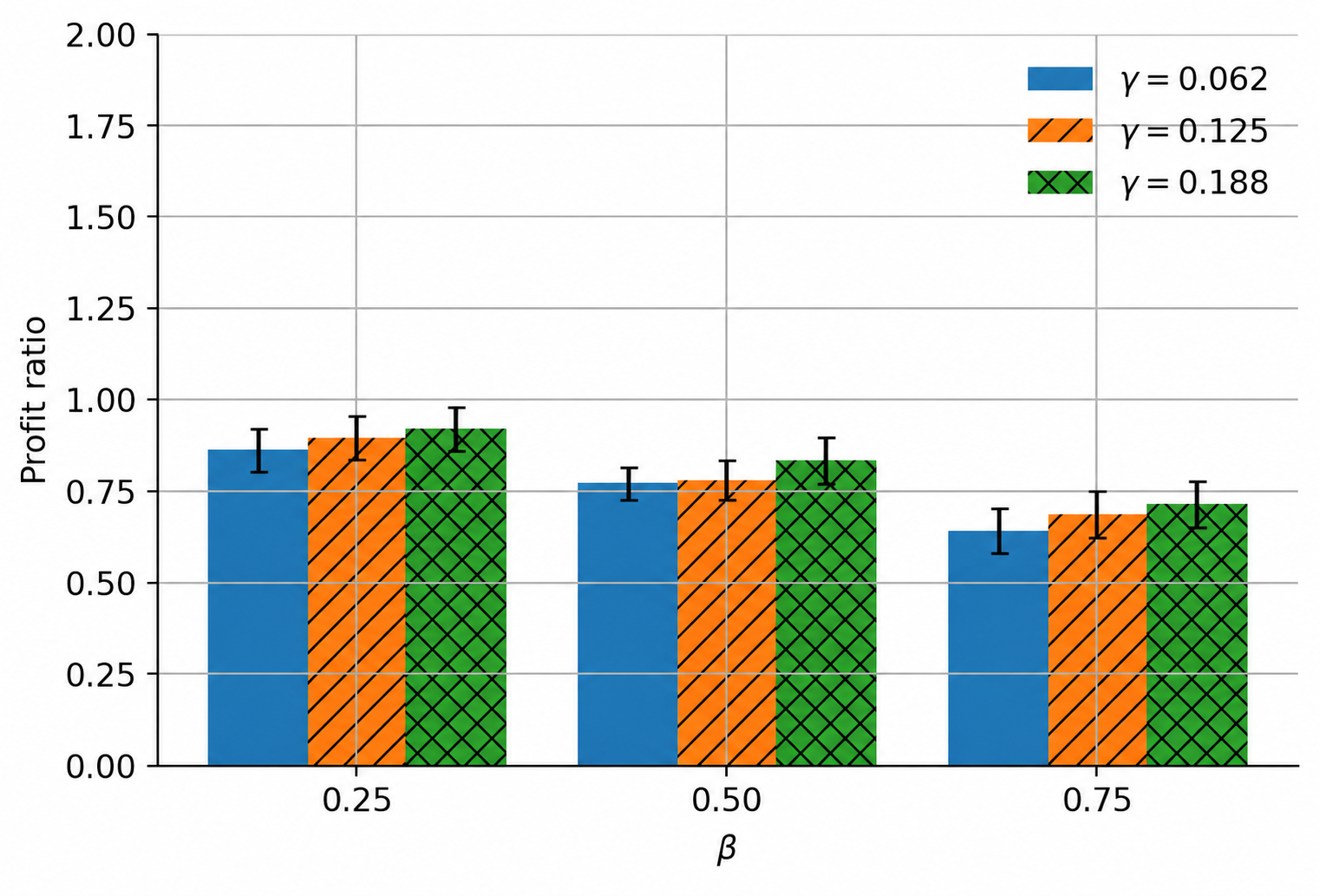}
        \caption{{Independent third attribute}}
        \label{fig:plot2}
    \end{subfigure}

    \vspace{1em} 

    \begin{subfigure}[b]{0.4\textwidth}
        \includegraphics[width=\textwidth]{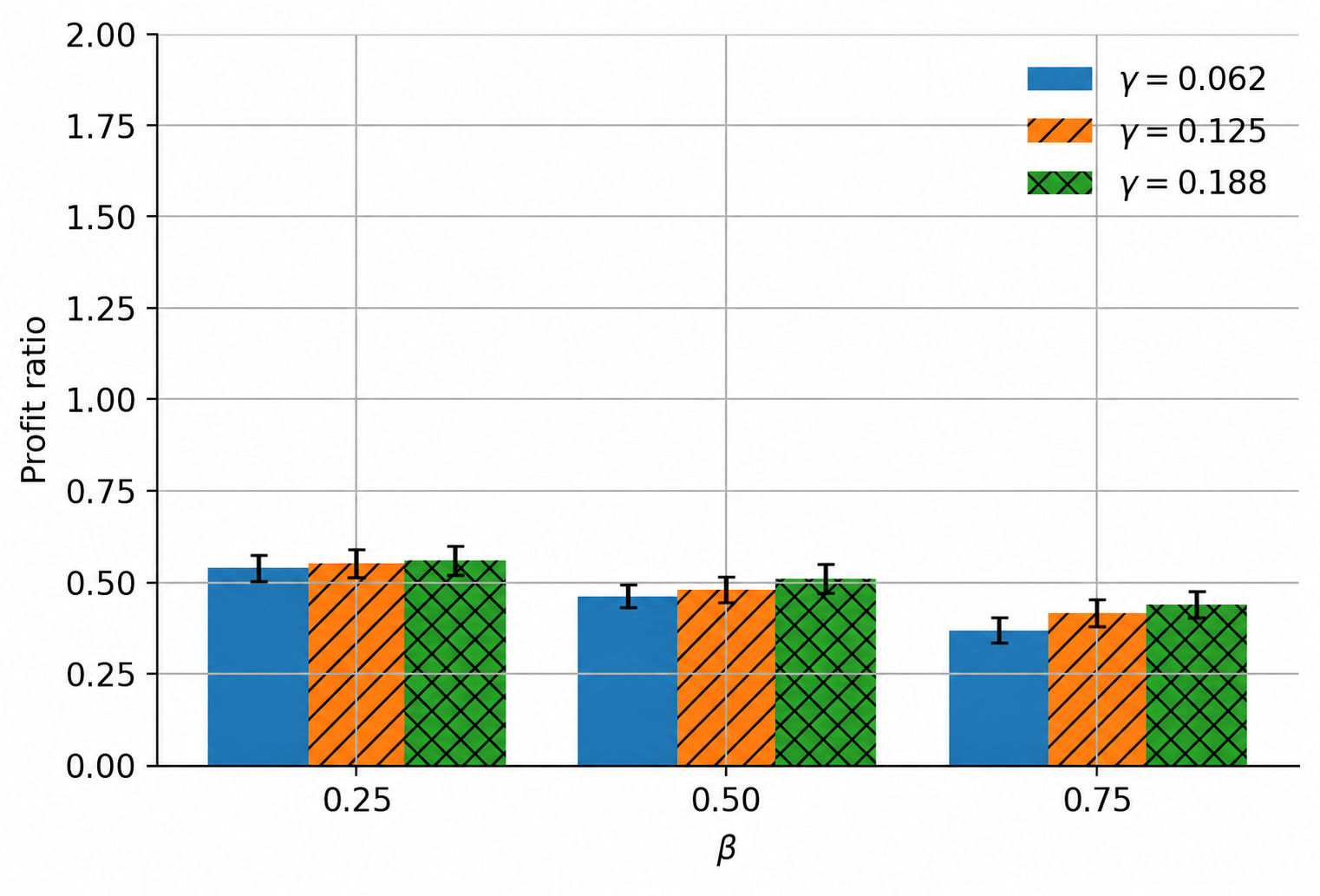}
        \caption{{Negatively correlated third attribute}}
        \label{fig:plot3}
    \end{subfigure}
    \caption{\centering {Impact of third differentiating attribute on platform revenue. Here, we let $\lambda = 0.5$, and plot the ratio of platform revenue after the disclosure of the third attribute to the revenue before its disclosure.}}
    \label{fig:third-attribute}
\end{figure}
}

We plot the ratio of total revenue after the introduction of the third attribute to the total revenue before its introduction in \Cref{fig:third-attribute}, letting $\lambda = 0.5$, and varying $\gamma \in \{0.062,0.125,0.188\}$ and $\beta \in \{0.25,0.5,0.75\}$. Our results highlight that the third attribute's correlation with seller ratings critically impacts the total revenue on the platform. In particular, \Cref{fig:plot1} shows that the platform's revenue increases when the third attribute is positively correlated with seller ratings. The reason for this is that approximately $\gamma\lambda + 1-\beta$ customers are effectively ratings-first when the third attribute is disclosed (given the alignment of the two attributes); this exceeds the original fraction of ratings-first customers, $\lambda$, as long as $1-\beta \geq \lambda(1-\gamma)$. At a high level, this condition states that disclosing the third attribute is beneficial as long as the fraction of price-first customers does not increase after the disclosure of the third attribute; this is a natural condition we expect to hold in practice. (Note that this condition is not satisfied for $\beta = 0.75$ and $\gamma = 0.062$, which is why we observe the decrease in relative profit.)

When the third attribute is independent or negatively correlated with seller ratings, however, disclosure decreases the total revenue, as can be seen in \Cref{fig:plot2,fig:plot3}. This impact is the most pronounced in the negatively correlated setting, with the platform's revenue decreasing by almost one half in the best case ($\beta = 0.25$, $\gamma = 0.188$). At a high level, this follows from the fact that high-rated sellers are {\it guaranteed} to be the least-preferred sellers for customers who evaluate the third attribute first. As a result, they have a strong incentive to undercut sellers who are strong according to this third attribute. The same holds in the reverse direction, leading to an overall decrease in equilibrium prices. When the third attribute is independent of seller ratings, this effect is attenuated, since there exist sample paths over which the third attribute is in line with seller ratings. However, the net effect on platform revenue remains negative, since differentiation decreases the monopolistic power of higher-rated sellers, who now have greater incentive to undercut in order to sell to lower-valuation customers, as discussed in the setting with negative correlations.

Finally, while $\gamma$ does not have a significant impact on platform revenue in each of the three settings, we do observe that, for a fixed value of $\gamma$, the platform's revenue decreases as $\beta$ increases. This is intuitive, given that the fraction of price-first customers is increasing in $\beta$.

Overall, these results further underscore the important managerial insight derived in \Cref{ssec:third-attribute}: that disclosing more information is not always better for the platform. Prior to doing so, the platform should carefully reason about how this attribute correlates with existing disclosed attributes, and whether disclosing it will reinforce certain sellers' market dominance, or promote price competition via differentiation.
}
}

%% file: final-parts/discussion.tex

\section{Conclusion}\label{sec:conclusion}

{
In this work we studied the problem of tractably predicting equilibrium outcomes of games of price competition on online platforms. Specifically, we considered an empirically validated model of customer behavior tailored to platforms' sorting and filtering functionalities. We first identified an interpretable sufficient condition --- {\it gradient dominance} --- which guarantees existence of local Nash equilibria. We moreover provided tractable procedures to compute equilibria under this condition. These findings address a critical need for platforms: the ability to efficiently evaluate the impact of various design interventions on customer behavior and the resulting pricing outcomes. {Finally, we used these tools to derive practical implications for platform design, characterizing how ratings disclosure, default rankings, and the disclosure of additional differentiating attributes affect equilibrium prices, platform revenue, and market penetration.}

Several future directions emerge from our work. While we presented strong evidence that gradient-driven pricing frequently converges to local Nash equilibria, an interesting future direction would be to study equilibrium computation and learning beyond the gradient dominance condition. Additionally, a natural next step would be to explore algorithms that provably yield convergence to other limit points of interest, such as global (rather than local) Nash equilibria. Finally, it would be interesting to leverage our framework to study settings in which certain non-price attributes are {endogenous}. In particular, while our model treats attributes such as average rating as fixed, sellers often make strategic investments to improve them. Analyzing the interplay between these decisions and pricing strategies is a promising future direction.
}

%% file: final-parts/msom24-appendix.tex
\OneAndAHalfSpacedXI

{
\section{Examples of CLC Choice Behavior}\label{ex:attributes}

The following example illustrates the wide range of behaviors encompassed by CLC choice.

\begin{example}
{\it Consider a set of third-party sellers listing athletic T-shirts for women, with attributes $\{\texttt{brand}, \texttt{average rating},  \texttt{number of reviews}, \texttt{price}\}$. The brand attribute takes values in the set $\{\texttt{Nike}, \texttt{Adidas}, \texttt{Puma}\}$; each T-shirt has an average rating in $[0,5]$; and the number of reviews lies in $\mathbb{N}$. The following three customers belong to the CLC class:
\begin{itemize}
\item {\bf A price-sensitive customer.} This customer is willing to pay at most \$25 and requires an average rating of at least 3. Of all such items, the customer chooses the cheapest one. If there are multiple lowest-priced items, she chooses the one with the highest rating, and breaks final ties by number of reviews.
\item {\bf A quality-sensitive customer.} This customer is willing to pay at most \$80 and requires an average rating of at least 4.5 with at least 100 reviews. Of all such items, she chooses the one with the highest average rating, breaking ties by the number of reviews, then by price. 
\item {\bf A loyal customer.} This customer is willing to pay \$100 and only considers \texttt{Nike} products. Of all such products, she purchases the highest-rated one, breaking ties according to the number of reviews, followed by price.
\end{itemize}
}
\end{example}
}

\section{\Cref{sec:pseudo-equilibria}: Additional Results and Omitted Proofs}\label{apx:structure-of-equilibria}

{
\subsection{On the non-existence of equilibria}\label{apx:non-existence-equilibria}

\begin{proposition}
Fix $\epsilon \in (0,1/2)$. Consider a duopoly wherein $\epsilon$ fraction of the population only purchases from seller 1, $\epsilon$ fraction of the population only purchases from seller 2, and $1-2\epsilon$ fraction of the population purchases from the cheapest seller, breaking ties uniformly at random. Suppose moreover that customers' willingness-to-pay $w_c \sim \textup{Unif}[0,1]$. In this setting, no local Nash equilibrium exists.
\end{proposition}

\begin{proof}{\it Proof.}
By \Cref{thm:general-LNE}, it suffices to analyze the VOPs associated with (i) seller 1 pricing above seller 2, and (ii) seller 2 pricing above seller 1.

Consider the case where seller 1 prices above seller 2. In this setting, seller 1 sets $p_1$ to maximize $R_1(p,0) = \epsilon \cdot p \cdot (1-p)$, whose maximum is attained at $p = 1/2$. Given $p_1 = 1/2$, the undercutting best-response of seller 2 is to set $p_2 = p_1 - \epsilon_0$, for arbitrarily small $\epsilon_0 > 0$, since there is an upward jump discontinuity at $p_1$. Since seller 2's undercutting optimum is not attained, there does not exist a LNE associated with this ordering. By symmetry, the same holds for seller 2 pricing above seller 1, which completes the proof of the claim.
\hfill\Halmos
\end{proof}
}

\subsection{Proof of \cref{prop:clc-choice-implies-psg}}\label{apx:psg-proof}

\begin{proof}{\it Proof.} {{We argue this at the customer choice level.} Consider a seller $i$ pricing at $p$ and suppose that a seller $j$ with $p_j < p$ changes its original price $p_j$ to some $p_j' < p$.  We show that this change does not impact whether or not seller $i$ is chosen by a customer; hence, it does not impact the revenue of seller $i$.

Note that the only {customers who could possibly impact seller $i$'s revenue under this change are those for whom both $i$ and $j$ were in the consideration set when $j$ priced at $p_j$}.
To see this, consider first a customer who did not include $i$ in her consideration set under $p_j$. In this case, $i$ will continue to be excluded under $p_j'$, and seller $i$ will continue to earn zero revenue from such a customer. Consider next a customer who included $i$ but not $j$ in her consideration set under $p_j$. Since $i$ was in the consideration set, the customer's willingness-to-pay exceeds $p$, and $p_j$ as a result. Hence, $j$ must have been excluded due to a non-price attribute, and $j$ will continue to be excluded due to a non-price attribute under $p_j'$. The choice outcome from the perspective of $i$ does not change in this case as well.

Consider finally customers who included both $i$ and $j$ in their consideration set under $p_j$. As argued above, since $i$ was in the consideration set under $p_j$, $w_c \geq p > p_j'$, and thus both $i$ and $j$ remain in the consideration set under $p_j'$. Now, under $j$'s old price, if the customer chose from the consideration set based on a non-price attribute ranked higher than price, she will continue to make the same choice under the new price. If she chose based on either the price attribute or some non-price attribute ranked lower than the price, then $i$ is not chosen before or after the price change, since its price is not the lowest in the consideration set under $p_j$ or $p_j'$. In summary, the outcome of whether or not $i$ is chosen remains unchanged, which proves the result.}\hfill\Halmos
\end{proof}

\subsection{Proof of \cref{thm:general-LNE}}\label{apx:proof-of-lne-result}

We first prove that the following property (which we refer to as {\it cross-monotonicity}) holds.
\begin{proposition}\label{prop:cross-monotonicity}
Under CLC choice, for all $i \in \cN$, $p_i \in \mathbb{R}$, $R_i(p_i,\bp_{-i})$ is non-decreasing in $\bp_{-i}$, i.e., if $\bp'_{-i}\mathbf{\geq} \bp_{-i}$ (component-wise), then $R_i(p_i,\mathbf{p}'_{-i})\geq R_i(p_i,\mathbf{p}_{-i})$.
\end{proposition}

\begin{proof}{\it Proof.}
Fix seller $i \in \cN$, and suppose there exists a seller $j\in \cN\setminus\{i\}$ that strictly increased its price to $p_j' > p_j$. Note that this change only affects a customer $c$ such that $\{i,j\} \in \mathcal{N}_c$. Consider any customer $c$ for which this holds.

If $c$ had not chosen $i$ under $p_j$, then $i$'s revenue from $c$ is 0, and can only weakly increase after a price change of $j$. Thus, it suffices to argue that if $c$ chooses $i$, then they continue to choose $i$ after an increase in $p_j$, implying a weak increase in revenue for $i$. 

If $c$ had previously chosen $i$ due to a non-price attribute ranked higher than price, then the price change of seller $j$ does not change her previous decision. Suppose now that $c$ chose $i$ based on either the price attribute, or an attribute ranked lower than the price. Then, it must be that $p_i \leq p_j < p_{j'}$. In this case, even after the price increase in $p_j$, $c$ would continue to choose $i$. This is because the customer cannot choose $j$ after this change since $p_i<p_j'$, and it cannot choose some other seller since if that were the case, then they would have chosen that seller under $p_j$ as well, contradicting the fact that $i$ was chosen under $p_j$. \hfill\Halmos
\end{proof}

With this result in hand, we prove \Cref{thm:general-LNE}, deferring the proofs of all auxiliary lemmas to the end of the section.
\begin{proof}{\it Proof of \cref{thm:general-LNE}.}
Suppose $\bp$ is a LNE, and let $\pi$ denote the corresponding ordering (breaking ties arbitrarily). 

Suppose first that all prices in $\bp$ are distinct. Then, these prices are necessarily valid with respect to $\pi$ (i.e., $p_{\pi(i)}$ is a local optimum in the undercutting region of $\pi(i)$, for all $i$). This follows from the pseudo-competitive property: for any $i$, $R(p, \bp_{-\pi(i)}^{\mathcal{S}_i}) = R(p, \bp_{-\pi(i)})$. Hence, the local optimality of a seller in the ordering is preserved even if the prices of lower-ranked sellers are reduced to $0$, implying the validity of these prices with respect to $\pi$. 

Hence, it remains to consider the case where at least two prices in $\bp$ are equal. Consider the first set of sellers with equal prices (i.e., those ranked highest in the ordering), and let $l+1$ be the size of this set. We denote these sellers by $\{j_0,\ldots,j_l\}$, with $j_0,\ldots,j_l$ following the same order as $\pi$. By the same argument as the one used when all prices in $\bp$ are distinct, the price of any seller ranked higher than $j_0$ is locally optimal with respect to $\bp$. We argue that the local optimality of $p_{j_0} = \ldots = p_{j_l}$ is not perturbed.

\noindent\textbf{Case 1: $p_{j_0}= p_{j_1}=\cdots= p_{j_l}=0.$}  For any $j \in \{j_0,\ldots,j_l\}$, $R_{j}(p_j,\bp_{-j}) = R_j(p_j,\bp_{-j}^{S_j}) = 0$. Thus, if 0 is a local optimum for $j$ over $[0,\pmax]$, it is necessarily a local optimum over $j$'s relevant undercutting region.

\noindent\textbf{Case 2: $p_{j_0} = p_{j_1}=\cdots=p_{j_l} > 0.$} We will argue that $p_{j_0}$ is locally optimal over its relevant undercutting region if any subset $\mathcal{S} \subseteq \{j_1,\cdots, j_l\}$ of sellers reduce their price to $0$. By symmetry, then, the same will hold true for any $j \in \{j_0,\cdots, j_l\}$. Repeating this argument for the next set of equally priced sellers below $\{j_0,\ldots,j_l\}$, until no equally priced sellers remain (or all sellers are priced at zero), we will obtain the theorem.

\cref{lem:rev-doesn't-change} plays a key role in proving this claim. It establishes that sellers $j_1,\ldots,j_l$ reducing their prices does not impact the expected revenue of $j_0$.
\begin{lemma}\label{lem:rev-doesn't-change}
$R_{j_{0}}(p,\bp_{-j_0}) = R_{j_{0}}(p,\bp'_{-j_0})$ for all $\bp'$ such that $p'_{j_i} \leq p_{j_i}$, $i \in \{1,\ldots, l\}$.
\end{lemma}

We leverage \cref{lem:rev-doesn't-change} to show that $p_{j_0}$ remains locally optimal under any price vector $\bp'$ such that $p_{j_i}' = 0$ for some subset of sellers $\mathcal{S} \subseteq \{j_1,\ldots,j_l\}$. \cref{lem:don't-undercut} establishes that it is never optimal for $j_0$ to $\epsilon$-undercut its original price $p_{j_0}$ under $\bp'$.

\begin{lemma}\label{lem:don't-undercut}
There exists $\epsilon_0 > 0$ such that, for all $0 < \epsilon \leq \epsilon_0$, $R_{j_0}(p_{j_0}-\epsilon, \bp'_{-j_0}) \leq R_{j_0}(p_{j_0}, \bp'_{-j_0})$.
\end{lemma}

We conclude the proof by establishing that it is never optimal for $j_0$ to $\epsilon$-overcut $p_{j_0}$ under $\bp'$.

\begin{lemma}\label{lem:don't-overcut}
There exists $\epsilon_0 > 0$ such that, for all $0 < \epsilon \leq \epsilon_0$, $R_{j_0}(p_{j_0}+\epsilon, \bp'_{-j_0}) \leq R_{j_0}(p_{j_0}, \bp'_{-j_0})$.
\end{lemma}

Putting Lemmas \ref{lem:don't-undercut} and \ref{lem:don't-overcut} together, the local optimality of $p_{j_0}$ is preserved under any such $\bp'$.  
\hfill\Halmos
\end{proof}

\begin{proof}{\it Proof of \Cref{lem:rev-doesn't-change}.}
Consider the subset $\mathcal{S} \subseteq \{j_1,\ldots,j_l\}$ for which $p_{j_{i}}' <p_{j_i}$. As argued in the proof of \cref{prop:cross-monotonicity}, for $j_0$ to be affected by this price decrease, it must be that there exists a customer $c$ such that $\{j_0,j\} \subseteq \mathcal{N}_c$ for some $j \in \mathcal{S}$. For any such customer $c$, it moreover must have been that $j_0$ and $j$ were ranked equally on all attributes up until price, and that $c$ had chosen $j_0$ (else, $c$ chose $j$ both before and after the price change, resulting in no change in revenue for $j_0$). 

Given this, since $j$ strictly decreased its price, all customers who chose $j_0$ based on an attribute $k$ ranked lower than price under $\bp$, now choose $j$ under $\bp'$. We claim that this occurs for a measure-zero set of customers.  To see this, suppose there existed $j \in \{j_1,\ldots,j_l\}$ such that a strictly positive mass of customers chose $j_0$ over $j$ based on an attribute ranked lower than price. Then, $j$ would stand to gain a discrete mass of customers by infinitesimally undercutting $p_{j_0}$, thus contradicting the local optimality of $p_j = p_{j_0}$.

Since only a measure-zero set of customers move their choice from $j_0$ to $j\in\mathcal{S}$ under this price change, the expected demand (and hence expected revenue) of $j_0$ is invariant to weakly lower price changes of $j_1,\ldots,j_l$.  \hfill\Halmos
\end{proof}

\begin{proof}{\it Proof of \Cref{lem:don't-undercut}.}
By local optimality of $p_{j_0}$ under $\bp$, there exists $\epsilon_0 > 0$ such that for $0 < \epsilon \leq \epsilon_0$: 
\begin{align}
0\leq R_{j_0}(p_{j_0},\bp_{-j_0}) - R_{j_0}(p_{j_0}-\epsilon, \bp_{-j_0}) &\leq R_{j_0}(p_{j_0},\bp_{-j_0}) -   R_{j_0}(p_{j_0}-\epsilon, \bp'_{-j_0}) \\ &= R_{j_0}(p_{j_0}, \bp'_{-j_0}) - R_{j_0}(p_{j_0}-\epsilon, \bp'_{-j_0}).\notag 
\end{align}
Here, the second inequality follows from cross-monotonicity (\cref{prop:cross-monotonicity}), and the equality follows from \cref{lem:rev-doesn't-change}. \hfill\Halmos
\end{proof}

\begin{proof}{\it Proof of \Cref{lem:don't-overcut}.}
By local optimality of $p_{j_0}$ under $\bp$, there exists $\epsilon_0 > 0$ such that, for all $0 < \epsilon \leq \epsilon_0$:
\begin{align}
0\leq R_{j_0}(p_{j_0},\bp_{-j_0}) - R_{j_0}(p_{j_0}+\epsilon, \bp_{-j_0})&= R_{j_0}(p_{j_0},\bp_{-j_0}) -   R_{j_0}(p_{j_0}+\epsilon, \bp'_{-j_0})\\& = R_{j_0}(p_{j_0}, \bp'_{-j_0}) - R_{j_0}(p_{j_0}+\epsilon, \bp'_{-j_0}). \notag 
\end{align}
Here, the first equality follows from the pseudo-competitive property, and the second follows from \cref{lem:rev-doesn't-change}.  Thus, increasing its price is not locally beneficial to $j_0$ under price vector $\bp'$.\hfill\Halmos
\end{proof}

{
\subsection{Unimodality: Special cases}\label{apx:unimodality}

Below, we provide a sufficient condition --- a competitive extension of the notion of {\it regularity} in the monopolist pricing literature --- for Assumption \ref{asp:unimodality} to hold.

Let $\mathcal{C}_i(p, \mathbf{p}_{-i}^{\cS})$ denote the probability that a customer chooses seller $i$ at price $p$, given the prices of other sellers in set $\cS \subseteq \cN \setminus \{i\}$. Moreover, let $c_i(p, \mathbf{p}_{-i}^{\cS}) = \frac{\partial \mathcal{C}_i(p, \mathbf{p}_{-i}^{\cS})}{\partial p}$. We say that $\mathcal{C}_i(p, \mathbf{p}_{-i}^{\cS})$ is {\it regular} if $\frac{\mathcal{C}_i(p, \mathbf{p}_{-i}^{\cS})}{c_i(p, \mathbf{p}_{-i}^{\cS})}+p$ exists and is increasing in $p$, for all $\cS \subseteq \cN\setminus \{i\}$, and $\mathbf{p}_{-i}^{\cS}$. 

\begin{proposition}\label{prop:choice-reg}
 Suppose $\mathcal{C}_i(p,\mathbf{p}_{-i}^{\cS})$ is {regular}, for all $i \in \cN$. Then, Assumption \ref{asp:unimodality} holds. 
\end{proposition}

We will use \Cref{prop:choice-reg} to show that the two special cases we present satisfy Assumption \ref{asp:unimodality}. We defer its proof to Appendix \ref{apx:choice-reg}.

In the remainder of this work, we let $F$ denote the cumulative distribution function of customers' willingness-to-pay, and $f$ the corresponding density. For ease of notation we define $\bar{F}(\cdot) = 1-F(\cdot)$.

\begin{proposition}\label{prop:gen-asp1}
Consider a setting in which (i) there is a mass $\lambda \in [0,1)$ of BICA-type customers, (ii) each seller $i \in \cN$ has $\gamma_i > 0$ mass of loyal customers, and (iii) the remaining mass of customers are greedy, breaking ties according to quality. If customers' willingness-to-pay distribution is independent of other model primitives, and moreover has a monotone hazard rate\footnote{A distribution has a {monotone hazard rate} if its hazard rate $h(p) = \frac{f(p)}{1-F(p)}$ is non-decreasing for all $p$, where $f(\cdot)$ and $F(\cdot)$ are respectively used to denote the probability density and cumulative distribution functions \citep{hartline2013mechanism}.} (MHR), then $\mathcal{C}_i(p,\bp_{-i}^{\cS})$ is regular for all $i \in \cN$. Therefore, Assumption \ref{asp:unimodality} holds. 
\end{proposition}

Examples of distributions with monotone hazard rate include the uniform distribution, the exponential distribution, and the Weibull distribution with shape $k \geq 1$. 

\begin{proof}{\it Proof.}
For ease of notation, let $\gamma = \sum_{i=1}^N \gamma_i$. Then, the mass of greedy customers is given by $1-\gamma-\lambda$. 
Now, fix seller $i \in \cN$, and suppose $p_i = p$. We consider two cases.

\noindent\textbf{Case 1:} There exists $k \in \cN$ such that $p_k \leq p$ and $q_k > q_i$ in $\bp_{-i}^{\cS}$. In this case, no BICA customer chooses $i$, since anyone who can afford $p$ can also afford $p_k$, and will choose the higher-quality seller $k$ as a result. Additionally, no greedy customer will choose $i$. Therefore, 
$\mathcal{C}_i(p,\bp_{-i}^{\cS}) = \gamma_i\bar{F}(p)$, since only seller $i$'s loyal customers purchase from $i$. 
We now verify that the regularity condition holds. We have:
\begin{align*}
\frac{\mathcal{C}_i(p,\bp_{-i}^{\cS})}{c_i(p,\bp_{-i}^{\cS})} + p = -\frac{\bar{F}(p)}{f(p)}+p,
\end{align*}
which is increasing by the MHR assumption.

\noindent\textbf{Case 2:}  $p_k > p$ for all $k \in \cN$ such that $q_k > q_i$ in $\bp_{-i}^{\cS}$. In this case, the BICA-type customers who choose $i$ are those for whom the prices of the higher-quality sellers exceed their willingness-to-pay. We further consider the following cases:
\begin{enumerate}
\item $p > 0$ and there exists $k \in \cN$ such that $p_k < p$: In this case, no greedy customer chooses $i$. Therefore:
\begin{align}\label{eq:case2a-gen}
\mathcal{C}_i(p,\bp_{-i}^{\cS}) &= \gamma_i\bar{F}(p) + \lambda \mathbb{P}(p \leq w_c < \min_{k\in\cN:q_k > q_i} p_k)= \bar{F}(p)\left(\lambda + \gamma_i \right)-\lambda \bar{F}(\min_{k\in\cN:q_k > q_i} p_k)
\end{align}
Therefore: 
\begin{align}\label{eq:regular}
&\frac{\mathcal{C}_i(p,\bp_{-i}^{\cS})}{c_i(p,\bp_{-i}^{\cS})} + p = -\frac{\bar{F}(p)\left(\lambda + \gamma_i \right)-\lambda \bar{F}(\min_{k\in\cN:q_k > q_i} p_k)}{f(p)\left(\lambda + \gamma_i \right)} + p.
\end{align}
We prove that the right-hand side of \eqref{eq:regular} is increasing in $p$, which will establish the regularity of $\mathcal{C}_i$. For any interval over which $f(p)$ is non-decreasing, the right-hand side of \eqref{eq:regular} is increasing, since {$\bar{F}(p)$} is decreasing. For any interval over which $f(p)$ is decreasing, the right-hand side of \eqref{eq:regular} is also increasing, since $\frac{\bar{F}(p)}{f(p)}$ is non-increasing by the MHR assumption, and $\frac{1}{f(p)}$ is non-decreasing.

\item $\{p = 0\}$, or $\{p > 0$ and $p_k \geq p$ for all $k \in \cN\}$: In this case, all greedy customers choose $i$ since it is both the cheapest and the highest-quality seller. Adding the contribution of these greedy customers to \eqref{eq:case2a-gen}, we have:
\begin{align}\label{eq:case2b-gen}
\mathcal{C}_i(p,\bp_{-i}^{\cS}) &= \bar{F}(p)\left(\lambda + \gamma_i + (1-\gamma-\lambda)\right)-\lambda \bar{F}(\min_{k\in\cN:q_k > q_i} p_k) = \bar{F}(p)\left(\gamma_i + 1-\gamma\right)-\lambda \bar{F}(\min_{k\in\cN:q_k > q_i} p_k).
\end{align}
Regularity then follows from the same arguments as above.
\end{enumerate}
\hfill\Halmos
\end{proof}

We now provide sufficient conditions on the CLC model for Assumption \ref{asp:unimodality} to hold in a duopoly. For any customer $c$, let $\nopriceconsiderationset$ denote the consideration set formed {\it without} accounting for {the willingness-to-pay constraint}, i.e., $\nopriceconsiderationset = \{i\in\mathcal{N} : (v^1_i,\ldots, v^K_i)\in  \attributecriteria{c}\}.$
For $i,j\in\mathcal{N}$, conditional on the event $\nopriceconsiderationset = \{1,2\}$, we use $\npaplus{i}{j}$ to denote the probability that a customer strictly prefers $i$ to $j$ due to a non-price attribute ranked higher than price. Moreover, we let $\overline{\lambda}_{i=j}$ denote the probability that a customer equally prefers $i$ and $j$ on all attributes until price.

\begin{proposition}\label{prop:duopoly-asp1}
Consider a duopoly, and suppose each seller $i \in \cN$ has a positive fraction of customers who strictly prefer $i$ to $j$ based on an attribute ranked higher than price. Suppose moreover that customers' willingness-to-pay distribution is independent of all other CLC choice primitives, and has a monotone hazard rate. Then, $\mathcal{C}_i(p,\bp_{-i}^{\cS})$ is regular for all $i \in \cN$. Therefore, Assumption \ref{asp:unimodality} holds.
\end{proposition}

\begin{proof}{\it Proof.}
Fix $i \in \cN$, and suppose $p_i = p$. Moreover, for ease of notation, let $j = -i$. We derive a closed-form expression for $\mathcal{C}_i$ by considering the possible revenues obtained for each realization of $\nopriceconsiderationset$.

\noindent\textbf{Case 1:} $\nopriceconsiderationset = \{i\}$. In this case, seller $i$ obtains revenue from all customers such that $w_c \geq p$.

\noindent\textbf{Case 2:} $\nopriceconsiderationset = \{j\}$. In this case, seller $i$ obtains no revenue.

\noindent\textbf{Case 3:} $\nopriceconsiderationset = \{1,2\}$. In this case, customers purchase from $i$ only if $w_c \geq p$. Of those customers, seller $i$ obtains revenue from all customers who strictly prefer $i$ based on an attribute ranked higher than price. Consider now the following sub-cases for the other sources of revenue:
\begin{enumerate}
\item $p > p_j$: In this case, seller $i$ receives no additional revenue. This follows from the fact that customers who equally prefer $i$ and $j$ on all attributes up until price will break ties in favor of $j$. Moreover, customers who strictly prefer $j$ based on an attribute ranked before price, and for whom $w_c < p_j$, will also have $w_c < p$.
\item $p < p_j$: In this case, seller $i$ additionally receives revenue from (i) customers who equally prefer $i$ and $j$ on all attributes up until price, and (ii) customers who strictly prefer $j$ based on an attribute ranked before price, but for whom $p \leq w_c < p_j$.
\item $p = p_j$: This would only occur if $j \not\in \mathcal{S}$, since if $j \in \mathcal{S}$, we assume $p < \nearestminprice = p_j$, by definition. Since $j \not\in \mathcal{S}$, this implies that $p_j = 0$ in $\bp_{-i}^{\cS}$. This therefore implies that $p = 0$; we omit this case from consideration as a result, since $i$ can always make strictly positive revenue from the customers who strictly prefer $i$ to $j$ based on an attribute ranked higher than price. Therefore, the gradient cannot be nonpositive at zero. 
\end{enumerate}

Putting this case-by-case analysis together, we finally obtain the following two cases for the closed-form expression of $\mathcal{C}_i$.

\noindent\textbf{Case a:} $p > p_j$. Then,
\begin{align}\label{eq:choice-prob-duopoly-1}
&\mathcal{C}_i(p,\bp_{-i}^{\cS}) = \mathbb{P}(\nopriceconsiderationset = \{i\}) \bar{F}(p) + \mathbb{P}(\nopriceconsiderationset = \{1,2\})\npaplus{i}{j}\bar{F}(p) 
\implies \frac{\mathcal{C}_i(p,\bp_{-i}^{\cS})}{c_i(p,\bp_{-i}^{\cS})} + p = -\frac{\bar{F}(p)}{f(p)}+p.
\end{align}
Since $F$ is MHR by assumption, $\frac{\bar{F}(p)}{f(p)}$ is non-increasing. Plugging this fact into the above, we obtain that $\mathcal{C}_i$ is regular in this case.

\noindent\textbf{Case b:} $p < p_j$. Then,
\begin{align}\label{eq:choice-prob-duopoly-2}
\mathcal{C}_i(p,\bp_{-i}^{\cS}) &= \mathbb{P}(\nopriceconsiderationset = \{i\}) \bar{F}(p) + \mathbb{P}(\nopriceconsiderationset = \{1,2\})\left(\bar{F}(p)\left(\npaplus{i}{j} + \overline{\lambda}_{i=j}\right) + \npaplus{j}{i}\left(\bar{F}(p)-\bar{F}(p_j)\right)\right) \notag \\
&= \mathbb{P}(\nopriceconsiderationset = \{i\}) \bar{F}(p) + \mathbb{P}(\nopriceconsiderationset = \{1,2\})\left(\bar{F}(p) - \npaplus{j}{i}\bar{F}(p_j)\right) \notag \\
&= \bar{F}(p)\left(\mathbb{P}\left(\nopriceconsiderationset = \{i\}\right) +  \mathbb{P}\left(\nopriceconsiderationset = \{1,2\}\right)\right) - \mathbb{P}\left(\nopriceconsiderationset = \{1,2\}\right)\bar{F}(p_j)\npaplus{j}{i},
\end{align}
where the second equality follows from the fact that $\npaplus{i}{j}+\overline{\lambda}_{i=j}+\npaplus{j}{i} = 1$.

To see that $\mathcal{C}_i$ is regular, note that:
\begin{align}\label{eq:duopoly-reg}
\frac{\mathcal{C}_i(p,\bp_{-i}^{\cS})}{c_i(p,\bp_{-i}^{\cS})} + p = p-\frac{\bar{F}(p)\left(\mathbb{P}\left(\nopriceconsiderationset = \{i\}\right) +  \mathbb{P}\left(\nopriceconsiderationset = \{1,2\}\right)\right) - \mathbb{P}\left(\nopriceconsiderationset = \{1,2\}\right)\bar{F}(p_j)\npaplus{j}{i}}{f(p)\left(\mathbb{P}\left(\nopriceconsiderationset = \{i\}\right) +  \mathbb{P}\left(\nopriceconsiderationset = \{1,2\}\right)\right)}.
\end{align}
For any interval over which $f(p)$ is non-decreasing, the right-hand side of the above is increasing, since $\bar{F}(p)\left(\mathbb{P}\left(\nopriceconsiderationset = \{i\}\right) +  \mathbb{P}\left(\nopriceconsiderationset = \{1,2\}\right)\right) - \mathbb{P}\left(\nopriceconsiderationset = \{1,2\}\right)\bar{F}(p_j)\npaplus{j}{i}$ is decreasing. On the other hand, for any interval over which $f(p)$ is increasing, $\bar{F}(p)/f(p)$ is non-increasing due to the MHR assumption; moreover, $\frac{\mathbb{P}\left(\nopriceconsiderationset = \{1,2\}\right)\bar{F}(p_j)\npaplus{j}{i}}{f(p)\left(\mathbb{P}\left(\nopriceconsiderationset = \{i\}\right) +  \mathbb{P}\left(\nopriceconsiderationset = \{1,2\}\right)\right)}$ is decreasing. Plugging these facts in \eqref{eq:duopoly-reg}, we obtain that $\mathcal{C}_i$ is regular. \hfill\Halmos
\end{proof}

\subsubsection{Proof of \Cref{prop:choice-reg}}\label{apx:choice-reg}
\begin{proof}{\it Proof.}
By definition, $R_i(p,\bp_{-i}^{\cS}) = p \cdot \mathcal{C}_i(p, \mathbf{p}_{-i}^{\cS})$. Therefore, $g_i(p,\bp_{-i}^{\cS})=\mathcal{C}_i(p, \mathbf{p}_{-i}^{\cS}) + p\cdot c_i(p, \mathbf{p}_{-i}^{\cS})$. Suppose now that $g_i(p,\bp_{-i}^{\cS}) \leq 0$ for some $p \in [0, \nearestminprice)$. Then, we must show that $g_i(p',\bp_{-i}^{\cS})< 0$ for all $p'\in (p,\nearestminprice)$. This latter inequality holds if, for all $p'\in (p,\nearestminprice)$:
\begin{align*}
\mathcal{C}_i(p', \mathbf{p}_{-i}^{\cS}) + p'\cdot c_i(p', \mathbf{p}_{-i}^{\cS}) < 0 \iff \frac{\mathcal{C}_i(p', \mathbf{p}_{-i}^{\cS})}{c_i(p', \mathbf{p}_{-i}^{\cS})} + p' > 0,
\end{align*}
where the implication follows from the fact that a price increase can only decrease the likelihood that a customer purchases from seller $i$, therefore $c_i(p', \mathbf{p}_{-i}^{\cS}) \leq 0$. Since $\mathcal{C}_i(p,\bp_{-i}^{\cS})$ is regular and $p' > p$, we have:
\begin{align*}
\frac{\mathcal{C}_i(p', \mathbf{p}_{-i}^{\cS})}{c_i(p', \mathbf{p}_{-i}^{\cS})} + p' > \frac{\mathcal{C}_i(p, \mathbf{p}_{-i}^{\cS})}{c_i(p, \mathbf{p}_{-i}^{\cS})} + p \geq 0,
\end{align*}
which proves the claim.\hfill\Halmos
\end{proof}
}

\subsection{Proof of \cref{cor:eq-easy}}\label{apx:efficient-comp}
\begin{proof}{\it Proof.}
The proof is constructive. Let $\pi$ denote an arbitrary ordering over sellers.

\textbf{Initialization.} Fix an arbitrary $\pi$, mark it as inspected, and proceed to Phase 1 below.

\textbf{Phase 1: Valid order-price pair computation.} Compute the VOP associated with $\pi$ (see \cref{def:valid-order}). If a corresponding VOP exists, let $\bp$ denote the resulting price vector, and proceed to Phase 2. Else, choose a new, uninspected ordering. Mark this new ordering as inspected, and go back to the beginning of Phase 1.

\textbf{Phase 2: LNE / GNE verification.} For $i \in [N]$, compute: $$\max\left\{\sup_{p \in [p_{\pi(i)}, p_{\pi(i-1)}]}R_{\pi(i)}(p,\bp_{-\pi(i)}),\sup_{p \in [p_{\pi(i+1)}, p_{\pi(i)}]}R_{\pi(i)}(p,\bp_{-\pi(i)})\right\}.$$
If this quantity strictly exceeds $R_{\pi(i)}(\bp)$ for some $i \in [N]$, then $\bp$ is not a LNE (and hence, not a GNE); else, it is a LNE.  

If $\bp$ is a LNE, verify that it is a GNE by computing 
for $i\in [N]$, $j \in [N]\setminus\{i\}$:
$$\begin{cases}\sup_{p\in[p_{\pi(j)},p_{\pi(j-1)}]} R_{\pi(i)}(p,\bp_{-\pi(i)}) \quad \text{if } j < i \text{ or } j > i+1 \\
\sup_{p\in[p_{\pi(i+1)},p_{\pi(i-1)}]} R_{\pi(i)}(p,\bp_{-\pi(i)}) \quad \text{if } j = i+1.
\end{cases}$$ 
If for $i \in [N]$, any of these revenues is strictly higher than $R_{\pi(i)}(\bp)$, then $\bp$ is not a GNE; else, it is. 
Choose a new, uninspected ordering, mark it as inspected, and proceed to the beginning of Phase 1. If all $N!$ orderings have been traversed, either output all identified LNE/GNE, or return that none exists.

\textbf{Correctness and termination.}
By the unimodality assumption, for each ordering $\pi$ there exists at most one VOP $(\pi,\bp)$. Hence, by \cref{thm:general-LNE}, if a LNE exists, it must have been found in Phase 1. Since a GNE is a LNE, it must also have been found in Phase 1, if it exists. Phase 2 is a simple consistency check to see if the local optimum computed for each seller is locally / globally optimal with respect to $\bp$.

We conclude the proof by giving an upper bound on the number of undercutting best-response computations required for the procedure to terminate. For each ordering (of which there are $N!$), at most $N$ undercutting best-response computations are performed during Phase 1. Verifying whether each of these VOPs is a LNE requires an additional two best-response computations per seller, per ordering. This yields the upper bound on the number of best-response computations to find all LNE, or to find that none exists. 

The second GNE verification phase requires at most $N$ undercutting best-response computations {per seller}, for a given LNE. This then yields the $\textup{O}(N^2)$ best-response computations across all sellers. 

The bound in Part 2 of the statement follows from putting these two together. Namely, each ordering requires at most $N^2 + N$ local best-response computations to check if the LNE is a GNE. Doing this for all $N!$ orderings gives rise to an upper bound of $\textup{O}((N^2+N)N!)$ undercutting best-response computations.
\hfill\Halmos
\end{proof}

\subsection{Gradient dominance: Special cases}\label{ssec:gradient-dominance-discussion} 

In this section we present special cases of the CLC class for which gradient dominance holds, in particular highlighting the connection between gradient dominance and the strength of a seller along a non-price attribute.

{We first illustrate that gradient dominance can arise due to quality differences between sellers.} In particular, we show that gradient dominance weakly holds under Best-I-Can-Afford behavior, allowing us to leverage \Cref{prop:gd-implies-lne-existence-uniqueness} to characterize all non-trivial LNE for this special case. Recall, under BICA behavior, the customer chooses the highest-quality seller whose price is at most her willingness-to-pay. Let $q_i$ denote the quality of seller $i$, for $i \in \cN$.

\begin{proposition}
\label{prop:bica-gd}
Under BICA behavior, for any two sellers $i, j \in \cN$ with qualities $q_i > q_j$, seller $i$ weakly gradient-dominates seller $j$. Moreover, if all sellers' qualities are distinct and customers' willingness-to-pay distribution has a monotone hazard rate, there exists a unique non-trivial LNE in which sellers' prices are strictly ordered according to their respective qualities. 
\end{proposition}

{Relating this back to the market dominance described in \Cref{sec:efficient-computation}, under BICA behavior, for any sellers $i$ and $j$ such that $q_i > q_j$, seller $i$ has no incentive to undercut $j$, since customers who can afford both sellers will always purchase from $i$. Seller $j$'s customer base, on the other hand, is entirely composed of customers for whom $i$ is too expensive, which pushes seller $j$'s revenue-maximizing price downward. Mathematically, this translates to gradient dominance.} We defer the proof of \Cref{prop:bica-gd} to Appendix \ref{apx:bica-gd}.

{We now provide sufficient conditions for gradient dominance to hold in a duopoly.}
Recall, for any customer $c$, $\nopriceconsiderationset$ denotes the consideration set formed {\it without} accounting for {the willingness-to-pay constraint}, i.e., $\nopriceconsiderationset = \{i\in\mathcal{N} : (v^1_i,\ldots, v^K_i)\in  \attributecriteria{c}\}.$
For $i,j\in\mathcal{N}$, $\npaplus{i}{j}$ denotes the probability that a customer strictly prefers $i$ to $j$ due to a non-price attribute ranked higher than price, conditional on the event $\nopriceconsiderationset = \{1,2\}$. 

\begin{proposition}
\label{prop:bica-gd-robustness}
Consider a duopoly in which customers' willingness-to-pay distribution is independent of all other CLC choice primitives, and has a monotone hazard rate. Then, seller $i$ gradient-dominates seller $j$ if and only if the following hold:
\begin{enumerate}
\item\label{cond1} Non-price attribute dominance of seller $i$:
$$\npaplus{i}{j}\Prob{\nopriceconsiderationset = \{1,2\}} + \Prob{\nopriceconsiderationset=\{i\}} > \left(1-\npaplus{i}{j}\right)\Prob{\nopriceconsiderationset = \{1,2\}} +\Prob{\nopriceconsiderationset = \{j\}}$$
\item\label{cond2} Relative competitiveness of seller $j$:
$$\left(1-\npaplus{i}{j}\right)\Prob{\nopriceconsiderationset = \{1,2\}} +\Prob{\nopriceconsiderationset = \{j\}}\geq \Prob{\nopriceconsiderationset=\{i\}}$$
\end{enumerate}
Under these two conditions, there exists a LNE in which $p_i^* > p_j^*$.
\end{proposition}

{We provide some high-level intuition for \Cref{prop:bica-gd-robustness}, deferring its proof to Appendix \ref{apx:bica-gd-robustness}. Condition \ref{cond1} requires the market dominance of seller $i$: if more customers are loyal to $i$, or prefer $i$ to $j$ based on an attribute ranked higher than price, then (i) seller $j$ can capture low-valuation customers who could not afford their preferred seller $i$, and (ii) seller $i$ has a higher incentive to increase its price when it is below its optimum. Condition \ref{cond2} simply guarantees that $j$ has an incentive to undercut $i$ by ensuring that not too many customers are loyal to $i$; this essentially states that the market has to indeed be competitive. 

The following corollary of \Cref{prop:bica-gd-robustness} further illustrates the relationship between market and gradient dominance.
\begin{corollary}\label{cor:bica-gd-robustness}
Consider a duopoly in which customers' willingness-to-pay distribution is independent of all other CLC choice primitives, and has a monotone hazard rate. Suppose moreover that customers always evaluate both sellers, and strictly more than half of customers prefer seller 1 to seller 2 according to an attribute ranked higher than price (i.e., $\Prob{\nopriceconsiderationset = \{1,2\}}  =1$ and $\npaplus{1}{2} > 1/2$). Then, seller 1 gradient-dominates seller 2. In this case, there exists a LNE in which $p_1^* > p_2^*$.
\end{corollary}
}

\subsubsection{Proof of \Cref{prop:bica-gd}}\label{apx:bica-gd}

\begin{proof}{\it Proof.}
We first argue that weak gradient dominance holds. Consider $\bp$ such that $p_i > p_j$. We consider two cases.

\noindent \textbf{Case 1: There exists $k \in \cN \setminus\{i,j\}$ such that $q_k > q_i$ and $p_k \leq p_i$.} This is a trivial price configuration since no customer chooses $i$ (anyone who can afford $p_i$ can also afford $p_k$, and will choose this higher-quality seller). Hence, for any such price vector seller $i$ has a revenue of zero.

\noindent\textbf{Case 2: $p_k > p_i$ for all $k \in \cN$ such that $q_k > q_i$.} We have:
\begin{align*}
&R_i(p_i, p_j, \bp_{-\{i,j\}}) = p_i\cdot \mathbb{P}\left(p_i \leq w_c < \min_{k: q_k > q_i} p_k\right) 
\\\implies &g_i(p_i, p_j, \bp_{-\{i,j\}}) = \mathbb{P}\left(p_i \leq w_c < \min_{k: q_k > q_i} p_k\right)-p_if(p_i).
\end{align*}
Similarly:
{
\begin{align}\label{eq:reuse}
&R_j(p_i, p_j, \bp_{-\{i,j\}}) = p_j\cdot \mathbb{P}\left(p_j \leq w_c < \min_{k:q_k > q_j}p_k\right) \notag  \\
\implies &g_j(p_i, p_j, \bp_{-\{i,j\}}) = \mathbb{P}\left(p_j \leq w_c < \min_{k:q_k > q_j}p_k\right)-p_jf(p_j)\\
\implies &\lim_{p_j\uparrow p_i} g_j(p_i, p_j, \bp_{-\{i,j\}}) = -p_if(p_i),
\end{align}}
{where the final implication follows from the fact that $q_i > q_j$, therefore: \[\lim_{p_j\uparrow p_i}\mathbb{P}\left(p_j \leq w_c < \min_{k:q_k > q_j}p_k\right) \leq \lim_{p_j\uparrow p_i}\mathbb{P}\left(p_j \leq w_c < p_i\right)   = 0.\]
Since $\lim_{p_j\uparrow p_i}\mathbb{P}\left(p_j \leq w_c < \min_{k:q_k > q_j}p_k\right) \geq 0$, the above implies $\lim_{p_j\uparrow p_i}\mathbb{P}\left(p_j \leq w_c < \min_{k:q_k > q_j}p_k\right) = 0$.}
Hence, we obtain:
\begin{align*}
g_i(p_i, p_j, \bp_{-\{i,j\}})-\lim_{p_j\uparrow p_i}g_j(p_i, p_j, \bp_{-\{i,j\}}) = \mathbb{P}\left(p_i \leq w_c < \min_{k: q_k > q_i} p_k\right) > 0. 
\end{align*}

Consider now the ordering $\pi$ that ranks sellers in increasing order of quality, and note that this is the unique ordering that gives rise to a non-trivial price vector. We now argue that $\pi(i)$ attains its unique maximum in $(0,p_{\pi(i-1)})$, for any price vector $\bp$ aligned with $\pi$. By \eqref{eq:reuse}:
\begin{align}
g_{\pi(i)}(p,\bp_{-\pi(i)}) = \bar{F}(p)-\bar{F}(p_{\pi(i-1)})-pf(p).
\end{align}
Moreover, by the first-order condition:
\begin{align}\label{eq:prove-unimod}
p = \frac{\bar{F}(p)-\bar{F}(p_{\pi(i-1)})}{f(p)}.
\end{align}
For any interval over which $f(p)$ is non-decreasing, the right-hand side of \eqref{eq:prove-unimod} is decreasing, since $\bar{F}(p)-\bar{F}(p_{\pi(i-1)})$ is decreasing. For any interval over which $f(p)$ is decreasing, the right-hand side of \eqref{eq:prove-unimod} is  decreasing, since $\frac{\bar{F}(p)}{f(p)}$ is non-increasing by the MHR assumption, and $-\frac{\bar{F}(p_{\pi(i-1)})}{f(p)}$ is decreasing. Hence, the solution to \eqref{eq:prove-unimod} is unique and attained in $(0,p_{\pi(i-1)})$, since any price in this interval yields strictly positive revenue.

Observe that the result does not follow immediately from \Cref{prop:gd-implies-lne-existence-uniqueness}, since unimodality and gradient dominance fail to hold for trivial price configurations in which a lower-quality seller prices above a high-quality seller. However, one can apply the proof of \Cref{prop:gd-implies-lne-existence-uniqueness} to $\pi$, and we obtain the result.
\hfill\Halmos
\end{proof}

\subsubsection{Proof of \Cref{prop:bica-gd-robustness}}\label{apx:bica-gd-robustness}

\begin{proof}{\it Proof.}
For ease of notation, we let $i=1, j=2$. To show that seller 1 gradient-dominates seller 2, it suffices to consider $\bp$ such that $p_1 > p_2$. We partition our analysis on the realization of $\nopriceconsiderationset$, and let $\mbox{$R_k(\bp\mid\nopriceconsiderationset)$}$ denote the expected revenue of seller $k \in \mathcal{N}$, given $\nopriceconsiderationset$.

\noindent\textbf{Case 1:} $\nopriceconsiderationset = \{1\}$. Then,
$R_1(\bp\mid\nopriceconsiderationset) = p_1\bar{F}(p_1)$, $R_2(\bp\mid\nopriceconsiderationset) = 0.$

\noindent\textbf{Case 2:} $\nopriceconsiderationset = \{2\}$. Then, $R_1(\bp\mid\nopriceconsiderationset) = 0,$ $R_2(\bp\mid\nopriceconsiderationset) = p_2\bar{F}(p_2).$

\noindent\textbf{Case 3:} $\nopriceconsiderationset = \{1,2\}$. Let $\equalpref{2}{1}$ be the mass of customers who rank $2$ and $1$ equally on all attributes up until price, given $\nopriceconsiderationset = \{1,2\}$. Since $p_1 > p_2$, the only customers who purchase the item from seller 1 are customers who strictly preferred it based on a non-price attribute ranked higher than price, and for whom $w_c \geq p_1$. Hence: $R_1(\bp\mid\nopriceconsiderationset) = \npaplus{1}{2}p_1\bar{F}(p_1).$

On the other hand, the customers who purchase the item from seller 2 are those who satisfy one of three conditions: (1) $w_c \geq p_2$, and seller 2 was strictly preferred to seller 1 on a non-price attribute ranked higher than price, (2) $w_c \geq p_2$, and seller 2 was ranked identically to seller 1 on all attributes up until price, (3) $p_2 \leq w_c < p_1$, and seller 1 was strictly preferred to seller 2 based on a non-price attribute ranked higher than price. Then:
\begin{align*}
R_2(\bp\mid\nopriceconsiderationset)
&=p_2\left(\left(\bar{F}(p_2)-\bar{F}(p_1)\right)\npaplus{1}{2} + \bar{F}(p_2)\left(\npaplus{2}{1}+\equalpref{2}{1}\right)\right)= p_2\left(\bar{F}(p_2) - \bar{F}(p_1)\npaplus{1}{2}\right),
\end{align*}
where the second equality follows from the fact that $\npaplus{1}{2} + \npaplus{2}{1} + \equalpref{2}{1} = 1$.

Putting these three cases together, we obtain the following:
\begin{align}\label{eq:g1}
R_1(\bp) &= p_1\bar{F}(p_1)\left[\Prob{\nopriceconsiderationset = \{1\}} + \npaplus{1}{2}\Prob{\nopriceconsiderationset = \{1,2\}}\right] \notag \\
\implies g_1(p_1,p_2) &= \left(\bar{F}(p_1)-p_1f(p_1)\right)\left(\Prob{\nopriceconsiderationset = \{1\}} + \npaplus{1}{2}\Prob{\nopriceconsiderationset = \{1,2\}}\right),
\end{align}
and
\begin{align}\label{eq:g2}
R_2(\bp) &= p_2\left[\bar{F}(p_2)\Prob{\nopriceconsiderationset = \{2\}} + \Prob{\nopriceconsiderationset = \{1,2\}}\left(\bar{F}(p_2) - \bar{F}(p_1)\npaplus{1}{2}\right)\right] \notag \\ &= p_2\left[\bar{F}(p_2)\left(\Prob{\nopriceconsiderationset = \{2\}} + \Prob{\nopriceconsiderationset = \{1,2\}}\right) - \bar{F}(p_1)\npaplus{1}{2}\Prob{\nopriceconsiderationset = \{1,2\}}\right] \\
\implies g_2(p_1, p_2) &= \bar{F}(p_2)\left(\Prob{\nopriceconsiderationset = \{2\}} + \Prob{\nopriceconsiderationset = \{1,2\}}\right) - \bar{F}(p_1)\npaplus{1}{2}\Prob{\nopriceconsiderationset = \{1,2\}} \notag \\
&\quad -p_2f(p_2)\left(\Prob{\nopriceconsiderationset = \{2\}} + \Prob{\nopriceconsiderationset = \{1,2\}}\right) \notag \\
&= \left(\bar{F}(p_2)-p_2f(p_2)\right)\left(\Prob{\nopriceconsiderationset = \{2\}} + \Prob{\nopriceconsiderationset = \{1,2\}}\right) - \bar{F}(p_1)\npaplus{1}{2}\Prob{\nopriceconsiderationset = \{1,2\}} \label{eq:mhr2}\\
\implies \lim_{p_2\uparrow p_1}g_2(p_1, p_2) &= \left(\bar{F}(p_1)-p_1f(p_1)\right)\left(\Prob{\nopriceconsiderationset = \{2\}} + \Prob{\nopriceconsiderationset = \{1,2\}}\right) - \bar{F}(p_1)\npaplus{1}{2}\Prob{\nopriceconsiderationset = \{1,2\}} \label{eq:g2}
\end{align}
Then:
\begin{align*}
&g_1(p_1,p_2) - \lim_{p_2\uparrow p_1}g_2(p_1,p_2) \\ &=
\bar{F}(p_1)\npaplus{1}{2}\Prob{\nopriceconsiderationset = \{1,2\}}+ 
\left(\bar{F}(p_1)-p_1f(p_1)\right)\Bigg(\Prob{\nopriceconsiderationset = \{1\}}-\Prob{\nopriceconsiderationset = \{2\}}+ (\npaplus{1}{2}-1)\Prob{\nopriceconsiderationset = \{1,2\}}\Bigg).
\end{align*}
Dividing by $\bar{F}(p_1)$, we then have that, for all $p_1 \in (0,\pmax)$:
\begin{align*}
&g_1(p_1,p_2) - \lim_{p_2\uparrow p_1}g_2(p_1,p_2) > 0 \\\iff &\npaplus{1}{2}\Prob{\nopriceconsiderationset = \{1,2\}} > \left(1-\frac{p_1f(p_1)}{\bar{F}(p_1)}\right)\Bigg(\Prob{\nopriceconsiderationset = \{2\}}-\Prob{\nopriceconsiderationset = \{1\}}+ (1-\npaplus{1}{2})\Prob{\nopriceconsiderationset = \{1,2\}}\Bigg).
\end{align*}

Suppose first that $\Prob{\nopriceconsiderationset = \{2\}}-\Prob{\nopriceconsiderationset = \{1\}}+ (1-\npaplus{1}{2})\Prob{\nopriceconsiderationset = \{1,2\}} \geq 0$. Since $F$ is MHR, $1-\frac{p_1f(p_1)}{\bar{F}(p_1)}$ is non-increasing. Thus, the maximum of the right-hand side is achieved at $p_1 = 0$. For seller 1 to gradient-dominate seller 2, we must then have:
\begin{align*}
\npaplus{1}{2}\Prob{\nopriceconsiderationset = \{1,2\}} > \Prob{\nopriceconsiderationset = \{2\}}-\Prob{\nopriceconsiderationset = \{1\}}+ (1-\npaplus{1}{2})\Prob{\nopriceconsiderationset = \{1,2\}}
\end{align*}

Suppose now that $\Prob{\nopriceconsiderationset = \{2\}}-\Prob{\nopriceconsiderationset = \{1\}}+ (1-\npaplus{1}{2})\Prob{\nopriceconsiderationset = \{1,2\}} < 0$. Then, again by the MHR assumption, the supremum of the right-hand side is attained at $p_1 = \bar{v}$, with $$\lim_{p_1\uparrow\bar{v}} \left(1-\frac{p_1f(p_1)}{\bar{F}(p_1)}\right)\Bigg(\Prob{\nopriceconsiderationset = \{2\}}-\Prob{\nopriceconsiderationset = \{1\}}+ (1-\npaplus{1}{2})\Prob{\nopriceconsiderationset = \{1,2\}}\Bigg) = +\infty.$$
Thus, there exist values of $p_1\in(0,\bar{v})$ for which $g_1(p_1,p_2) \leq \lim_{p_2\uparrow p_1}g_2(p_1,p_2)$, and seller 1 does not gradient-dominate seller 2.

To establish that there exists a LNE $\bp^*$ under which $p_1^* > p_2^*$, it suffices to show that each seller's undercutting revenue curve is unimodal for price vectors $\bp$ such that $p_1 > p_2$. Inspecting \eqref{eq:g1} and \eqref{eq:mhr2}, unimodality follows from the MHR assumption. One can then use the same arguments as those used to establish \Cref{prop:gd-implies-lne-existence-uniqueness} for these price configurations and obtain the result.
\hfill\Halmos
\end{proof}

\subsection{Proof of \Cref{prop:gd-implies-lne-existence-uniqueness}}\label{apx:gd-implies-lne-existence-uniqueness}

We first introduce some useful terminology.  {If seller $i$ gradient-dominates seller $j$ and $p_i > p_j$, we say that sellers $i$ and $j$ are {\it correctly ordered}. Otherwise, $i$ and $j$ are said to be {\it incorrectly ordered}.} The following fact emerges as a consequence of gradient dominance.
\begin{lemma}\label{prop:grad-dom-implication}
Under gradient dominance, there exists an ordering $\pi$ such that, for all $\bp \in \cV^N$ aligned with $\pi$ such that $p_{\pi(N)} > 0$ and all $l \in [N-1]$:
\begin{align}\label{eq:ordered-gradients-1}
g_{\pi(l)}\left(p_{\pi(l)},\bp_{-\pi(l)}^{\pi(1:l-1)}\right) &> \lim_{p\uparrow p_{\pi(l)}} g_{\pi(l+1)}\left(p,\bp_{-\pi(l+1)}^{\pi(1:l)}\right). 
\end{align}
\end{lemma}

\begin{proof}{\it Proof.}
The proof is constructive. Consider any ordering $\pi$ such that \eqref{eq:ordered-gradients-1} does not hold, and let $l^*$ denote the first rank in the ordering such that $\exists \, \bp \in \cV^N$ such that {$g_{\pi(l)}\left(p_{\pi(l)},\bp_{-\pi(l)}^{\pi(1:l-1)}\right) \leq \lim_{p\uparrow p_{\pi(l)}} g_{\pi(l+1)}\left(p,\bp_{-\pi(l+1)}^{\pi(1:l)}\right)$}, with $p_{\pi(N)} > 0$. By gradient dominance, then, $g_{\pi(l+1)}\left(p_{\pi(l+1)},\bp_{-\pi(l+1)}^{\pi(1:l)}\right)>\lim_{p\uparrow p_{\pi(l+1)}} g_{\pi(l)}\left(p,\bp_{-\pi(l)}^{\pi(1:l-1)}\right)$ for all $\bp \in \cV^N$ such that $p_{\pi(N)} > 0$. Define ordering $\pi'$, which flips the two incorrectly ordered sellers, i.e.: 
\begin{align*}
\begin{cases}
\pi'(l) &= \pi(l) \quad  \forall \ l \in [N] \setminus\{l^*,l^*+1\}  \\
\pi'(l^*) &= \pi(l^*+1) \\
\pi'(l^*+1) &= \pi(l^*).
\end{cases}
\end{align*}
Repeat this process with $\pi'$, as long as there exist incorrectly ordered sellers in the newly generated ordering $\pi'$. To see why this process must terminate with all sellers correctly ordered, note that, once two sellers have been swapped, by gradient dominance they will never be swapped again. Thus, the algorithm never revisits the same ordering. Since there are finitely many orderings and, for each ordering, finitely many pairs to potentially swap, this algorithm must terminate in an ordering $\pi$ such that \eqref{eq:ordered-gradients-1} holds.\hfill\Halmos
\end{proof}

We use this fact to establish existence of LNE.
\begin{proof}{\it Proof of \Cref{prop:gd-implies-lne-existence-uniqueness}.}
For $\mathcal{S} \subseteq \cN \setminus \{i\}$, let $p_i^*(\bp_{-i}^\cS) \in \arg\max_{p\in[0, \nearestminprice)} R_i(p,\bp_{-i}^\cS)$ denote the unique optimal price of seller $i$ in the open undercutting region relative to the set of sellers $\cS$. (Uniqueness follows from \Cref{asp:unimodality}.) 

We first argue that $p_i^*(\bp_{-i}^\cS) \neq 0$, if it exists. Suppose for contradiction that $p_i^*(\bp_{-i}^\cS) = 0$. Since $R_i(0,\bp_{-i}^\cS) = 0$, the optimal revenue is $0$, which implies that $R_i(p,\bp_{-i}^\cS)=0$ for all $p\in[0,\nearestminprice]$, contradicting the unimodality assumption. Thus, $p_i^*(\bp_{-i}^\cS) \in (0, \nearestminprice)$, and moreover $g_i(p_i^*(\bp_{-i}^\cS),\bp_{-i}^\cS)=0$.


Now, let $\pi$ be such that \eqref{eq:ordered-gradients-1} holds, and define $p^*_{\pi(1)} = \arg\max_p R_{\pi(1)}(p, \bp^{\phi}_{-\pi(1)})$. Using similar arguments as above, $p^*_{\pi(1)} \in (0,\bar{v})$, and 
 $g_{\pi(1)}(p^*_{\pi(1)},\bp_{-\pi(1)}^{\phi})=0$. By \eqref{eq:ordered-gradients-1}, we then have
$$\lim_{p\uparrow p_{\pi(1)}} g_{\pi(2)}\left(p,\bp_{-\pi(2)}^{\pi(1)}\right) < g_{\pi(1)}(p_{\pi(1)}^*,\bp_{-\pi(1)}^{\phi})  = 0,$$
which implies that  $R_{\pi(2)}(\cdot,\bp_{-\pi(2)}^{\pi(1)})$ attains its unique maximum $p_{\pi(2)}^*$ in $(0,p_{\pi(1)}^*)$, with $g_{\pi(2)}\left(p_{\pi(2)}^*,\bp_{-\pi(2)}^{\pi(1)}\right) = 0$, by \Cref{asp:unimodality}. Iterating this argument for $l = 3,\ldots,N$, we obtain that $\bp^*$ is a VOP with distinct prices, and hence a LNE. 

Note that finding an ordering $\pi$ such that seller $i$ is ranked above seller $j$ if and only if $i$ gradient-dominates $j$ requires $O(N^2)$ comparisons. Given $\pi$, it suffices to find the aligned price vector $\bp^*$, which can be done in $O(N)$ computations. 
\hfill\Halmos
\end{proof}

{
{
\subsection{Non-Equivalence of LNE and VOPs}\label{apx:vop-not-lne}

Given \cref{thm:general-LNE}, a natural question that arises is whether a VOP is necessarily a LNE. By the pseudo-competitive property, if all prices in a VOP are distinct, the corresponding price vector is necessarily a LNE. The issue then arises for VOPs in which two sellers price equally. We next establish that there exist settings for which such VOPs need not be LNE; as a result, the converse of \Cref{thm:general-LNE} does not hold in general.

\begin{proposition}\label{ex:vop-not-lne}
There exist instances of CLC choice for which a VOP is not a LNE.
\end{proposition}

\begin{proof}{\it Proof.}
{Consider two sellers selling an identical {item}, with the only attributes being price and brand. All customers form their consideration set $\mathcal{N}_c$ based only on whether their willingness-to-pay $w_c$ exceeds the item's price. Moreover, customers are partitioned equally into three types:
\begin{itemize}
\item Type 1 customers: rank brand over price and prefer {seller 1 to seller 2};
\item Type 2 customers: rank brand over price and prefer {seller 2 to seller 1};
\item Type 3 customers: rank price over brand and prefer {seller 1 to seller 2}.
\end{itemize}
Finally, suppose the willingness-to-pay $w_c$ of customers is independent of type, and $w_c \sim \textup{Unif}[0,1]$. 

In this case, both {sellers} setting a price of $0.5$ is the only VOP in this game, corresponding to the ordering $\pi = (2,1)$. To see this, consider first the ordering $\pi = (1,2)$ in which seller 1 comes first, assuming seller 2 prices at 0. In this case, the best seller 1 can do is price at 0.5. This is obtained by optimizing only over Type 1 customers who have $w_c\sim\textup{Unif}[0,1]$, since Type 2 and Type 3 customers will both choose seller 2 at a price of 0. For any $p < 0.5$, Type 2 and Type 3 customers choose seller 2; optimizing over these customers over $[0,0.5)$ yields an optimum of $0.5-\epsilon$ for seller 2, for arbitrarily small $\epsilon > 0$. If seller 2 prices at 0.5, on the other hand, it loses all Type 3 customers. Hence, the supremum of seller 2's undercutting revenue is not attained over $[0,0.5]$, and the valid order-price pair corresponding to $\pi = (1,2)$ does not exist.

Consider now the ordering where {seller 2} comes first, assuming {seller 1} prices at $0$, followed by {seller 1} pricing in the undercutting region. In this case, the best {seller 2} can do is to optimize over Type $2$ customers, since Type 1 and 3 customers always choose {seller 1} when it prices at 0. Then, {seller 2} sets a price of $0.5$. For any $p\leq 0.5$, Type 1 and Type 3 customers choose {seller 1}. Optimizing over these customers, {seller 1} has a local optimum of 0.5 over its undercutting region. 

This VOP is not a LNE since {seller 2} stands to gain all Type 3 customers from {seller 1} by marginally undercutting {seller 1} and setting a price of $0.5-\epsilon$ for some small $\epsilon>0$, for the same reason as described above. Since the only VOP is not a LNE, Theorem~\ref{thm:general-LNE} implies that there is no LNE in this game.}\hfill\Halmos
\end{proof}
}
}

\input{final-parts/algorithmic-price-competition-apx.tex}

\input{final-parts/numerical-results-r1}

\input{final-parts/apx-case-study}

%% file: final-parts/algorithmic-price-competition-apx.tex
\section{\Cref{sec:numerical-results}: Additional Results and Omitted Proofs}\label{app:lne-and-ogd}

{\subsection{Specification of simulation framework}\label{apx:simulation-framework}

In \Cref{alg:sim} we provide pseudo-code for the simulation framework described in \Cref{sec:numerical-results}. Our framework assumes sellers implement the classical Kiefer-Wolfowitz algorithm for stochastic function optimization with zero-order feedback \citep{kiefer1952stochastic}.

\begin{algorithm}[t]
{
\begin{algorithmic}
\Require {Customer population size $M$, estimated CLC choice model, time horizon $T$, finite difference width schedule $(c_\tau, \tau \in \mathbb{N})$, stepsize $(a_\tau, \tau \in \mathbb{N})$}
\Ensure {Final prices $\bp_{T}$}
\State Generate an initial set of prices $\bp_0 \in \cV^N$ uniformly at random.
\State Initialize $\tau_i = 0$, for $i \in \cN$.
\For{$t= 0,\ldots,T-1$}
\State Generate batch of $M$ customers. 
\State Choose a seller $i \in \cN$ uniformly at random to update their price.
\State For all sellers $j \neq i$, all $M$ customers observe prices $p_{j,t}$. 
\State Randomly sample $M/2$ customers to observe $p_{i,t}^+ = p_{i,t}+c_{\tau_i}/2$; the remaining $M/2$ customers observe $p_{i,t}^- = p_{i,t}-c_{\tau_i}/2$. 
\State Simulate purchase decisions of all $M$ customers, and estimate average revenue per customer at each of these prices, denoted by $\widehat{R}_i(p_{i,t}^+,\bp_{-i,t})$, $\widehat{R}_i(p_{i,t}^-,\bp_{-i,t})$.
\State Estimate partial gradient of seller $i$ at $p_{i,t}$, i.e., 
\[\hat{g}_{i,t}(p_{i,t},\bp_{-i,t}) = \frac{\widehat{R}_i(p_{i,t}^+,\bp_{-i,t})-\widehat{R}_i(p_{i,t}^-,\bp_{-i,t})}{c_{\tau_i}}.\]
\State Update price of seller $i$:
\[p_{i,t+1} = [p_{i,t} + a_{\tau_i}\hat{g}_{i,t}(p_{i,t},\bp_{-i,t})]_{\cV}.\]
\State Increment update counter of seller $i$: $\tau_i \leftarrow \tau_i + 1$.
\EndFor
\end{algorithmic}
}
	\caption{Equilibrium simulator\label{alg:sim}}
\end{algorithm}

}

\subsection{Formal description of Distributed Gradient Ascent dynamics}\label{apx:alg-ogd}

We present the Distributed Gradient Ascent algorithm in \cref{alg:ogd}.

\begin{algorithm}[t]
{
\begin{algorithmic}
\Require {Convex set $\cV$, time horizon $T$, learning rate schedule $(\eta_t)_{0\leq t\leq T-1}$, initial prices $\bp_0\in\cV^N$}
\Ensure {Final prices $\bp_{T}$}
\For{$t= 0,\ldots,T-1$}
\For{$i\in \cN$}
	\If{$p_{i,t} \neq p_{j,t}$ for all $j \neq i$}
	let $f_i(p_{i,t},\bp_{-i,t}) = g_i(p_{i,t},\bp_{-i,t})$
	\Else 
	\State Let $\cS = \{j \in \cN \setminus \{i\} : p_{i,t} = p_{j,t}\}$, and
    \begin{align}
    f_i(p_{i,t},\bp_{-i,t}) = \begin{cases}
    \lim_{p_i\uparrow p_{i,t}}g_i(p_i ,p_j = p_{j,t} \, \forall \, j \in \cS, \bp_{-\{{\cS},i\},t}) \quad &\text{with probability } 1/2\\
    \lim_{p_i\downarrow p_{i,t}}g_i(p_i ,p_j = p_{j,t} \, \forall \, j \in \cS, \bp_{-\{{\cS},i\},t}) &\text{with probability } 1/2.
    \end{cases}
    \end{align}
	\EndIf
	\State $p_{i,t+1} = \left[p_{i,t} + \eta_{t} {f_i(p_{i,t}, \bp_{-i,t})}\right]_{\cV}.$
\EndFor
\EndFor
\end{algorithmic}
}
	\caption{Distributed gradient ascent (GA)\label{alg:ogd}}
\end{algorithm}

{
Note that, in the case where $p_{i,t} = p_{j,t}$ for some $j \neq i$, by the pseudo-competitive property, the right-derivative $\lim_{p_i\downarrow p_{i,t}}g_i(p_i ,p_j = p_{j,t} \, \forall \, j \in \cS, \bp_{-\{{\cS},i\},t}) = g_i(p_{i,t}, p_j = 0 \, \forall \, j \in \cS, \bp_{-\{{\cS},i\},t})$.
}

\input{final-parts/n-firm-convergence.tex}

%% file: final-parts/n-firm-convergence.tex

{
\subsection{Assumption 3: Sufficient Conditions}\label{apx:asp3-discussion}

We first provide a sufficient condition on sellers' revenue functions for Condition 3 to hold. We defer its proof to Appendix \ref{apx:cond3-strongly-ccv}.

\begin{lemma}\label{lem:strong-ccv}
Suppose $R_i(p,\bp_{-i}^{\cS})$ is strongly concave with respect to $p$, for all $\cS \subseteq \cN \setminus \{i\}$, $\bp_{-i}^{\cS}$. Then, Condition 3 holds.
\end{lemma}

We now leverage \Cref{lem:strong-ccv} to provide sufficient conditions on the CLC choice model for Assumption 3 to hold in a duopoly.

\begin{proposition}\label{prop:cond13-specialcase-2}
Consider a duopoly in which customers' willingness-to-pay distribution is independent of all other CLC choice primitives, and has Lipschitz continuous pdf and cdf. Then, Conditions 1 and 2 of \Cref{asp:ogd-rev-n-firms} hold. Moreover:
\begin{itemize}
\item If $\min_{i}\mathbb{P}(\nopriceconsiderationset = \{i\}) + \mathbb{P}(\nopriceconsiderationset = \{1,2\})\npaplus{i}{-i} > 0$, and there exists $c > 0$ such that $2f(p)+pf'(p) \geq c$ for all $p \in [0,\bar{v}]$, then Condition 3 holds. 
\item If $\mathbb{P}(\nopriceconsiderationset = \{1,2\})\npaplus{-i}{i}  > 0$ for all $i \in \{1,2\}$ and $\bar{F}(p) > 0$ for all $p \in [0,\bar{v})$, then Condition 4 holds.
\end{itemize}
\end{proposition}

In words, the first sufficient condition for Condition 3 to hold is that, for each seller, either a strictly positive fraction of customers is loyal, or prefers them on an attribute ranked higher than price. The second condition is a technical condition on the willingness-to-pay distribution: it holds for uniform distributions, truncated exponential distributions, and distributions with linearly decreasing density. 

The first sufficient condition for Condition 4 to hold states that, for each seller, a strictly positive mass of customers must prefer the competing seller based on an attribute ranked higher than price. This condition implies that, when the competitor is pricing at its optimum in a VOP, the focal seller benefits from undercutting in order to capture customers who cannot afford the competitor. The second sufficient condition is an extremely mild technical condition stating that the willingness-to-pay distribution is supported over the entire interval $[0,\bar{v})$. 

\begin{proof}{\it Proof.}
Fix $i \in \cN$, $\bp_{-i} \in \cV$. Moreover, for ease of notation, let $j = -i$. In order to derive $R_i(p,\bp_{-i})$, we leverage the closed-form expression of $\mathcal{C}_i$ derived in the proof of \Cref{prop:duopoly-asp1}. In particular, recall the following cases:

\noindent\textbf{Case a:} $p > p_j$. Then, by \eqref{eq:choice-prob-duopoly-1}, we have:
\begin{align}\label{eq:partial-duopoly-1}
&R_i(p,\bp_{-i}) = p\bar{F}(p)\left(\mathbb{P}(\nopriceconsiderationset = \{i\}) + \mathbb{P}(\nopriceconsiderationset = \{1,2\})\npaplus{i}{j} \right) \notag \\
\implies &g_i(p,\bp_{-i}) = \left(\mathbb{P}(\nopriceconsiderationset = \{i\}) + \mathbb{P}(\nopriceconsiderationset = \{1,2\})\npaplus{i}{j} \right)\left(\bar{F}(p)-pf(p)\right).
\end{align}

\noindent\textbf{Case b:} $p < p_j$. Then, by \eqref{eq:choice-prob-duopoly-2}, we have:
\begin{align}\label{eq:partial-duopoly-2}
&R_i(p,\bp_{-i}) = p\left(\bar{F}(p)\left(\mathbb{P}\left(\nopriceconsiderationset = \{i\}\right) +  \mathbb{P}\left(\nopriceconsiderationset = \{1,2\}\right)\right) - \mathbb{P}\left(\nopriceconsiderationset = \{1,2\}\right)\bar{F}(p_j)\npaplus{j}{i}\right) \notag \\
\implies &g_i(p,\bp_{-i}) = \bar{F}(p)\left(\mathbb{P}\left(\nopriceconsiderationset = \{i\}\right) +  \mathbb{P}\left(\nopriceconsiderationset = \{1,2\}\right)\right) - \mathbb{P}\left(\nopriceconsiderationset = \{1,2\}\right)\bar{F}(p_j)\npaplus{j}{i} \notag \\
&\hspace{5cm}-pf(p)\left(\mathbb{P}\left(\nopriceconsiderationset = \{i\}\right) +  \mathbb{P}\left(\nopriceconsiderationset = \{1,2\}\right)\right) \notag \\
&= \left(\mathbb{P}\left(\nopriceconsiderationset = \{i\}\right) +  \mathbb{P}\left(\nopriceconsiderationset = \{1,2\}\right)\right)\left(\bar{F}(p)-pf(p)\right)-\mathbb{P}\left(\nopriceconsiderationset = \{1,2\}\right)\bar{F}(p_j)\npaplus{j}{i}.
\end{align}

\noindent\textbf{Case c:} $p = p_j$. As described in the proof of \Cref{prop:duopoly-asp1}, in this case it must be that $p = p_j = 0$. 
In this case, the seller will only ever evaluate the right partial-gradient (since she is at the boundary), which is the expression given in \eqref{eq:partial-duopoly-1}.

In all of the above cases, Conditions 1 and 2 follow from Lipschitz continuity of the pdf and cdf, as well as boundedness of the space of feasible prices.

We now show that Condition 3 holds under the sufficient conditions described in the statement of the proposition. Specifically, we prove that $R_i(p,\bp_{-i}^{\cS})$ is strongly concave with respect to $p$, for all $\cS\subseteq \cN\setminus\{i\}, \bp_{-i}^{\cS}$. We again consider the three cases.

\noindent\textbf{Case a:} $p > p_j$ in $\bp_{-i}^{\cS}$.
\begin{align*}
g_i'(p,\bp_{-i}^{\cS}) &= \left(\mathbb{P}(\nopriceconsiderationset = \{i\}) + \mathbb{P}(\nopriceconsiderationset = \{1,2\})\npaplus{i}{j} \right)\left(-2f(p)-pf'(p)\right) \\ &\leq -c\left(\mathbb{P}(\nopriceconsiderationset = \{i\}) + \mathbb{P}(\nopriceconsiderationset = \{1,2\})\npaplus{i}{j} \right) \leq -K,
\end{align*}
for some $K > 0$.

\noindent\textbf{Case b:} $p < p_j$ in $\bp_{-i}^{\cS}$.
\begin{align*}
g_i'(p,\bp_{-i}^{\cS}) &= \left(\mathbb{P}(\nopriceconsiderationset = \{i\}) + \mathbb{P}(\nopriceconsiderationset = \{1,2\})\right)\left(-2f(p)-pf'(p)\right) \\ &\leq -c\left(\mathbb{P}(\nopriceconsiderationset = \{i\}) + \mathbb{P}(\nopriceconsiderationset = \{1,2\})\right) \leq -K,
\end{align*}
for some $K > 0$.

\noindent\textbf{Case c:} $p=p_j$ in $\bp_{-i}^{\cS}$. As in Case a.

We conclude by proving that Condition 4 holds under the condition described in the statement of the proposition. It suffices to consider the case where $p < p_j$, since Condition 4 describes settings where the seller is undercutting their competitor at a partial VOP. 

Since seller $j$ is pricing at its optimum, we have that $\lim_{p\uparrow p_j}\bar{F}(p)-pf(p) = 0$. Plugging this fact into \eqref{eq:partial-duopoly-2}, we have:
$
\lim_{p\uparrow p_j} g_i(p, p_j) = - \mathbb{P}(\nopriceconsiderationset = \{1,2\})\bar{F}(p_j)\npaplus{j}{i} < 0,$
by assumption. Hence, the supremum of the lower-priced seller is always attained in a partial VOP, and Condition 4 holds. \hfill\Halmos
\end{proof}

Finally, to illustrate the idea that the settings in which Condition 4 fails to hold are relatively extreme, we provide the following instance.
\begin{proposition}\label{prop:no-cond-4}
Consider a setting in which each seller has its own pool of loyal customers. Let $\gamma_i$ denote the mass of loyal customers for seller $i \in \cN$, with $\sum_{i\in\cN}\gamma_i = 1$. Suppose moreover that all customers share the same willingness-to-pay distribution. In this case, Condition 4 fails to hold.
\end{proposition}

\begin{proof}{\it Proof.}
Fix any ordering $\pi$. It suffices to consider the top two sellers in this ordering. 

The first seller in the ordering sets the price $p_{\pi(1)}$ that maximizes $\gamma_{\pi(1)}p\bar{F}(p)$, since $\pi(1)$ can only sell to its loyal customers. Consider now the problem $\sup_{0 \leq p < p_{\pi(1)}} R_{\pi(2)}(p,\bp^{\pi(1)}_{-\pi(2)}) = \sup_{0 \leq p < p_{\pi(1)}} \gamma_{\pi(2)}p\bar{F}(p)$, where the equality follows from the fact that $\pi(2)$ can only sell to its loyal customers. Notice that the supremum is not attained in this latter problem, since  $\arg\max_{0\leq p\leq p_{\pi(1)}} \gamma_{\pi(2)}p\bar{F}(p) = p_{\pi(1)}$. However, this fact also implies that $\lim_{p\uparrow p_{\pi(1)}}g_{\pi(2)}(p,\bp^{\pi(1)}_{-\pi(2)}) = 0$. Therefore, Condition 4 fails to hold.\hfill\Halmos
\end{proof}

Note that in the above instance, Condition 4 fails to hold not because customers are loyal, but because, conditional on the seller to whom customers are loyal, the willingness-to-pay distribution is {\it identical}. It is easy to see, given the proof of \Cref{prop:no-cond-4}, that if the revenue functions associated with each loyal pool of customers have maximizers that differ by any $\epsilon > 0$, then Condition 4 does hold.

\subsubsection{Proof of \Cref{lem:strong-ccv}}\label{apx:cond3-strongly-ccv}

\begin{proof}{\it Proof.}
By definition, $R_i(p,\bp_{-i}^{\cS})$ is strongly concave if and only if there exists $K > 0$ such that $g_i'(p,\bp_{-i}^{\cS}) \leq -K$ for all $p \in [0, \nearestminprice)$.

\noindent\textbf{Case 1:} $p_i^*(\bp_{-i}^\cS) \leq \nearestminprice$.  In this case, we seek to show that
$|g_i(p,\bp_{-i}^{\cS})| \geq K |p-p_i^*(\bp_{-i}^\cS)|.$
By the Mean Value Theorem (MVT), there exists $c \in \left(\min\{p,p_i^*(\bp_{-i}^\cS)\},\max\{p,p_i^*(\bp_{-i}^\cS)\} \right)$ such that:
\begin{align}\label{eq:strong-unimod}
|g_i(p,\bp_{-i}^{\cS})-g_i(p_i^*(\bp_{-i}^\cS),\bp_{-i}^{\cS})| = |g_i'(c,\bp_{-i}^{\cS})|\cdot |p-p_i^*(\bp_{-i}^\cS)| \geq K |p-p_i^*(\bp_{-i}^\cS)|.
\end{align}
Noting that $g_i(p_i^*(\bp_{-i}^\cS),\bp_{-i}^{\cS}) = 0$ proves the claim.

\noindent\textbf{Case 2:} $p_i^*(\bp_{-i}^\cS) > \nearestminprice$. In this case, we seek to show that
$|g_i(p,\bp_{-i}^{\cS})| \geq K |p-\nearestminprice|.$ This however again follows from the MVT and strong concavity, since $p < \nearestminprice < p_i^*(\bp_{-i}^\cS) \implies |p-\nearestminprice| < |p-p_i^*(\bp_{-i}^\cS)|$. Applying this to \eqref{eq:strong-unimod}, we obtain the result. \hfill\Halmos
\end{proof}

}

\subsection{Proof of \cref{thm:ogd-converges-to-LNE-n-firms}} \label{ssec:n-firm-convergence}

We introduce some terminology that will be used throughout the proof. If $\bp$ is aligned with an ordering $\pi$ such that \eqref{eq:ordered-gradients-1} holds, we say that $\bp$ is in a {\it right} configuration. If $\bp$ is {\it strictly} ordered but \eqref{eq:ordered-gradients-1} does not hold, $\bp$ is said to be in a {\it wrong} configuration. {Finally, if there exist $i, j \in \cN$ such that $p_i = p_j$, then $\bp$ is in an {\it equal} configuration.}

\Cref{thm:ogd-converges-to-LNE-n-firms} follows from the following facts:
\begin{enumerate}
\item \Cref{asp:grad-dom} and \Cref{asp:ogd-rev-n-firms} imply that the LNE for which the prices are in a right configuration exists.
\item There exists $t^*_1$ such that, if sellers' prices are in a right configuration at $t^*_1$, then they remain in a right configuration for all $t > t^*_1$ (\cref{cor:right-config-remains-right-config-n-firms}).
\item If there exists $t^*_2$ such that sellers' prices are in the same strictly aligned configuration for all $t > t^*_2$, then the LNE aligned with the corresponding ordering exists, and sellers' prices converge to this LNE (\cref{lem:main-convergence-lem}).
\end{enumerate}

We put these facts together to prove \cref{thm:ogd-converges-to-LNE-n-firms}. 
\begin{proof}{\it Proof of \cref{thm:ogd-converges-to-LNE-n-firms}.} Existence of LNE follows from the fact that strong unimodality implies \Cref{asp:unimodality}. Putting this together with \Cref{asp:grad-dom} and applying \Cref{prop:gd-implies-lne-existence-uniqueness}, we obtain existence.

We now prove convergence to a LNE. Throughout the proof we assume that $t$ is large enough that $p_{i,t+1} = p_{i,t} + \eta_t g_i(p_{i,t},\bp_{-i,t})$ for all $i \in \cN$. This is without loss of generality, since by \cref{prop:never-at-boundary} below, the partial gradients of sellers' revenue functions point away from the boundaries of the decision space. We defer the straightforward proof of this fact to Appendix \ref{apx:never-at-boundary}.

\begin{lemma}\label{prop:never-at-boundary}
Under strong unimodality, the following holds:
\begin{enumerate}
\item if  $\nearestminprice > 0$, then $g_i(0, \bp_{-i}^\cS) > 0$;
\item if $\nearestminprice = \pmax$, then $\lim_{p\uparrow \pmax} g_i(p,\bp_{-i}^\cS) < 0$.
\end{enumerate}
\end{lemma}

 Consider $t > t^*_1$. We consider three cases.

\noindent\textbf{Case 1: Sellers' prices are in a right configuration at time $t$.}  By \cref{cor:right-config-remains-right-config-n-firms}, they remain in the right configuration for all $t' > t$. Then, by \cref{lem:main-convergence-lem}, the LNE for which sellers' prices are in this right configuration exists, and they converge to this LNE.

\noindent\textbf{Case 2: Sellers' prices are in a wrong configuration at time $t$.} In this case, either:
\begin{enumerate}
    \item Sellers' prices remain in this wrong configuration for all $t' > t$. Then, by \cref{lem:main-convergence-lem} the LNE corresponding to this wrong configuration exists, and the sellers' prices converge to this LNE.
    \item There exist $t' > t$ and $i, j \in \cN$ such that $p_{i,t'} < p_{j,t'}$ and $p_{i,t'+1} > p_{j,t'+1}$. In this sub-case, either:
    \begin{enumerate}[(a)]
\item $\bp_{t'+1}$ is in a right configuration, in which case we have established that sellers' prices will remain in this right configuration and converge to the corresponding LNE. 
\item $\bp_{t'+1}$ is in another wrong configuration. Then, 
by \cref{lem:never-cross-back-over}, $p_{i,t''} > p_{j,t''}$ for all $t'' > t'$. As a result, sellers will never re-enter the wrong configuration they just exited.
\end{enumerate}
    \item There exist $t' > t$ and $i,j \in \cN$ such that $p_{i,t'} = p_{j,t'} =: p$. Again, by \cref{lem:correctly-ordered-stay-correctly-ordered} it must be that seller $i$ gradient-dominates seller $j$. Due to the randomized gradient oracle at equality (see \cref{alg:ogd}), the following events occur with strictly positive probability:
    \begin{enumerate}[(i)]
    \item sellers $i$ and $j$ observe different gradients,
    \item $i$ observes $g_i(p_i = p, p_j = 0, \bp_{-\{i,j\},t}^{\cS})$, and $j$ observes $\lim_{q\uparrow p} g_j(p_i = p, p_j = q,\bp_{-\{i,j\},t}^{\cS})$, for some $\cS \subseteq \cN \setminus \{i,j\}$.
    \end{enumerate}
Thus, with probability 1 there exists $t''  > t'$ such that $p_{i,t''} \neq p_{j,t''}$, and one of three facts holds at $t''$:
\begin{enumerate}[(a)]
\item sellers' prices are in a right configuration (Case 1),
\item sellers' prices are in a wrong configuration and remain in this wrong configuration (Case 2.1), or
\item sellers $i$ and $j$ become correctly ordered, and will remain correctly ordered from then on (Case 2.2(b)).
\end{enumerate}
\end{enumerate}

\noindent\textbf{Case 3: Sellers' prices are in an equal configuration at time $t$.} See Case 2.3.

Thus, in the worst case, sellers keep exiting and entering into successive wrong configurations. However, since there are finitely many configurations and sellers never re-enter a wrong configuration they have exited, it must be that sellers' prices eventually enter a right configuration and converge to the corresponding LNE. 
\end{proof}

We now show each of these key component results.

\begin{lemma}\label{lem:correctly-ordered-stay-correctly-ordered}
There exists large enough $t^*_1 \in \mathbb{N}$ such that, for all $i, j \in \cN$, if sellers $i$ and $j$ are {correctly ordered} at time $t^*_1$, they remain correctly ordered for all $t \geq t^*_1$.
\end{lemma}

\begin{proof}{\it Proof.}
Without loss of generality, suppose $i$ gradient-dominates $j$. Let $c > 0$ be defined as follows:
$$ c = \min \left\{g_i(p_i = p, p_j = 0, \bp_{-\{i,j\}}^{\cS}) - \lim_{q\uparrow p} g_j(p_i = p, p_j = q,\bp_{-\{i,j\}}^{\cS}) \, \bigg{|} \, \cS \subseteq \cN \setminus \{i,j\}, \bp \in \cV^{N} \right\}.$$
Fix $0 < \delta < c/(2L)$. We first show that there exists large enough $t$ such that, if $|p_{i,t} - p_{j,t}| \geq \delta$, then $p_{i,t} > p_{j,t} \implies p_{i,t+1} > p_{j,t+1}$. We have:
\begin{align}\label{eq:i-j-diff}
    p_{i,t+1}-p_{j,t+1} &= \left(p_{i,t} + \eta_t g_{i}\left(p_{i,t},\bp_{-i,t}\right)\right) - \left(p_{j,t} + \eta_t g_{j}(p_{j,t},\bp_{-j,t})\right) \notag \\
    &= p_{i,t}-p_{j,t} + \eta_t \left(g_{i}\left(p_{i,t},\bp_{-i,t}\right) - g_{j}(p_{j,t},\bp_{-j,t})\right) \\
    &\geq \delta - 2G\eta_t > 0 \notag
\end{align}
for large enough $t$, since $\eta_t$ is decreasing and $\lim_{t\to\infty} \eta_t = 0$ by square-summability of $\eta_t$.  
 
We now argue that, if $|p_{i,t} - p_{j,t}| < \delta$, then $p_{i,t} > p_{j,t} \implies p_{i,t+1} > p_{j,t+1}$. By smoothness, for all $\epsilon > 0$ such that $(p_{i,t}-\epsilon,\bp_{-j,t})$ is aligned with $\bp_t$,
$g_{j}(p_{j,t},\bp_{-j,t}) \leq g_{j}(p_{i,t}-\epsilon,\bp_{-j,t}) + L|p_{j,t} - (p_{i,t}-\epsilon)|.$
Using $|p_{i,t}-p_{j,t}| < \delta$ and taking the limit as $\epsilon$ approaches 0, we obtain
$g_{j}(p_{j,t},\bp_{-j,t}) \leq  \lim_{p\uparrow p_{i,t}}g_{j}(p,\bp_{-j,t}) + L\delta$.
Plugging this into \eqref{eq:i-j-diff}, we have:
\begin{align*}
p_{i,t+1} - p_{j,t+1} &\geq (p_{i,t} - p_{j,t}) + \eta_t \left(g_{i}\left(p_{i,t},\bp_{-i,t}\right) - \lim_{p\uparrow p_{i,t}} g_{j}(p_{i,t},\bp_{-j,t}) - L\delta \right) \\
&\geq (p_{i,t} - p_{j,t}) + \eta_t \cdot c/2 > p_{i,t}-p_{j,t} > 0.
\end{align*}
%
Taking $t^*_1 = \min\{t \, \mid \, \delta - 2G\eta_t > 0 \, \forall \, i, j \in \cN\}$, we obtain the claim.\hfill\Halmos
\end{proof}

The following corollaries emerge as a direct consequence of \cref{lem:correctly-ordered-stay-correctly-ordered}. 
\begin{corollary}\label{cor:right-config-remains-right-config-n-firms}
If $\bp_t$ is in a right configuration for some $t \geq t_1^*$, then, for all $t' > t$, $\bp_{t'}$ remains in this configuration.
\end{corollary}

\begin{corollary}\label{lem:never-cross-back-over}
Fix $t \geq t_1^*$. If sellers $i$ and $j$ were incorrectly ordered at time $t$ but became correctly ordered at time $t+1$, then they remain correctly ordered for all $t' \geq t+1$.
\end{corollary}

\input{final-parts/main-convergence-lem.tex}

%% file: final-parts/main-convergence-lem.tex
We now present our key lemma.
\begin{lemma}\label{lem:main-convergence-lem}
Suppose there exists  $t^* \in \mathbb{N}$ such that $\bp_t$ is aligned with the same ordering $\pi$ for all $t>t^*$. Then, there exists a valid order-price pair to which $\pi$ belongs, and $\bp_t$ converges to the corresponding LNE.
\end{lemma}

\begin{proof}{\it Proof.}
\input{final-parts/quasi-concavity-cvg-proof.tex}

Thus, for $t > \max\{t_{\epsilon}, t'_{\epsilon}, t_{\delta}\}$, as long as $d_t \geq \epsilon$, $d_t$ strictly decreases by $\gamma_t$. Since $\sum_t \gamma_t = \infty$, there exists a point past which $d_t < \epsilon$. When this occurs, we have $d_{t+1} < 2\epsilon$. Taking $\epsilon$ to be arbitrarily small, we obtain convergence of $p_{\pi(l+1),t}$ to $\widetilde{p}_{\pi(l+1)}$. This completes the proof of the inductive step. 

Rolling this argument out, we must have that $l^* = N$, and (i) the valid order-price pair to which $\pi$ belongs exists, and (ii) $\lim_{t\to\infty} p_{\pi(l),t} = \widetilde{p}_{\pi(l)}$ for all $l \in [N]$.\hfill\Halmos
\end{proof}


We conclude with a proof of \cref{lem:bounded-away}.

\begin{proof}{\it Proof of \cref{lem:bounded-away}.}
The proof is similar to that of \cref{lem:correctly-ordered-stay-correctly-ordered}. By assumption, {$p_{\pi(l+1),t} < p_{\pi(l),t} < \widetilde{p}_{\pi(l)}+\delta$} for all $t > t_{\delta}$. We consider two cases.

\medskip 
\noindent\textbf{Case 1: $p_{\pi(l+1),t} \leq \widetilde{p}_{\pi(l)} - \delta$.}  As in the proof of \cref{lem:correctly-ordered-stay-correctly-ordered}, we have:
\begin{align*}
p_{\pi(l+1),t+1} = p_{\pi(l+1),t} + \eta_t g_{\pi(l+1)}(p_{\pi(l+1),t},\bp_{-\pi(l+1),t}^{\pi(1:l)}) \leq p_{\pi(l+1),t} + \eta_t G \stackrel{(a)}{<} p_{\pi(l+1),t} + \delta < \widetilde{p}_{\pi(l)},
\end{align*}
where $(a)$ follows from the fact that $\eta_t < \delta/G$.

\medskip

\noindent\textbf{Case 2: $|{p}_{\pi(l+1),t}-\widetilde{p}_{\pi(l)}| < \delta$.} We show that, for all $t > t_{\delta}$, $p_{\pi(l+1),t+1} \leq p_{\pi(l+1),t} - \gamma_t$, where $\gamma_t > 0$ for all $t$ and $\sum_t \gamma_t = \infty$. If this holds, there exists $t' > t_{\delta}$ such that $p_{\pi(l+1),t'} < \widetilde{p}_{\pi(l)} - \delta$, and we are back in Case 1. We have:
\begin{align}\label{eq:aux-lem-case-2}
p_{\pi(l+1),t+1} &= p_{\pi(l+1),t}  + \eta_t g_{\pi(l+1)}(p_{\pi(l+1),t},\bp_{-\pi(l+1),t}^{\pi(1:l)}) \notag \\
&\leq p_{\pi(l+1),t} + \eta_t \cdot \left(\lim_{p\uparrow \widetilde{p}_{\pi(l)}} g_{\pi(l+1)}(p, \widetilde{\bp}_{-\pi(l+1)}^{\pi(1:l)}) + (l+1)L\delta \right) \notag \\
&\leq p_{\pi(l+1),t} + \eta_t \cdot \frac12 \cdot \lim_{p\uparrow \widetilde{p}_{\pi(l)}} g_{\pi(l+1)}(p, \widetilde{\bp}_{-\pi(l+1)}^{\pi(1:l)})
\end{align}
where \eqref{eq:aux-lem-case-2} follows from the definition of $\delta$. Taking $\gamma_t = - \frac12\eta_t \lim_{p\uparrow \widetilde{p}_{\pi(l)}} g_{\pi(l+1)}(p, \widetilde{\bp}_{-\pi(l+1)}^{\pi(1:l)})$, we obtain $\gamma_t > 0$ since $\lim_{p\uparrow \widetilde{p}_{\pi(l)}} g_{\pi(l+1)}(p, \widetilde{\bp}_{-\pi(l+1)}^{\pi(1:l)}) < 0$ by the assumption that the supremum is attained. Moreover, $\sum_t \gamma_t = \infty$, by non-summability of $\eta_t$.\hfill\Halmos
\end{proof}

\subsubsection{Proof of \Cref{prop:never-at-boundary}}\label{apx:never-at-boundary}

\begin{proof}{\it Proof.} Strong unimodality implies that $p_i^*(\bp_{-i}^\cS)$ is unique. We first argue that $p_i^*(\bp_{-i}^\cS) \neq 0$. Suppose for contradiction that $p_i^*(\bp_{-i}^\cS) = 0$. This implies that the optimal revenue is $0$. Since the revenue is non-negative, this means that $i$ obtains zero revenue for all $p \in [0, \nearestminprice)$, contradicting strong unimodality. Thus $p_i^*(\bp_{-i}^\cS) > 0$, implying $\min\{p_i^*(\bp_{-i}^\cS), \nearestminprice\}>0$. The strong unimodality condition then implies that $g_i(0, \bp_{-i}^\cS) > 0$, thus showing the first claim.

We next turn to the second claim. We have already argued that $p_i^*(\bp_{-i}^S) \neq 0$. We now argue that $\bar{v} \notin \argmax_{p\in [0,\bar{v}]}R_i(p,\bp_{-i}^S)$ if $\nearestminprice = \pmax$. Suppose again for contradiction that $\bar{v} \in \argmax_{p\in [0,\bar{v}]}R_i(p,\bp_{-i}^S)$. Similarly, since $R_i(\bar{v},\bp_{-i}^S) = 0$, this implies that seller $i$ obtains zero revenue for all $p \in [0, \nearestminprice)$, contradicting the strong unimodality of $R_i(p,\bp_{-i}^\phi)$ over $[0,\bar{v})$. Thus, $p_i^*(\bp_{-i}^\phi) \in (0, \bar{v})$. Since $g_i(p_i^*(\bp_{-i}^\phi), \bp_{-i}^\phi) = 0$, strong unimodality again implies the second result. \hfill\Halmos
\end{proof}

%% file: final-parts/quasi-concavity-cvg-proof.tex
Let $l^* = \sup\left\{l \, \big{|} \, \exists \text{ a valid order-price pair } (\widetilde{\pi}, \widetilde{\bp}) \text{ such that } \widetilde{\pi}(l') = \pi(l') \, \forall \, l' \leq l\right\}.$
If $l^*$ fails to exist, we let $l^* = 1$. We prove that $\lim\limits_{t\to\infty} p_{\pi(l'),t} = \widetilde{p}_{\pi(l')}$ for all $l' \leq l^*$ by induction.

\medskip

\noindent\textbf{Base case: $l' = 1$.} This simply follows from the fact that for all $t > t^*$, Seller $\pi(1)$'s revenue is independent of $\bp_{-\pi(1),t}$, and thus seller $\pi(1)$ is simply running gradient ascent {over $R_{\pi(1)}(p, \bp^{\phi}_{-\pi(1)})$, which under \Cref{asp:ogd-rev-n-firms} converges to the unique maximizer $p_{\pi(1)}^*$}~\citep{broadie2011general}. 

\medskip

\noindent\textbf{Inductive step.} Suppose $\lim\limits_{t\to\infty} p_{\pi(l'),t} = \widetilde{p}_{\pi(l')}$ for all $l' \leq l$, with $l < l^*$.  We first argue that $\sup_{p < p_{\pi(l)}} R_{\pi(l+1)}(p, \widetilde{\bp}_{-\pi(l+1)}^{\pi(1:l)})$ is attained. {By continuity of the partial gradients, Condition 4 of \Cref{asp:ogd-rev-n-firms} implies there exists $c > 0$ such that $\lim_{p \uparrow \widetilde{p}_{\pi(l)}}g_{\pi(l+1)}\left(p, \widetilde{\bp}_{-\pi(l+1)}^{\pi(1:l)}\right) \geq c$.
}

Let $\delta >0$ be such that 
\begin{align}\label{eq:delta-def}
\delta < \frac12 \min \left\{\frac{c}{2\left(1+clL/K+lL\right)},\min_{l'<N}\left\{ |\widetilde{p}_{\pi(l')} - \widetilde{p}_{\pi(l'+1)}|\right\}\right\}.
\end{align}
Define $t_{\delta} > t^*$ such that $|\widetilde{p}_{\pi(l')} - p_{\pi(l'),t}| < \delta$ for all $t > t_{\delta}$, $l' \leq l$. 
Consider $t \geq t_\delta$, and suppose $p_{\pi(l+1),t} < \widetilde{p}_{\pi(l)}-\epsilon$, for $\epsilon \geq \frac{c+lL\delta}{K}$. Since $\bp_t$ and $\widetilde{\bp}$ are both aligned with $\pi$ for all $l' \leq l$ and $t \geq t_\delta$, we have:
{
\begin{align*}
g_{\pi(l+1)}(p_{\pi(l+1),t}, \bp_{-\pi(l+1),t}^{\pi(1:l)}) \geq g_{\pi(l+1)}(p_{\pi(l+1),t},\widetilde{\bp}_{-\pi(l+1)}^{\pi(1:l)})-l\cdot L\cdot\delta
&\stackrel{(a)}\geq K |p_{\pi(l+1),t}-\widetilde{p}_{\pi(l)}| - l \cdot L \cdot \delta \\
&\geq K\epsilon - l\cdot L \cdot \delta 
\geq c,
\end{align*}
}
where $(a)$ follows from strong unimodality.

Consider now the case where $p_{\pi(l+1),t} \geq \widetilde{p}_{\pi(l)}-\epsilon$, where $\epsilon < \frac{clL\delta}{K}$. Since $p_{\pi(l+1),t} < p_{\pi(l),t} < \widetilde{p}_{\pi(l)} + \delta$, we have $|p_{\pi(l+1),t}-\widetilde{p}_{\pi(l)}| \leq \max\{\delta,\frac{clL\delta}{K}\}$. Then, by smoothness and continuity:
\begin{align*}
g_{\pi(l+1)}(p_{\pi(l+1),t}, \bp_{-\pi(l+1),t}^{\pi(1:l)}) \geq \lim_{p\uparrow \widetilde{p}_{\pi(l)}} g_{\pi(l+1)}(p,\widetilde{\bp}_{-\pi(l+1)}^{\pi(1:l)})-\delta\left(\max\{1,\frac{clL}{K}\}+l\cdot L\right) \geq c/2.
\end{align*}

By non-summability of $(\eta_t)$, then, there exists $t' > t_\delta$ such that $p_{\pi(l+1),t'} > \widetilde{p}_{\pi(l)} + \delta > p_{\pi(l),t'}$, a contradiction.

Having established that $\sup_{p < \widetilde{p}_{\pi(l)}} R_{\pi(l+1)}(p, \widetilde{\bp}_{-\pi(l+1)}^{\pi(1:l)})$ is attained, we now prove $\lim\limits_{t\to\infty} p_{\pi(l+1),t} = \widetilde{p}_{\pi(l+1)}$. Let $d_t = (\widetilde{p}_{\pi(l+1)} - p_{\pi(l+1),t})^2$. We have:
\begin{align*}
d_{t+1} &= d_t - 2(\widetilde{p}_{\pi(l+1)} - p_{\pi(l+1),t})(p_{\pi(l+1),t+1}-p_{\pi(l+1),t}) + (p_{\pi(l+1),t+1}-p_{\pi(l+1),t})^2
\end{align*}

Plugging in the fact that $p_{\pi(l+1),t+1} = p_{\pi(l+1),t} + \eta_t g_{\pi(l+1)}(p_{\pi(l+1),t},\bp_{-\pi(l+1),t})$, we obtain:
\begin{align}\label{eq:induc-n-firms}
d_{t+1}&= d_t- 2\eta_tg_{\pi(l+1)}(p_{\pi(l+1),t},\bp_{-\pi(l+1),t})(\widetilde{p}_{\pi(l+1)} - p_{\pi(l+1),t}) + \eta_t^2 g_{\pi(l+1)}(p_{\pi(l+1),t},\bp_{-\pi(l+1),t})^2 \notag \\
&\leq d_t- 2\eta_tg_{\pi(l+1)}(p_{\pi(l+1),t},\bp_{-\pi(l+1),t})(\widetilde{p}_{\pi(l+1)} - p_{\pi(l+1),t}) + \eta_t^2 G^2.
\end{align}

Fix $\epsilon > 0$, and suppose $d_t \geq \epsilon$. We show that, for large enough $t$, $d_{t+1} < d_t - \gamma_t$, with $\gamma_t > 0$ and $\lim_{T\to\infty}\sum_{t=1}^T \gamma_t = \infty$. 

By the inductive hypothesis, $\lim\limits_{t\to\infty} p_{\pi(l'),t} = \widetilde{p}_{{\pi}(l')}$ for all $l' \leq l$. 
{Let $\delta > 0$ be such that:
\begin{align*}
\delta < \frac12 \min\left\{\frac{K\epsilon}{lL\pmax} \, , \, \min_{l'< N}|\widetilde{p}_{\pi(l')}-\widetilde{p}_{\pi(l'+1)}|,-\frac{1}{(l+1)L}\lim_{p\uparrow \widetilde{p}_{\pi(l)}} g_{\pi(l+1)}(p, \widetilde{\bp}_{-\pi(l+1)}^{\pi(1:l)})\right\}.
\end{align*}
}
({The fact that $\sup_{p < \widetilde{p}_{\pi(l)}} R_{\pi(l+1)}(p,\widetilde{\bp}_{-\pi(l+1)}^{\pi(1:l)})$ is attained implies $\lim_{p\uparrow \widetilde{p}_{\pi(l)}} g_{\pi(l+1)}(p, \widetilde{\bp}_{-\pi(l+1)}^{\pi(1:l)}) = -c$, for some $c > 0$, and thus $\delta$ necessarily exists.})

Let $t_\delta > t^*$ be such that, for all $t > t_\delta$:
(i) $|p_{\pi(l'),t} - \widetilde{p}_{{\pi}(l')}| < \delta$ for all $l' \leq l$, and
(ii) $\eta_t < \delta / G$. ({Such a $t$ exists by the summability assumptions.})
{

\cref{lem:bounded-away} establishes that there exists $t' > t_\delta$ such that, for all $t > t'$, $p_{\pi(l+1),t} < \widetilde{p}_{\pi(l)}$. We formally state the lemma below, and defer its proof to the end of this section.

\begin{lemma}\label{lem:bounded-away}
Consider an ordering $\pi$ and rank $l \in [N]$, and the price vector $\widetilde{\bp}$ defined as follows: 
\begin{align*}
\begin{cases}
\widetilde{p}_{\pi(1)} &:= \arg\max_{p} R_{\pi(1)}(p, \bp^{\phi}_{-\pi(1)}) \\
\widetilde{p}_{\pi(l')} &:= \arg\max_{p < \widetilde{p}_{\pi(l'-1)}}R_{\pi(l')}(p, \widetilde{\bp}_{-\pi(l')}^{\pi(1:l'-1)}) \quad \textrm{for } l' \in \{2,\ldots,l+1\}.
\end{cases}
\end{align*}
(That is, assume each of these local suprema is attained.)
Let $\delta > 0$ be such that
\begin{align*}
\delta < \frac12 \min\left\{ \min_{l'< l}|\widetilde{p}_{\pi(l')}-\widetilde{p}_{\pi(l'+1)}|, -\frac{1}{(l+1)L}\cdot\lim_{p\uparrow \widetilde{p}_{\pi(l)}} g_{\pi(l+1)}(p, \widetilde{\bp}_{-\pi(l+1)}^{\pi(1:l)}) \right\}.
\end{align*}
Suppose moreover that there exists $t_{\delta} \in \mathbb{N}$ such that, for all $t> t_{\delta}$,
(i) $\bp_t$ is aligned with $\pi$,
(ii) $|p_{\pi(l'),t}-\widetilde{p}_{\pi(l')}| < \delta$ for $l' \leq l$, and
(iii) $\eta_t < \delta/G$.
Then, there exists $t' > t_\delta$ such that $p_{\pi(l+1),t} < \widetilde{p}_{\pi(l)}$ for all $t > t'$.
\end{lemma}

Thus, by strong unimodality:
\begin{align*}
g_{\pi(l+1)}(p_{\pi(l+1),t},\widetilde{\bp}_{-\pi(l+1)}^{\pi(1:l)})(\widetilde{p}_{\pi(l+1)}-p_{\pi(l+1),t}) \geq K(\widetilde{p}_{\pi(l+1)}-p_{\pi(l+1),t})^2 \geq K\epsilon.
\end{align*}
}
 Re-arranging, we obtain:
\begin{align}\label{eq:gd-cvg-n-firms}
-g_{\pi(l+1)}(p_{\pi(l+1),t},\widetilde{\bp}_{-\pi(l+1)}^{\pi(1:l)})(\widetilde{p}_{\pi(l+1)} - p_{\pi(l+1),t}) \leq -K\epsilon.
\end{align}

Since $\bp_t$ and $\widetilde{\bp}_t$ are aligned with the same ordering for $l' \leq l$, and $p_{\pi(l+1), t} < \widetilde{p}_{\pi(l)}$, by smoothness we have that, for $t > t_\delta$:
\begin{align*}
g_{\pi(l+1)}(p_{\pi(l+1),t},\widetilde{\bp}_{-\pi(l+1)}^{\pi(1:l)}) - l L\delta \leq g_{\pi(l+1)}(p_{\pi(l+1),t},{\bp}_{-\pi(l+1),t}^{\pi(1:l)}) \leq g_{\pi(l+1)}(p_{\pi(l+1),t},\widetilde{\bp}_{-\pi(l+1)}^{\pi(1:l)}) + l L\delta.
\end{align*}

Plugging this into~\eqref{eq:gd-cvg-n-firms}:
\begin{align*}
-g_{\pi(l+1)}(p_{\pi(l+1),t},{\bp}_{-\pi(l+1),t}^{\pi(1:l)})(\widetilde{p}_{\pi(l+1)} - p_{\pi(l+1),t}) \leq -K\epsilon + l L\delta \pmax \leq -K\epsilon/2.
\end{align*}
Thus, we obtain $d_{t+1} \leq d_t - \eta_tK\epsilon + \eta_t^2 G^2$.

Let $\gamma_t = \eta_tK\epsilon - \eta_t^2G^2$. For large enough $t$, $\gamma_t > 0$, and moreover $\sum_t \gamma_t = \infty$ by the summability assumptions on $(\eta_t)$. Let $t_{\epsilon} = \min \left\{t \, \mid \, \gamma_t > 0, t > t'\right\}$.

We complete the proof by showing that there exists $t'_{\epsilon}$ such that for all $t > t'_{\epsilon}$, if $d_t < \epsilon$, then $d_{t+1} < 2\epsilon$. We use the facts that $|g_{\pi(l+1)}(p_{\pi(l+1),t},\bp_{-\pi(l+1),t})| \leq G$ and $|\widetilde{p}_{\pi(l+1)} - p_{\pi(l+1),t}| < \sqrt{\epsilon}$ and plug into~\eqref{eq:induc-n-firms}, to obtain
$d_{t+1} < \epsilon +2\eta_tG\sqrt{\epsilon} + \eta_t^2 G^2 < 2\epsilon$
for large enough $t$, since $\eta_t$ is decreasing and $\lim\limits_{t\to\infty} \eta_t = 0$ by the square-summability assumption.

%% file: final-parts/numerical-results-r1.tex
{
\subsection{Empirical analysis of gradient-driven price-update dynamics}\label{sec:empirics}
Motivated by the fact that price and rating are two of the main attributes used for purchase decisions on e-commerce platforms, we consider price competition amongst $N\in \{2,\cdots, 5\}$ sellers who differ on these two attributes. Given $N$, we assume that the quality of seller $i \in [N]$ is $q_i = i/N$. 

{\bf Choice model.} We assume that customers are either loyal or non-loyal. Loyal customers purchase from the seller to whom they are loyal as long as its price is at most their willingness-to-pay. Non-loyal customers are one of two types: quality-sensitive (QS) or price-sensitive (PS). QS and PS customers both form their consideration set based on a maximum willingness-to-pay and a minimum rating floor. 
However, they differ in the lexicographic choice phase: QS customers choose the highest-rated item in their consideration set, breaking ties based on price; PS customers, on the other hand, choose the cheapest item in their consideration set, breaking ties based on rating. QS and PS customers respectively generalize the BICA and CICT customers presented in \cref{sec:preliminaries}.

We consider three settings for the distribution over these model primitives:
\begin{itemize}
\item {\bf Setting A:} Customers are loyal, PS, or QS with an equal probability of $1/3$. A loyal customer is loyal to seller $i$ with probability $1/N$. The willingness-to-pay and rating floor for each customer (where applicable) are drawn independently from a $\textrm{Beta}(\alpha,\beta)$ distribution. 

\item {\bf Setting B:} A customer is loyal with probability $1/3$ and non-loyal with probability $2/3$. A loyal customer is loyal to seller $i$ with probability $2i/(N(N+1))$, capturing the fact that higher-rated sellers have more loyal customers. The willingness-to-pay of a loyal customer is drawn from a $\textrm{Beta}(\alpha,\beta)$ distribution. Conditional on being non-loyal, the willingness-to-pay, rating floor, and probability of the customer being QS are all equal, and drawn from a $\textrm{Beta}(\alpha,\beta)$ distribution. This captures the reality that willingness-to-pay and quality sensitivity for non-loyal customers may be correlated, {i.e., a higher willingness-to-pay may indicate a higher quality floor and a higher propensity to be quality-sensitive.}
\item {\bf Setting C:} Same as Setting B, except that there are no loyal customers: all customers are non-loyal with probability $1$, with the same parameters. 
\end{itemize}

{\bf Price dynamics.} Sellers deploy distributed gradient-driven pricing schemes over a fixed time horizon $T = 10^4\cdot N$. At the beginning of the horizon, we generate an initial set of prices independently and uniformly in $[0,1]$. To model the reality that sellers make price updates asynchronously, we choose a seller uniformly at random in each period to price experiment and subsequently make a price update. We assume that sellers implement the classical Kiefer-Wolfowitz algorithm for function optimization with zero-order feedback~\citep{kiefer1952stochastic}. Specifically, the chosen seller in each timestep observes a batch of $10^3$ customers, and randomly sets one of two prices for each customer in the batch: one above and one below its current incumbent price. The seller then estimates the gradient given the average revenue at each price and updates its incumbent price by taking a gradient-scaled step. Letting $c_\tau$ and $a_\tau$ respectively denote the finite difference width for gradient estimation and stepsize after $\tau$ updates, we let $c_\tau = 1/\log \tau$ and $a_\tau = 1/\tau$.

Observe that if the price trajectories converge to distinct prices under such a gradient-based pricing scheme, then the limiting price vector must be a LNE.  We say that a price trajectory has converged if (i) for each seller, the difference between the maximum and minimum prices set in the last $10^3$ updates is within $1\%$ of the empirical average of the past $10^3$ prices, and (ii) the ordering of all prices remains unchanged in the last $10^3$ periods. If only (ii) is satisfied after $T$ periods, we say that the trajectory is {\it order-convergent}. When the price trajectory is only order-convergent, we expect prices to eventually converge due to the pseudo-competitive property, by the same arguments as those provided in the proof sketch of \Cref{thm:ogd-converges-to-LNE-n-firms}.

\begin{table}
  \centering
  \resizebox{\textwidth}{!}{%
  \begin{tabular}{cccc}
    \begin{subtable}{.25\linewidth}
      \centering
      \caption{Setting A}
      \begin{tabular}{cccc}
        \toprule
        $(\alpha, \beta)$ & $N$ & C & OC \\
        \midrule
        (1, 1) & 2 & 100 & 100 \\
        (1, 1) & 3 & 97 & 100 \\
        (1, 1) & 4 & 97 & 100 \\
        (1, 1) & 5 & 94 & 99 \\
        \addlinespace
        (0.5, 0.5) & 2 & 100 & 100 \\
        (0.5, 0.5) & 3 & 98 & 100 \\
        (0.5, 0.5) & 4 & 93 & 100 \\
        (0.5, 0.5) & 5 & 93 & 100 \\
        \addlinespace
        (2, 2) & 2 & 100 & 100 \\
        (2, 2) & 3 & 98 & 100 \\
        (2, 2) & 4 & 91 & 100 \\
        (2, 2) & 5 & 91 & 100 \\
        \addlinespace
        (2, 4) & 2 & 99 & 100 \\
        (2, 4) & 3 & 95 & 99 \\
        (2, 4) & 4 & 91 & 100 \\
        (2, 4) & 5 & 80 & 100 \\
        \addlinespace
        (4, 2) & 2 & 98 & 100 \\
        (4, 2) & 3 & 90 & 100 \\
        (4, 2) & 4 & 84 & 100 \\
        (4, 2) & 5 & 72 & 100 \\
        \bottomrule
      \end{tabular}
    \end{subtable} &
    \begin{subtable}{.25\linewidth}
      \centering
      \caption{Setting B}
      \begin{tabular}{cccc}
        \toprule
        $(\alpha, \beta)$ & $N$ & C & OC \\
        \midrule
        (1, 1) & 2 & 99 & 100 \\
        (1, 1) & 3 & 96 & 100 \\
        (1, 1) & 4 & 98 & 100 \\
        (1, 1) & 5 & 97 & 100 \\
        \addlinespace
        (0.5, 0.5) & 2 & 100 & 100 \\
        (0.5, 0.5) & 3 & 99 & 99 \\
        (0.5, 0.5) & 4 & 98 & 100 \\
        (0.5, 0.5) & 5 & 95 & 99 \\
        \addlinespace
        (2, 2) & 2 & 100 & 100 \\
        (2, 2) & 3 & 99 & 100 \\
        (2, 2) & 4 & 98 & 100 \\
        (2, 2) & 5 & 95 & 99 \\
        \addlinespace
        (2, 4) & 2 & 100 & 100 \\
        (2, 4) & 3 & 100 & 100 \\
        (2, 4) & 4 & 99 & 99 \\
        (2, 4) & 5 & 92 & 99 \\
        \addlinespace
        (4, 2) & 2 & 100 & 100 \\
        (4, 2) & 3 & 97 & 100 \\
        (4, 2) & 4 & 95 & 99 \\
        (4, 2) & 5 & 89 & 99 \\
        \bottomrule
      \end{tabular}
    \end{subtable} &
    \begin{subtable}{.25\linewidth}
      \centering
      \caption{Setting C}
      \begin{tabular}{cccc}
        \toprule
        $(\alpha, \beta)$ & $N$ & C & OC \\
        \midrule
        (1, 1) & 2 & 99 & 100 \\
        (1, 1) & 3 & 98 & 100 \\
        (1, 1) & 4 & 98 & 100 \\
        (1, 1) & 5 & 96 & 99 \\
        \addlinespace
        (0.5, 0.5) & 2 & 98 & 100 \\
        (0.5, 0.5) & 3 & 100 & 100 \\
        (0.5, 0.5) & 4 & 96 & 100 \\
        (0.5, 0.5) & 5 & 96 & 99 \\
        \addlinespace
        (2, 2) & 2 & 100 & 100 \\
        (2, 2) & 3 & 97 & 100 \\
        (2, 2) & 4 & 93 & 100 \\
        (2, 2) & 5 & 97 & 100 \\
        \addlinespace
        (2, 4) & 2 & 100 & 100 \\
        (2, 4) & 3 & 97 & 100 \\
        (2, 4) & 4 & 96 & 99 \\
        (2, 4) & 5 & 93 & 99 \\
        \addlinespace
        (4, 2) & 2 & 99 & 100 \\
        (4, 2) & 3 & 96 & 100 \\
        (4, 2) & 4 & 91 & 100 \\
        (4, 2) & 5 & 83 & 100 \\
        \bottomrule
      \end{tabular}
    \end{subtable} &
    \begin{subtable}{.25\linewidth}
      \centering
      \caption{Setting D}
      \begin{tabular}{ccccc}
        \toprule
        $(\alpha, \beta)$ & $N$ & C & OC &E \\
        \midrule
        (1, 1) & 2 & 99 & 100 & 97 \\
        (1, 1) & 3 & 97 & 100 & 96\\
        (1, 1) & 4 & 99 & 100 & 97\\
        (1, 1) & 5 & 99 & 100 & 95\\
        \addlinespace
        (0.5, 0.5) & 2 & 100 & 100 & 91\\
        (0.5, 0.5) & 3 & 98 & 100 &96\\
        (0.5, 0.5) & 4 & 98 & 100 &95\\
        (0.5, 0.5) & 5 & 96 & 100 &93\\
        \addlinespace
        (2, 2) & 2 & 100 & 100 & 93\\
        (2, 2) & 3 & 97 & 99 &93\\
        (2, 2) & 4 & 99 & 100 &90\\
        (2, 2) & 5 & 95 & 100 &94\\
        \addlinespace
        (2, 4) & 2 & 100 & 100 &95\\
        (2, 4) & 3 & 99 & 100 &95\\
        (2, 4) & 4 & 95 & 99 &93\\
        (2, 4) & 5 & 91 & 99 &87\\
        \addlinespace
        (4, 2) & 2 & 99 & 100 &96\\
        (4, 2) & 3 & 96 & 99 &93\\
        (4, 2) & 4 & 88 & 99 &85\\
        (4, 2) & 5 & 85 & 100 &78\\
        \bottomrule
      \end{tabular}
    \end{subtable}%
  \end{tabular}%
  } 
  \vspace{0.1in}
  \caption{Number of convergent (C) and order-convergent (OC) trajectories out of 100 random initializations in each setting}\label{table:convergence}
\end{table}

{\bf Results.} 
Table~\ref{table:convergence} shows the number of replications out of $100$ random draws of initial prices that result in price and order convergence for Settings A-C. (Setting D shows a final set of robustness results, discussed at the end of the section.)  We observe that the price dynamics converge in the majority of settings, with at least 72 convergent trajectories out of 100 in the worst case over all parameters in Setting A, 89 in Setting B, and 83 in Setting C. If we consider order-convergent trajectories, the numbers further improve, with at least 99 order-convergent trajectories in the worst case across all settings. Overall, our results suggest that LNE are frequently learnable via distributed gradient-driven pricing, under CLC models of practical interest.

We include sample non-order-convergent price trajectories in Figure~\ref{fig:non-conv}. Two of the trajectories show sellers entering undercutting price wars as their prices approach each other (sellers 2-3 in Figure \ref{fig:non-conva}, and sellers 1-2 in Figure \ref{fig:non-convc}). These undercutting price wars are due to the existence of price-sensitive customers who choose the cheapest item in their consideration set. {Indeed, if all customers are price-sensitive, the only equilibrium is one in which all sellers price at zero.} Again, these results underscore the idea that convergence arises in part due to the existence of quality-sensitive and loyal customers, who counteract the incentive to enter a price war.

\begin{figure}[ht]
    \centering
    \begin{subfigure}{0.3\textwidth}
        \includegraphics[width=\linewidth]{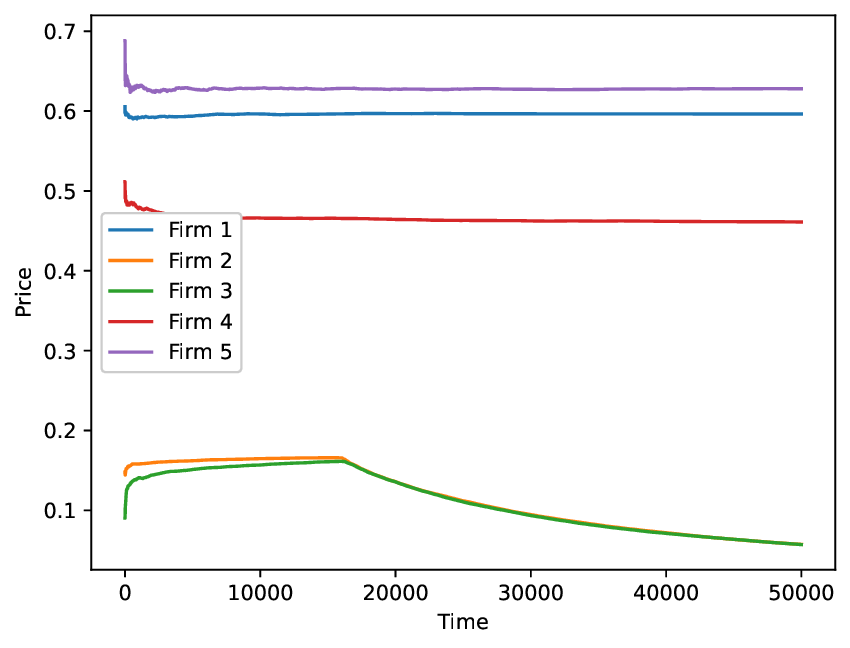}
        \caption{Setting A: $\alpha = 0.5,$ $\beta =0.5$}\label{fig:non-conva}
    \end{subfigure}
    \hfill
    \begin{subfigure}{0.3\textwidth}
        \includegraphics[width=\linewidth]{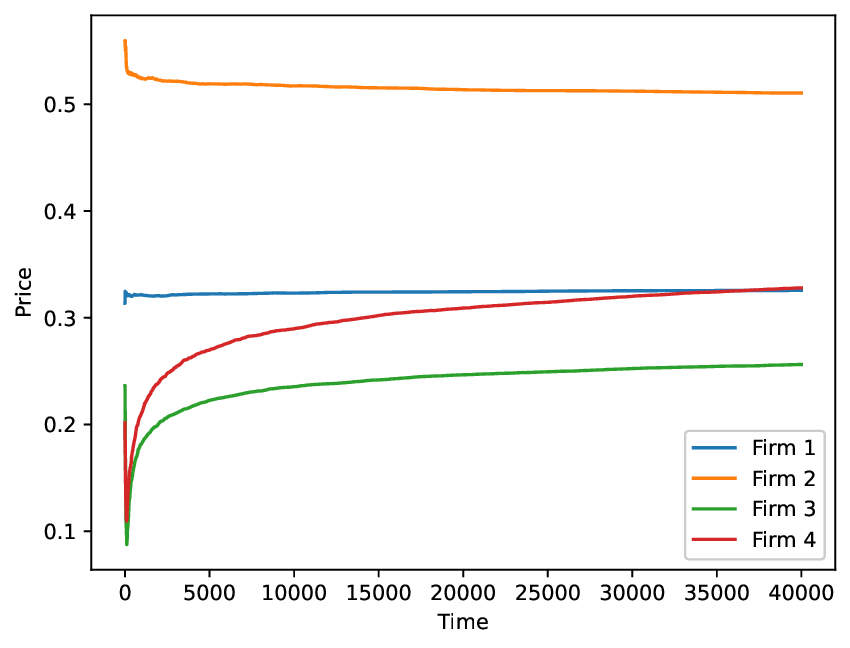}
        \caption{Setting B: $\alpha = 1,$ $\beta =1$}\label{fig:non-convb}
    \end{subfigure}
        \hfill
    \begin{subfigure}{0.3\textwidth}
        \includegraphics[width=\linewidth]{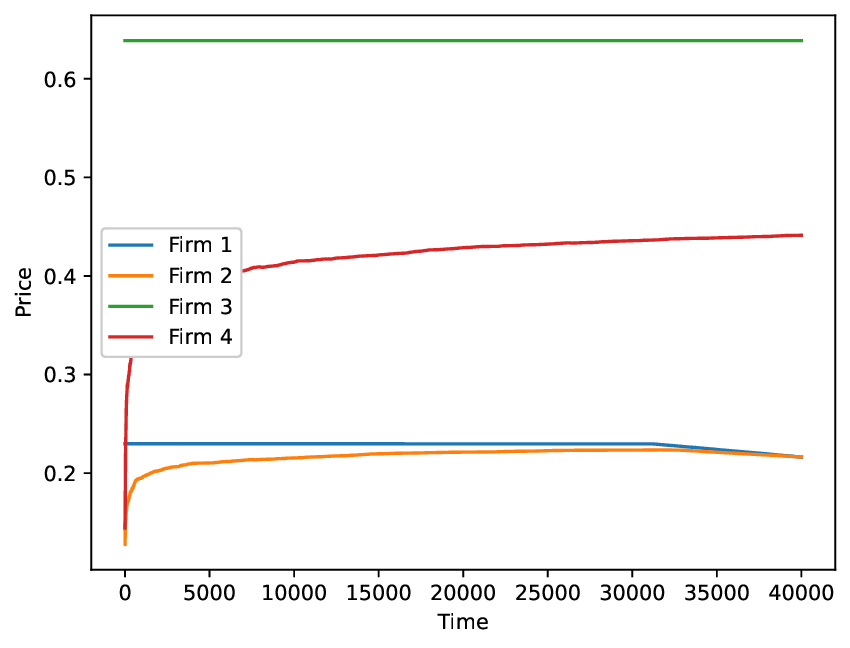}
        \caption{Setting C: $\alpha = 1,$ $\beta =1$}\label{fig:non-convc}
    \end{subfigure}
    \caption{\centering Sample non-convergent trajectories in Settings A-C}\label{fig:non-conv}
    \vspace{-0.2in}
\end{figure}

\subsubsection*{Robustness to non-CLC choice.} We conclude our numerical experiments by investigating the robustness of these observations to non-CLC customers. Specifically, we modify Setting C such that 10\% of customers behave according to the Multinomial Logit (MNL) choice model, with the remainder behaving as in Setting C. For an MNL customer, the utility of choosing a seller $i$ with rating $q$ and price $p$ is assumed to be $q-p+\epsilon$, where $\epsilon$ is an independent noise term with a Gumbel distribution with location parameter $0$ and scale parameter $1$. The customer chooses the seller with the highest utility. We refer to this setting as Setting D. We choose the same price initializations for each of the 100 trajectories for both Settings C and D, and study the impact of MNL customers both on convergence and on the limit points themselves. 

Convergence results for Setting D are shown in Table~\ref{table:convergence}. We obtain consistent convergence to a LNE (in at least 85 random draws, in the worst case), as well as order-convergence (in at least 99 draws, in the worst case). The column labeled ``E" for Setting D shows the number of convergent trajectories such that (i) the corresponding trajectories in Setting C converge as well (we call these {\it co-convergent} trajectories), and (ii) the convergence point is within a tolerance of $10\%$ of the limit point under Setting C; this tolerance aligns with the fact that $10\%$ of the population is of the MNL type. We observe that trajectories under Settings C and D are almost always co-convergent, and that these convergence points are frequently close, with both conditions holding in at least 78 random draws, in the worst case. Overall, these observations suggest that our equilibrium insights are robust to the limited inclusion of well-studied utility-maximizing behaviors.
}

%% file: final-parts/apx-case-study.tex
{
\section{Section \ref{sec:case-study} Omitted Proofs}

\subsection{Proof of \Cref{prop:special-case-LNE-q1}}\label{apx:special-case-LNE-q1}
\begin{proof}{\it Proof.}
The coarse LNE follows from the fact that the coarse setting is equivalent to a Bertrand game of price competition.

We now consider the granular LNE. By \Cref{thm:general-LNE}, it suffices to study two possible VOP, corresponding to each seller being first in the ordering. Note that any price configuration for which $p_2 \geq p_1$ results in seller 2 making zero revenue. Therefore, this cannot be a non-trivial LNE. Therefore, it suffices to consider the case where $p_1 > p_2$. In this case, we have $R_1(p_1,p_2) =\lambda p_1 (1-p_1),$
whose maximum is attained at $p_1^* = 1/2$. Moreover:
\begin{align*}
&R_2(p_1^*,p_2) = p_2\left(\lambda(p_1^*-p_2)+(1-\lambda)(1-p_2)\right)
\implies g_2(p_1^*, p_2) = \lambda(p_1^*-p_2)+(1-\lambda)(1-p_2)-p_2.
\end{align*}
Solving the first-order condition, we obtain $p_2^* = \frac12-\frac\lambda4 < p_1^*$. \hfill\Halmos
\end{proof}

\subsection{Proof of \Cref{prop:metrics-special-case-q1}}\label{apx:metrics-special-case-q1}

\begin{proof}{\it Proof.}
Revenue dominance of the granular LNE trivially follows from the fact that both sellers make zero revenue in the coarse LNE, as compared to strictly positive revenue in the granular LNE. Similarly, demand dominance of the coarse LNE trivially follows from the fact that all sellers' prices have increased. \hfill\Halmos
\end{proof}

\subsection{Proof of \Cref{prop:ranking}}\label{apx:ranking}

\begin{proof}{\it Proof.}
The first part of the proposition follows immediately from \Cref{prop:special-case-LNE-q1}.

For the second part of the proposition, fix $\lambda' \in (0,1).$ We have $R_1(p_1^*,p_2^*) = \lambda' \cdot p_1^* (1-p_1^*) = \frac{\lambda'}{4}$, and 
\begin{align*}
R_2(p_1^*,p_2^*) = p_2^*\left(\lambda'(p_1^*-p_2^*)+(1-\lambda')(1-p_2^*)\right)&=\left(\frac12-\frac{\lambda'}{4}\right)\left(\lambda'\cdot\frac{\lambda'}{4}+(1-\lambda')\left(\frac12+\frac{\lambda'}{4}\right)\right)  \\ &= \left(\frac12-\frac{\lambda'}{4}\right)^2.
\end{align*}
Summing these two, we obtain $\Pi^g = \frac{\lambda'}{4} + \left(\frac12-\frac{\lambda'}{4}\right)^2$, which is increasing in $\lambda'$. Therefore, the Ratings-First default ranking, which is such that $\lambda' = \lambda+\delta$, yields higher total revenue than the Price-First default ranking, which is such that $\lambda' = \lambda-\delta$.

We now analyze total demand on the platform. Since all customers for whom $w_c \geq p_2^*$ are served, we have
$\mathcal{D}^g = 1-p_2^* = \frac12 + \frac{\lambda'}{4}$, which is similarly increasing in $\lambda'$. Hence, the Ratings-First default ranking dominates the Price-First default ranking for this metric as well.\hfill\Halmos
\end{proof}

\subsection{Proof of \Cref{prop:special-case-LNE-q2}}\label{apx:special-case-LNE-q2}

\begin{proof}{\it Proof.}
We first derive the Ratings-First LNE. Again, by \Cref{thm:general-LNE}, it suffices to consider the two VOPs corresponding to each seller pricing first in the ordering. 

Suppose $p_1 > p_2$. In this case, seller 1 only sells to ratings-first customers with $w_c \geq p_1$. Then:
$R_1(p_1, p_2) = \alpha_1 p_1(1-p_1),$
whose maximum is achieved at $p_1^* = 1/2$. Seller 2 sells to ratings-first customers for whom $w_c \in [p_2, p_1)$, and certification-first and price-first customers for whom $w_c \geq p_2$. Therefore:
\begin{align}\label{eq:useful-for-not-equal}
&R_2(p_1^*, p_2) = p_2\left(\alpha_1\left(\frac12-p_2\right)+(1-\alpha_1)(1-p_2)\right) 
\implies g_2(p_1^*, p_2) = \alpha_1\left(\frac12-p_2\right)+(1-\alpha_1)(1-p_2)-p_2.
\end{align}
Solving the first-order condition, we obtain $p_2^* = \frac12-\frac{\alpha_1}{4} < p_1^*$ for all $\alpha_1 \in (0,1)$.

The Certification-First LNE follows by symmetry, since seller 2 only sells to certification-first customers for whom $w_c \geq p_2$, and seller 1 sells to certification-first customers for whom $w_c \in [p_1,p_2)$ and ratings-first and price-first customers for whom $w_c \geq p_1$.\hfill\Halmos
\end{proof}

\subsection{Proof of \Cref{prop:metrics-special-case-q2}}\label{apx:metrics-special-case-q2}

\begin{proof}{\it Proof.}
We analyze the two LNE separately.

\noindent \textbf{Ratings-First LNE: $p_1^* = \frac12$, $p_2^* = \frac12-\frac{\alpha_1}{4}$.} In this case, ratings-first customers choose seller 1 if $w_c \geq p_1^*$, and seller 2 if $w_c \in [p_2^*, p_1^*)$; certification-first and price-first customers choose seller 2 if $w_c \geq p_2^*$.

We first analyze the fraction of demand served in this LNE. By the above reasoning, we have $\mathcal{D}^{rf} = \alpha_1(1-p_2^*) + (\alpha_2 + \gamma)(1-p_2^*) = 1-p_2^* = \frac12+\frac{\alpha_1}{4}$, which is increasing in $\alpha_1$. Since $\alpha_1 \leq \lambda$, and $\mathcal{D}^g = \frac12+\frac\lambda4$, we have that $\mathcal{D}^{rf} \leq \mathcal{D}^g$.

We now analyze seller revenues under this LNE. Noting that this LNE reduces to the setting with $\alpha_1$ mass of ratings-first customers and $1-\alpha_1$ mass of price-first customers, we have: $\mathcal{R}_1^{rf} = \frac{\alpha_1}{4}$, $\mathcal{R}_2^{rf} = \left(\frac12-\frac{\alpha_1}{4}\right)^2$. Then, the results follow from the fact that $\mathcal{R}_1^{rf}$ is increasing in $\alpha_1$ and $\mathcal{R}_2^{rf}$ is decreasing in $\alpha_1$, with $\alpha_1\leq \lambda$.

We conclude by proving the bounds on $\rho^{rf}$. Using the closed-form expressions for $\mathcal{R}_1$ and $\mathcal{R}_2$ in each equilibrium, we have $\Pi^{{rf}} = \frac{\alpha_1^2}{16} + \frac14$, and $\Pi^{{g}} = \frac{\lambda^2}{16} + \frac14$. Hence, $\rho^{rf} \leq 1$ follows from the fact that $\alpha_1 \leq \lambda$. To prove the lower bounds, notice that, in the worst case, $\alpha_1 = 0$ and $\lambda = 1$, in which case $\rho^{rf} = \frac{1/4}{5/16} = \frac45$. Additionally, $\rho^{rf} \geq \left(\frac{\alpha_1}{\lambda}\right)^2$ follows from algebra, and the fact that $\alpha_1 \leq \lambda$.

\noindent \textbf{Certification-First LNE: $p_2^* = \frac12$, $p_1^* = \frac12-\frac{\alpha_2}{4}$.} In this case, ratings-first and price-first customers choose seller 1 if $w_c \geq p_1^*$; certification-first customers choose seller 2 if $w_c \geq p_2^*$, and seller 1 if $w_c \in [p_1^*, p_2^*)$.

We first study the fraction of demand served in this LNE. The arguments are symmetric to those of the Ratings-First LNE. Namely, we have: $\mathcal{D}^{cf} = 1-p_1^* = \frac12 + \frac{\alpha_2}{4}$, which is increasing in $\alpha_2$. Therefore, $\mathcal{D}^{cf} \geq \mathcal{D}^g \iff \alpha_2 \geq \lambda$.

For seller revenues, by symmetry, $\mathcal{R}_2^{cf} = \frac{\alpha_2}{4}$, $\mathcal{R}_1^{cf} = \left(\frac12-\frac{\alpha_2}{4}\right)^2$, $\Pi^{cf} =\frac{\alpha_2}{4} + \left(\frac12-\frac{\alpha_2}{4}\right)^2$. Therefore:
\begin{align*}
\mathcal{R}_1^{cf} \geq \mathcal{R}_1^g \iff \left(\frac12-\frac{\alpha_2}{4}\right)^2 \geq \frac\lambda4 \iff \alpha_2 \leq 2-2\sqrt{\lambda},
\end{align*}
and 
\begin{align*}
\mathcal{R}_2^{cf} \geq \mathcal{R}_2^g \iff \frac{\alpha_2}{4} \geq \left(\frac12-\frac\lambda4\right)^2 \iff \alpha_2 \geq 4\left(\frac12-\frac\lambda4\right)^2.
\end{align*}
Finally, $\Pi^{cf} = \frac{\alpha_2}{4} + \left(\frac12-\frac{\alpha_2}{4}\right)^2$. Using the fact that $\Pi^{cf}$ is increasing in $\alpha_2$, and $\Pi^g = \frac{\lambda}{4} + \left(\frac12-\frac{\lambda}{4}\right)^2$, we obtain $\Pi^{cf} \geq \Pi^g \iff \alpha_2 \geq \lambda$. 

If $\alpha_2 \leq \lambda$, the same arguments as those used for $\rho^{rf}$ apply to $\rho^{cf}$. If $\alpha_2 > \lambda$, on the other hand, the best case is when $\alpha_2 = 1$ and $\lambda = 0$, which gives us $\rho^{cf} \leq \frac54$. Moreover, simple algebra yields $\rho^{cf} \leq \left(\frac{\alpha_2}{\lambda}\right)^2$.\hfill\Halmos
\end{proof}

}